\documentclass[11pt]{amsart}
\usepackage[english]{babel}
\usepackage{xcolor}
\usepackage[a4paper, top=0.7in, bottom=1.5in, left=1.2in, right=1.2in]{geometry}
\usepackage{enumerate}
\usepackage{graphicx, psfrag}
\usepackage{subcaption}
\expandafter\def\csname ver@subfig.sty\endcsname{}
\usepackage{amssymb,dsfont,amsthm,amsmath}
\usepackage[numbers]{natbib}
\newcounter{thm}
\newcounter{ex}
\newcounter{re}
\newtheorem{Theorem}[thm]{Theorem}
\newtheorem{Example}[ex]{Example}
\newtheorem{Remark}[re]{Remark}
\newtheorem{Corollary}[thm]{Corollary}
\newtheorem{Definition}[thm]{Definition}

\newcommand{\set}[1]{\left\{#1\right\}}
\newcommand{\trr}{\triangleright}
\newcommand{\brr}{\blacktriangleright}
\newcommand{\rrt}{\triangleleft\,} 
\newcommand{\abs}[1]{\left\vert#1\right\vert}
\newcommand{\ass}{\stackrel{\textup{\tiny def}}{=}}
\newcommand{\brak}[1]{\left|#1\right)}
\newcommand{\dep}{\delta}

\allowdisplaybreaks

\begin{document}

\title[Computing with coloured tangles]{Computing with coloured tangles}

\author{Avishy Y. Carmi and Daniel Moskovich}

\address{Faculty of Engineering Sciences \& \\ Center for Quantum Information and Technology \\ Ben-Gurion University of the Negev, Beer-Sheva 8410501, Israel}


\thanks{The authors thank M.~Buliga and L.H.~Kauffman for useful discussions that inspired the present note.}%

\begin{abstract}
We suggest a diagrammatic model of computation based on an axiom of distributivity. A diagram of a decorated coloured tangle, similar to those that appear in low dimensional topology, plays the role of a circuit diagram. Equivalent diagrams represent bisimilar computations. We prove that our model of computation is Turing complete, and with bounded resources that it can decide any language in complexity class $\mathrm{IP}$, sometimes with better performance parameters than corresponding classical protocols.
\end{abstract}
\maketitle

\section{Introduction}

The present research represents a step in a programme whose goal is to study topological aspects of information and computation. Time is a metric notion, and so any such topological aspects would presumably have no internal notion of time. We consider a notion of computation which is independent of time and which is natively formulated in terms of information. We construct a diagrammatic calculus whose elements we call \emph{tangle machines}. These were first defined in \citep{CarmiMoskovich:15}. A key feature of tangle machines is that they come equipped with a natural notion of \emph{equivalence} which originates in the beautiful diagrammatic algebra of low dimensional topology. Tangle machines serve in this paper as abstract flowcharts of information in computation. We prove that the computational paradigm that we propose contains Turing Machines and Interactive Proofs (thus it does not `lose anything'), and it also contains additional models (\textit{e.g.} Section~\ref{SS:TangIP}).

The world we observe around us evolves along a time axis, so a tangle machine could not be used as a blueprint for a classical computer. Time is a more nebulous concept in the quantum realm, however, and it might be that tangle machine constructions are relevant for adiabatic quantum computations or in other quantum contexts \citep{CarmiMoskovich:15}. In particular, they naturally incorporate the axiom of uniform no-cloning (Remark~\ref{R:NoCloning} in Section~\ref{SS:Quagma}). The main relevance of our work would probably be to isolate and access natively topological aspects of classical and quantum computation. We also speculate that tangle machine computations can emerge physically via dynamical processes on a tangle machine, given a set of input colours, and that perhaps something like this actually occurs in nature. After all, natural computers are not Turing machines.

How might tangle machines manifest themselves in nature? The authors make the following speculation. Evolutionary biology provides an analogue to tangle machines in the notions of phenotype versus genotype \citep{Churchill:74,Johannsen:11}. The external characteristics of an organism such as its appearance, physiology, morphology, as well as its behaviours are collectively known as a \emph{phenotype}. The \emph{genotype} on the other hand refers to the inherent and immutable information encoded in the genome. Two phenotypes may look entirely different but may nevertheless share the same genotype. Could information about an organism be encoded as a tangle machine, where equivalent machines represent different phenotypes which share the same genotype? Might the process of evolution of an organism be described by a series of basic transformations akin to the Reidemeister moves exerted by the environment on the organism, which change its phenotype while preserving its genotype, along with occasional `violent' local moves on a current configuration which change its genotype? Might tangle machines describe a way in which nature process its information primitives--- its organisms?

There are two obvious advantages to a topological model of computation. The first is that it is very flexible by construction. Bisimilar computations (Definition~\ref{D:Bisimulation}) are represented by topologically equivalent objects, which are related in a simple way (Section~\ref{sec:lowdim}). The second, which we do not discuss in this paper, is that we have a notion of \emph{topological invariants} which are characteristic quantities which are intrinsic to a bisimilarity class of computations.


In the introduction we briefly introduce tangle machines in Section~\ref{SS:WhatTM} after which we state our main results in Section~\ref{SS:Results} and give scientific context in Section~\ref{SS:Background}.

\subsection{What is a tangle machine computation?}\label{SS:WhatTM}

A tangle machine is built up out of \emph{registers} each of which may hold an element of a set $Q$. The set $Q$ comes equipped with a set $B$ of binary operations representing basic computations. For $\trr\in B$, we read $x\trr y$ as `the result of running the programme $\trr y$ on input data $x$'. An alternative evocative image is that $x\trr y$ is a `fusion of information $x$ with information $y$ using algorithm $\trr$'. Our binary operations satisfy the following axioms which equip $(Q,B)$ with what is called a \emph{quandle} structure (see Section~\ref{SS:QuagmaDefn} for this and extensions):

\begin{description}
\item[Idempotence] $x\trr x=x$ for all $x\in Q$ and for all $\trr\in B$. Thus, $x$ cannot concoct any new information from itself.
\item[Reversibility] The map $\trr y\colon\, Q\to Q$, which maps each colour $x\in Q$ to a corresponding colour $x\trr y\in Q$, is a bijection for all $(y,\trr)\in (Q,B)$. In particular, if $x\trr y = z\trr y$ for some $x,y,z\in Q$ and for some $\trr\in B$, then $x=z$. Thus, the input $x$ of a computation may uniquely be reconstructed from the output $x\trr y$ together with the programme $\trr y$.
\item[Distributivity] For all $x,y,z\in Q$ and for all $\trr,\brr \in B$:
\begin{equation} (x\trr y)\brr z= (x\brr z)\trr (y\brr z)\enspace .\end{equation}
\noindent This is the main property. It says that carrying out a computation $\brr z$ on an output $x\trr y$ gives the same result as carrying out that computation both on the input $x$ and also on the state $y$, and then combining these as $(x\brr z)\trr (y\brr z)$. In the context of information, this is a \emph{No Double Counting} property \citep{CarmiMoskovich:14}.
\end{description}

Later in this paper, in Remark~\ref{R:NoCloning} in Section~\ref{SS:Quagma}, we show that uniform \emph{no-cloning} and \emph{no-deleting}, which are fundamental properties of quantum information, follow from Reversibility and Distributivity for an appropriate colouring by a generalization of a quandle called a \emph{quagma}. This observation argues for Reversibility and Distributivity as being nature's most fundamental information symmetries.

The basic unit of computation in a tangle machine is an \emph{interaction} representing multiple inputs $x_1,\ldots, x_k$ independently fed into a programme $\trr y$ as depicted in Figure~\ref{F:kebaba}. These are concatenated (perhaps also with wyes) to form a \emph{tangle machine}--- see Section~\ref{S:TM}.

\begin{figure}
\centering
\begin{minipage}{1.6in}
\psfrag{a}[c]{$x_1$}\psfrag{x}[l]{$x_2$}\psfrag{y}[c]{$x_k$}\psfrag{b}[c]{$y$}\psfrag{c}[c]{$\trr$}\psfrag{d}[c]{$z_1$}\psfrag{u}[c]{$z_2$}\psfrag{v}[c]{$z_k$}
\includegraphics[width=1.2in]{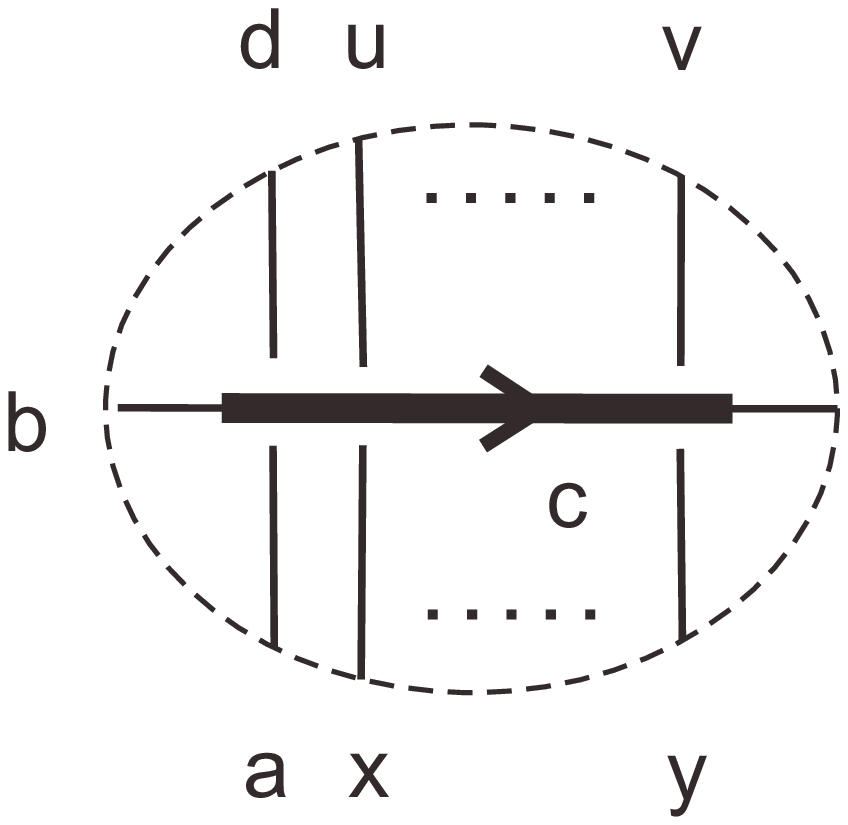}
\end{minipage}
\caption{\label{F:kebaba}An interaction in a machine, where $z_i\ass x_i \trr y$ for $i=1,2,\ldots,k$.}
\end{figure}

A \emph{tangle machine computation} begins with an initialization of a specified set of \emph{input registers} to chosen colours in $Q$. A disjoint set of \emph{output registers} is chosen. If the colours of the input registers uniquely determine colours for the output registers, then the colours of the output registers are the result of the computation. Otherwise the computation cannot take place. See Definition~\ref{D:TMComputation} and Figure~\ref{F:SampleComputation}). This provides a model of computation.

\begin{Remark}
More generally we may allow the result of the computation to be the (possibly empty) set of all possible colours out output registers. This level of generality is not required in this paper.
\end{Remark}

\begin{figure}
\centering
\begin{minipage}{0.4\textwidth}
\psfrag{a}[c]{$\trr_a$}\psfrag{b}[c]{$\trr_b$}\psfrag{c}[c]{$\trr_c$}\psfrag{d}[c]{$\trr_d$}
\psfrag{e}[c]{$\mathrm{In}_1$}\psfrag{f}[c]{$\mathrm{In}_2$}\psfrag{g}[c]{$\mathrm{Out}$}
\includegraphics[width=\textwidth]{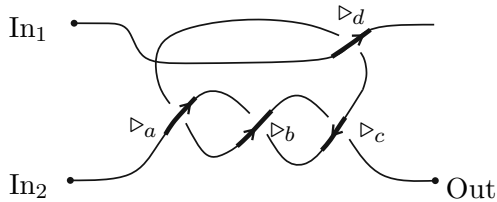}
\end{minipage}
\caption{\label{F:SampleComputation}A sample computation. Determining colours for input registers $\mathrm{In}_1$ and $\mathrm{In}_2$ uniquely and instantaneously determines the colour for the output register $\mathrm{Out}$.}
\end{figure}

A tangle machine computation has the following features:

\begin{enumerate}
\item Whereas the alphabet of a classical Turing machine is discrete (usually just $0$ and $1$ and maybe $2$), the alphabet $Q$ of a tangle machine can be any set, discrete or not. Two values of $Q$ may be chosen to represent $0$ and $1$, while the rest may represent something else--- perhaps electric signals.
\item Whereas a Turing machine computation is sequential with each step depending only on the state of the read/write head and on the scanned signal, a tangle machine computation is instantaneous and is dictated by an oracle. Time plays no role in a tangle machine computation.
\item A tangle machine computation may or many not be deterministic (colours may represent random variables and not their realizations), and it may or may not be bounded (contain a bounded number of interactions). Quandles and tangle machines are flexible enough to admit several different interpretations. In this paper, we use tangle machines to realize both logic gates (deterministic, composing into perhaps unbounded computations) and interactive proof computations (probabilistic, bounded size).
\item A tangle machine is flexible. There is a natural and intuitive set of local moves relating bisimilar tangle machine computations (Section~\ref{sec:lowdim}).
\item A tangle machine representation is abstract. A tangle machine computation takes place on the level of information itself, with no reference to time. The axioms of a quandle have intrinsic interpretation in terms of preservation and non-redundancy of information.
\end{enumerate}

\subsection{Results}\label{SS:Results}

The purpose of this note is to show the following theorems:

\begin{Theorem}\label{thm:boolean}
Any binary Boolean function can be realized by a tangle machine computation.
\end{Theorem}

This is neither hard nor new--- a previous such realization (not with tangle machines but with coloured braids) is recalled in Section~\ref{SS:Barrington}.

By Turing completeness of the boolean circuit model, we have the following:

\begin{Corollary}
Tangle machines (with an unbounded number of interactions) are Turing complete.
\end{Corollary}

In Section~\ref{S:TuringSimulation} we further prove the following.

\begin{Theorem}\label{T:TuringSimulation}
Any Turing machine can be simulated by a tangle machine. Such a tangle machine is coloured by a quandle $(Q,B)$
whose set of binary operations $B$ has cardinality $\mathcal{O}(n)$ where $n$ is the number of states in the finite control of
the underlying Turing machine.
\end{Theorem}

Tangle machines, whose notion of computation is based on an oracle which produced output colours from input colours, can in-fact perform super-Turing computations. See Remark~\ref{R:superturing}.

Colours of registers in tangle machines evolve at \emph{interactions}. If we bound the number of interactions, tangle machine computations include computations in a complexity class which we call $\mathrm{TangIP}$ which includes inside it a class which we call $\mathrm{BraidIP}$. Letting $\chi$ denote the number of interactions in the machine, letting $\dep$ denote a `noise parameter', and letting $c$ and $s$ denote \emph{completeness} and \emph{soundness} correspondingly (see Section~\ref{SS:IntroIP}), we have the following:

\begin{Theorem}
\label{thm:existence}
$\mathrm{IP} \subseteq \mathrm{BraidIP} \left\{ \dep, \chi \right\}$ where:
\begin{equation}
\label{eq:bounds}
I(c\dep) < \chi  < \frac{1}{I(1-s\dep)},
\end{equation}
with $I(p) \ass -p^{-1} \log p$. The growth rate of $\chi$ is $\mathcal{O}(\frac{1}{\dep})$ as $\dep\to 0$.
\end{Theorem}

Thus, tangle machine computations with bounded interactions can decide any language in class $\mathrm{IP}$, which is known to equal $\mathrm{PSPACE}$ \citep{Shamir:92}.

Moreover, in the special setting of non-adaptive $3$--bit probabilistically checkable proofs ($\mathrm{PCP}$), there exists a tangle machine which achieves a better soundness parameter than the best known classical single-verifier non-adaptive $3$--bit $\mathrm{PCP}$ algorithm (Section~\ref{SS:PCPHC}). Yet the soundness parameter for this tangle machine is above the conjectured lower limit. This suggests the possibility that tangle machines behave like very good single verifiers. 

Finally, in Section~\ref{sec:lowdim}, we discuss equivalence of machines. Machine equivalence formally parallels equivalence of tangled objects in low dimensional topology, and it gives us a formalism with which to discuss bisimulation. After justifying the definition, we introduce a notion of \emph{zero knowledge} for machines, in which the proof is kept secure from untrusted verifiers at intermediate nodes.

Tangle machines are thus revealed to be a flexible model for computation.


\subsection{Other low dimensional topological approaches to computation}\label{SS:Background}

The first person to consider distributivity as an axiom of primary importance, and to suggest diagrammatic calculi for logic (and perhaps by extension to computation) was American philosopher Charles Saunders Peirce. The following Peirce quotation was pointed out by Elhamdadi \citep{Elhamdadi:14}: ``These are other cases of the distributive principle\ldots These formulae, which have hitherto escaped notice, are not without interest.'' \citep{Peirce:80}.

The idea to use diagrammatic calculi from low-dimensional topology to model computation was pioneered by Louis Kauffman, who used knot and tangle diagrams to study automata \citep{Kauffman:94}, nonstandard set theory, and lambda calculus \citep{Kauffman:95,BuligaKauffman:13}. There is also a diagrammatic $\pi$ calculus formulation of virtual tangles \citep{MeredithSnyder:10}. The diagrammatic calculus of braids (braids are a special class of tangles) also underlies topological quantum computing--- see \textit{e.g.} \citep{KauffmanLomonaco:04,Nayak:08}. Universal logic gates (Toffoli gates) have been realized using coloured braids \citep{OgburnPreskill:99,Kitaev:03,Monchon:03}. This has led to a proposal for circuit obfuscation--- masking the true purpose of a circuit--- using braid equivalence \citep{Alagic:14}. Buliga has suggested to represent computations using a calculus of coloured tangles which is different from ours \citep{Buliga:11b}. In another direction, a different diagrammatic calculus, originating in higher category theory, has been used in the theory of quantum information--- see \textit{e.g.} \citep{AbramskyCoecke:09,BaezStay:11,Vicary:12}.

In our approach, the tangle diagrams themselves are computers, representing a flowchart of information during a computation whose basic operations are distributive (compare~\citep{Roscoe:90}). This is not a diagrammatic lambda calculus or pi calculus, but rather it is a natively low dimensional topological approach to computation. In this note, we relate this approach to other approaches by showing that tangle machine computation is Turing complete, and in the bounded resource setting that it can decide any language in the complexity class $\mathrm{IP}$.

\subsection{Contents of this paper}\label{SS:Contents}
We begin in Section~\ref{S:Models} by recalling relevant models of computation such as Turing machines, $\mathrm{IP}$, and $\mathrm{PCP}$. Then, in Section~\ref{S:TM}, we recall the formalism of tangle machines \citep{CarmiMoskovich:15}. Our definition is more general than the one used in that paper. Next in Section~\ref{S:TuringCompleteness} we show that tangle machines are Turing complete, and we show in Section~\ref{S:TuringSimulation} how tangle machines may simulate Turing machines. Restricting to a bounded resources setting, we construct networks of deformed $\mathrm{IP}$ verifiers in Section~\ref{S:DeformingIP}, defining a complexity class $\mathrm{BraidIP}$. In Section~\ref{sec:defip} we show that $\mathrm{IP}\subseteq \mathrm{BraidIP}$. Section~\ref{sec:efficientip} shows how to make our network computations more efficient by getting rid of the global time axis, making use of a machine we call the Hopf--Chernoff machine. Restricting further to a $\mathrm{PCP}$ proofs, in Section~\ref{sec:pcp} we show that the Hopf--Chernoff machine gives us perfect completeness and a better soundness parameter than the best-known non-adaptive $3$--bit $\mathrm{PCP}$ verifier. Finally, we define equivalence of machines in Section~\ref{sec:lowdim}, where we also discuss the tangle machine analogue of a zero knowledge proof.

\section{Models of computation}\label{S:Models}

In this section, mainly to fix terminology and notation, we recall the notion of a Turing machine (Section~\ref{SS:Turing}), of decidable languages (Section~\ref{SS:Decidable}), of interactive proof (Section~\ref{SS:IntroIP}), and of probabilistically checkable proof (Section~\ref{SS:IntroPCP}).

\subsection{Turing machines}\label{SS:Turing}

The theory of computation and complexity theory are based on the notion of a \emph{Turing machine} \citep{Turing:37}. We recall its definition, following \citep{Hopcroft:01}.

\begin{Definition}
A \emph{Turing machine} is a triple $(\Sigma, \mathcal{S}, \delta)$ where
\begin{itemize}
\item $\Sigma$ is a finite set of symbols called the \emph{alphabet} which contains a ``blank'' symbol.
\item $\mathcal{S}$ is a finite set of ``machine states'' with $q_0 \in \mathcal{S}$ and $q_h \in \mathcal{S}$ being, respectively, the initial and final (halting) states.
\item $\delta\colon\, \mathcal{S} \times \Sigma \longrightarrow \mathcal{S} \times \Sigma \times \epsilon$ is a transition function.
\end{itemize}
\end{Definition}
The set $\epsilon = \{0,1,2\}$ indicates the movement of a tape (\emph{Left}, \emph{Stationary}, \emph{Right}), or equivalently indicates the movement of a reading/writing (R/W) head following a writing operation. For convenience and without loss of generality, we limit the alphabet to three colours, $\Sigma= \{0,1,2\}$, where $2$ represents the blank symbol.

A Turing machine is composed of two primary units. A \emph{finite control unit} remembers the current state and determines the next state based on the current reading from a \emph{memory unit}. The memory unit records symbols on a finite, possibly unbounded \emph{tape}. Reading and writing operations retrieve or modify a symbol in the current position of the R/W head along the tape.

\subsection{Computable functions and decidable languages}\label{SS:Decidable}

Computable functions are the basic objects of study in computability theory. In the context of Turing machines, a partial function $f\colon\, \Sigma^k\to \Sigma$ is \emph{computable} if there exists a Turing machine that terminate on the input $x$ (\emph{input} means tape content) with the value $f(x)$ stored on the memory tape if $f(x)$ is defined, and which never terminates on input $x$ if $f(x)$ is undefined.

A related notion is the notion of a decidable language. A set $L\subseteq \{0,1\}^\ast$, called a \emph{language}, is said to be \emph{decided} by Turing machine $M$ if there exists a computable function $f\colon\, \Sigma^k\to \Sigma$ satisfying $f(x)=1$ if $x\in L$ and $f(x)=0$ if $x\notin L$ for all $x \in \{0,1\}^\ast$. A language is \emph{decidable} if it is decided by some Turing machine.

\subsection{Interactive proof}\label{SS:IntroIP}

The \emph{interactive proof} model of computation involves \emph{bounded resources}, by which we mean that computations are constrained to make use of only a finite number of steps, polynomial in the length $\abs{x}$ of the word $x \in \{0,1\}^\ast$ \citep{Goldwasser:89}. Again, the goal is to determine whether $x\in L$ or $x\notin L$ for a language $L$. A \emph{verifier} $V$ interrogates a \emph{prover} $P$ who claims to have a proof that $x\in L$. Both the prover and the verifier are assumed to be honest and queries are assumed to be independent. We are given two parameters, \emph{completeness} $c$ and \emph{soundness} $s$, with $c,s \in [0,1]$. For the classical setting of $\mathbf{IP}$ we set $c= 2/3$ and $s=1/3$. The verifier $V$ believes that $x\in L$ at time $t$ with probability $V_t$. This belief is updated each time $P$ responds to a query, beginning from $V_0= 0$. We say that the statement $x\in L$ is \emph{decided at time $t$} if:

\begin{equation}
\label{eq:ipdef2}
\begin{array}{ll}
\text{\emph{(Completeness)}} & x \in L \; \longrightarrow \; \Pr(V_t= 1) \geq c; \\[1ex]
\text{\emph{(Soundness)}} & x \notin L \; \longrightarrow \; \Pr(V_t=1) \leq s.
\end{array}
\end{equation}

The class $\mathrm{IP}$ (Interactive Polynomial time) consists of those languages $L$ that are decidable in time $\chi$ polynomial in $\abs{x}$. A celebrated result in complexity theory states that $\mathrm{IP}$ equals $\mathrm{PSPACE}$, the class of problems solvable by a Turing machine in polynomial space \citep{Shamir:92}.

The class $\mathrm{IP}$ can be expanded to the class $\mathrm{MIP}$ in which the verifier has access to not one but many provers, which can be interrogated independently \citep{BenOr:88}. It has been shown that $\mathrm{MIP}$ equals the large class $\mathrm{NEXPTIME}$ \citep{Babai:91}.

\subsection{Probabilistically checkable proofs}\label{SS:IntroPCP}

The class $\mathrm{PCP}_{c,s}(r(\abs{x}),q(\abs{x}))$ is a restriction of $\mathrm{IP}$ in which the verifier is a polynomial-time Turing machine with access to $\mathcal{O}(r(\abs{x}))$ uniformly random bits and the ability to query only $\mathcal{O}(q(\abs{x}))$ bits of the proof, with completeness $c$ and soundness $s$ \citep{AroraSafra:98, MoshkovitzRaz:10}. The celebrated $\mathrm{PCP}$ Theorem states that $\mathrm{PCP}[\mathcal{O}(\log \abs{x}),\mathcal{O}(1)] = \mathrm{NP}$ \citep{Arora:98}. For $\mathrm{PCP}$, the prover is thought of as an oracle. We restrict to the case of $3$--bit $\mathrm{PCP}$ verifiers, \textit{i.e.} to those for which $q(\abs{x})=3$, and to those which are \emph{non-adaptive}, \textit{i.e.} for which verifier queries are independent of previous responses by the oracle.

The ideal $\mathrm{PCP}$ verifier would have $c=1$ with minimal soundness $s$. A simple and good $3$--bit $\mathrm{PCP}$ verifier with $s=\frac{3}{4}+\sigma\approx 0.75$ for arbitrarily small $\sigma>0$ was designed by H{\aa}stad \citep{Has:97}. The current record is $s= \frac{20}{27}+\sigma\approx 0.741$ for an arbitrarily small $\sigma>0$ \citep{KhotSaket:06}. It is conjectures that the lowest possible value of $s$ is $\frac{5}{8}=0.625$ \citep{Zwick:98}.

\section{Tangle machines}\label{S:TM}

In this section we define tangle machines, first ignoring colours (Section~\ref{SS:NoCol}) and then, after defining the sets we colour with (Section~\ref{SS:QuagmaDefn}), then with colours (Section~\ref{SS:TMCol}), where we also define tangle machine computations.

\subsection{Without colours}\label{SS:NoCol}

We define here a \emph{tangle machine}, without colours, to be a diagram which occurs by concatenating (connecting endpoints) of a finite number of \emph{generators} of the form in Figure~\ref{F:Generators}. Thickened lines in interactions (Figure~\ref{F:Generators}\ref{F:Interaction}) are called \emph{agents}, and each thin lines are called \emph{patients}. Only agents are directed, and their is no compatibility condition between directions of different agents.

\begin{figure}
\centering
\begin{subfigure}{.28\textwidth}
  \centering
\includegraphics[width=1in]{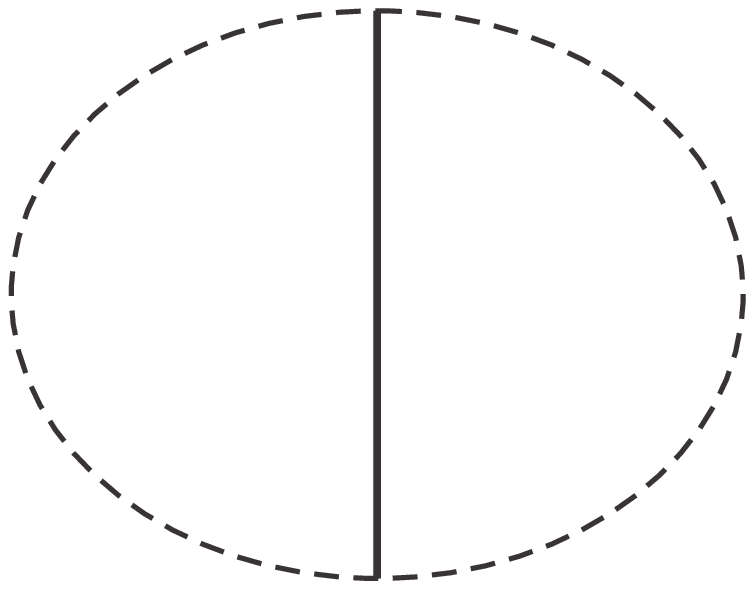}
\caption{Strut.}\label{F:Strut}
\end{subfigure}%
\begin{subfigure}{.28\textwidth}
  \centering
\includegraphics[width=1in]{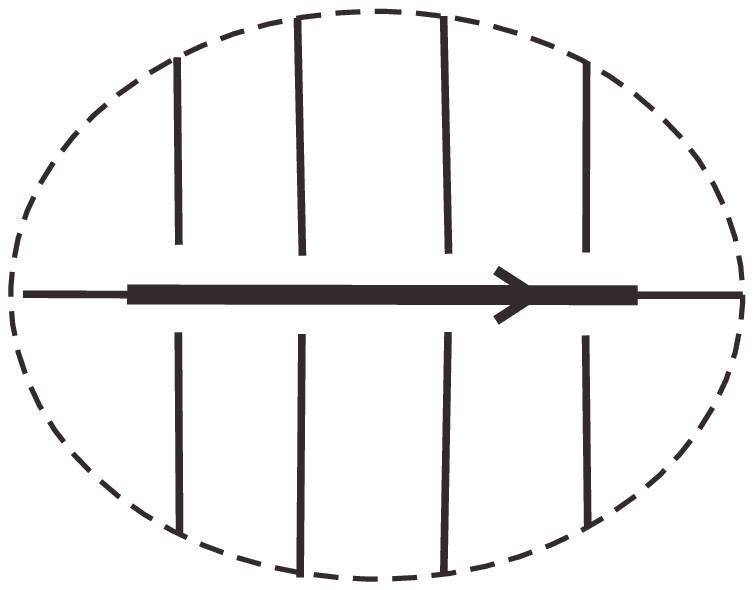}
\caption{Interaction.}\label{F:Interaction}
\end{subfigure}
\begin{subfigure}{.28\textwidth}
  \centering
\includegraphics[width=1in]{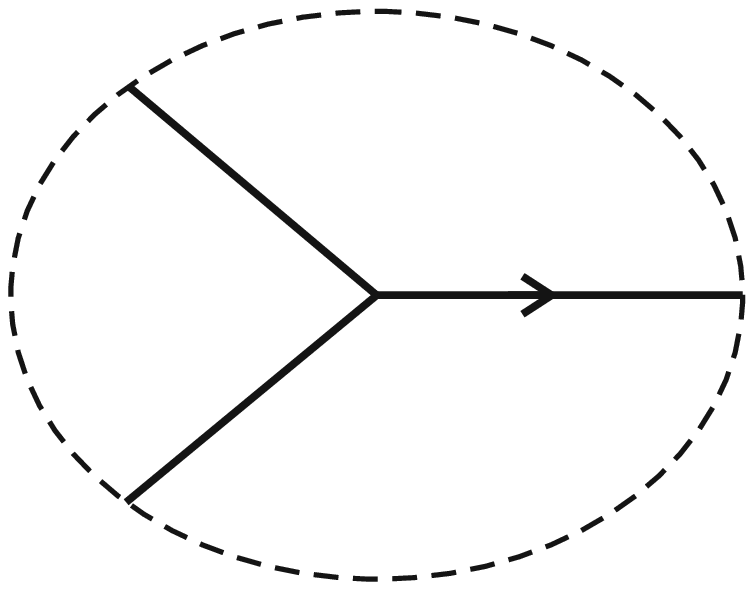}
\caption{Wye.}\label{F:Y}
\end{subfigure}
\caption{\label{F:Generators} Generators for tangle machines are struts and interactions. Generators for trivalent tangle machines are struts, interactions, and wyes.}
\end{figure}


Arcs in a tangle machine, ending as an under-arcs passing under an agent or at a wye or at a machine endpoint (thus ``continuing right through agents''), are called \emph{registers}. Later on, registers will contain colours.

\emph{Tangle machines} have struts and interactions as generators, whereas \emph{trivalent tangle machines} have struts, interactions, and wyes as generators. An example of a machine constructed by concatenation is presented in Figure~\ref{F:concat}. As in the theories of virtual knots and of w-knotted objects, concatenation lines may intersect \citep{BarNatanDancso:13,Kauffman:99}. Also, as in the theory of \emph{disoriented~tangles}, no compatibility condition is imposed for directions of concatenated agents \citep{ClarkMorrisonWalker:09}.

\begin{figure}
\centering
\begin{minipage}{0.75\linewidth}
\includegraphics[width=\linewidth]{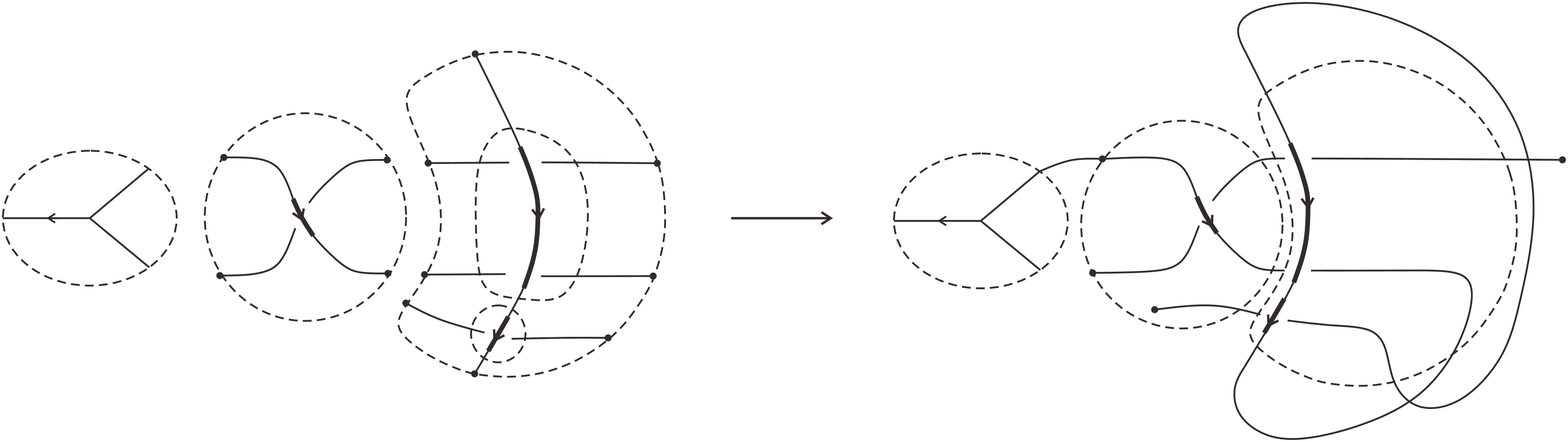}
\end{minipage}
\caption{\label{F:concat}Concatenation of tangle machines.}
\end{figure}

\subsection{Colours}\label{SS:QuagmaDefn}

Let $Q$ be a set equipped with a family $B$ of binary operations which satisfy the following three axioms:

\begin{description}
\item[Idempotence] $x\trr x=x$ for all $x\in Q$ and for all $\trr\in B$.
\item[Reversibility] The map $\trr y\colon\, Q\to Q$, which maps each colour $x\in Q$ to a corresponding colour $x\trr y\in Q$, is a bijection for all $(y,\trr)\in (Q,B)$. In particular, if $x\trr y = z\trr y$ for some $x,y,z\in Q$ and for some $\trr\in B$, then $x=z$.
\item[Distributivity] For all $x,y,z\in Q$ and for all $\trr \in B$:
\begin{equation}\label{E:Distributivity} (x\trr y)\trr z= (x\trr z)\trr (y\trr z)\enspace .\end{equation}
\end{description}

\begin{Remark}
Reversibility may be weakened by requiring only that $\trr y$ be an injection for all $y$, but not necessarily a surjection. This is indeed the case for the machines that we describe in the context of interactive proofs.
\end{Remark}

The pair $(Q,B)$ is called an \emph{quagma}. It is called an \emph{quandle} if the operations in $B$ distribute over one another, in the sense that:

\begin{equation}\label{E:Distributivity2} (x\trr y)\brr z= (x\brr z)\trr (y\brr z)\enspace ,\end{equation}

\noindent for all $\trr,\brr\in B$. This definition is a variant of definitions found in \citep{Buliga:11a,Przytycki:11,Ishii:13}. If we weaken reversibility to require only that $\trr y$ be an injection for all $y\in Q$ and for all $\trr \in B$, the resulting structure is called a \emph{quandloid} \citep{CarmiMoskovich:14}. Quagmas, quandloids, and quandles serve as the content of registers in tangle machines. 

\begin{Example}[Conjugation quandle]\label{E:QuantumGate}
Colours might be elements of a group $\Gamma$, and the operation might be conjugation:
\begin{equation}
x\brr y = y^{-1}xy\enspace .
\end{equation}
The pair $(\Gamma,\set{\brr})$ is called a \emph{conjugation quandle}. Such quandles feature in knot theory, \textit{e.g.} \citep{Joyce:82}.
\end{Example}

\begin{Example}[Linear quandle]\label{E:LinearQuandle}
Colours might real numbers and the operations might be convex combinations:
\begin{equation}
x\trr_s y = (1-s)x + sy \qquad s\in \mathds{R}\setminus\set{1}\enspace.
\end{equation}
The pair $\left(Q,\set{\trr_s}_{s\in \mathds{R}}\right)$ is called a \emph{linear quandle}.
\end{Example}

\subsection{With colours}\label{SS:TMCol}

A \emph{colouring} of a tangle machine $M$ is an assignment $\varrho$ of a binary operation in $B$ to each agent and an assigment $\rho$ of an element of $Q$ to each register, such that, at each interaction, the colour $z\in Q$ of the patient to the right of the agent (according to the right-hand rule) equals the colour of its corresponding patient to the left of the agent $x\in Q$ right-acted on by the colour of the agent $y\in Q$ via the operation of the agent $\trr\in B$:
    \begin{equation}
\begin{minipage}{80pt}\psfrag{a}[r]{$_x$}\psfrag{b}[c]{$_y$}\psfrag{c}[l]{$_{x\,\trr\, y}$}\includegraphics[width=80pt]{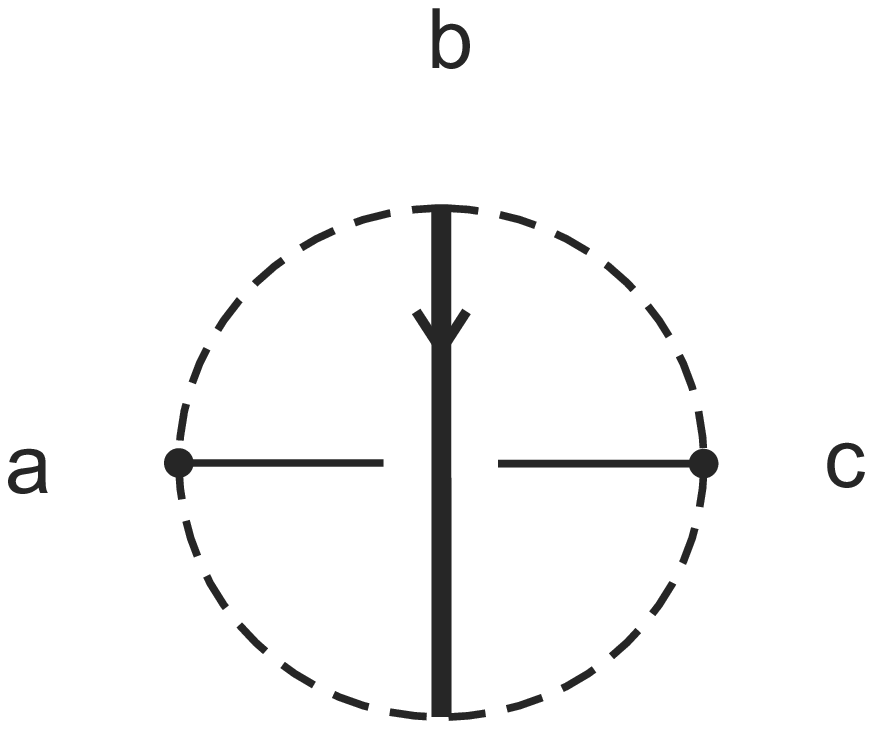}\end{minipage}
\end{equation}
Thus, an interaction `realizes' the action of the quagma or of the quandle.

To colour a trivalent tangle machines, $Q$ has to come equipped with a complete ordering. We include also an assigment of either $\max$ or $\min$ to each wye, so that the exiting register of the wye is the maximum of the two inputs if the wye is labeled $\max$ and is their minimum otherwise. We graphically denote a $\min$ label with a white circle at the branch-point of the wye, and a $\max$ label by a black circle. See Figure~\ref{F:MaxMinGates}.

 \begin{figure}
\centering
\begin{subfigure}{.38\textwidth}
  \centering
\psfrag{a}[c]{\small $x$}
\psfrag{b}[c]{\small $y$}
\psfrag{c}[l]{\small $\max(x,y)$}
\includegraphics[width=1.3in]{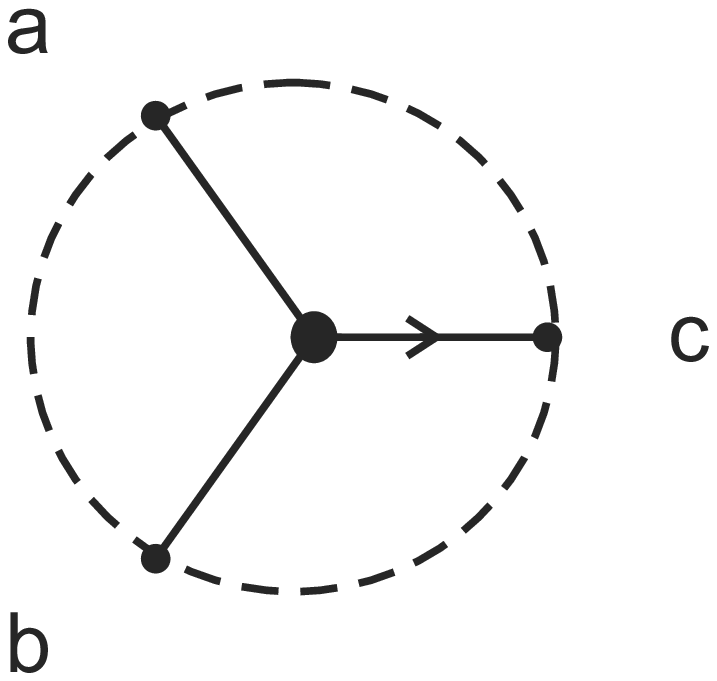}
\caption{}
\end{subfigure}%
\begin{subfigure}{.38\textwidth}
  \centering
\psfrag{a}[c]{\small $x$}
\psfrag{b}[c]{\small $y$}
\psfrag{c}[l]{\small $\min(x,y)$}
\includegraphics[width=1.3in]{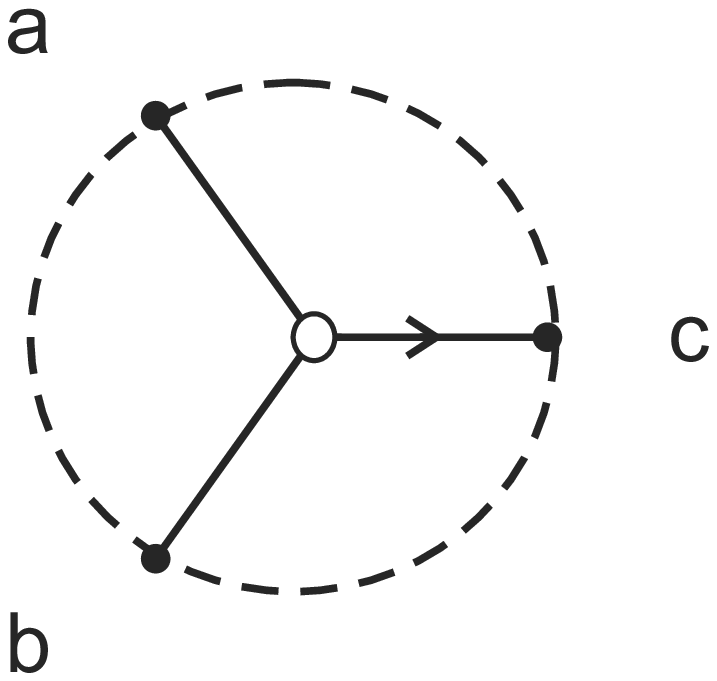}
\caption{}
\end{subfigure}
\caption{\label{F:MaxMinGates} Colouring wyes coloured by $\max$ and $\min$.}
\end{figure}

\begin{Remark}
The notion of tangle machine given above is more general than in \citep{CarmiMoskovich:14,CarmiMoskovich:15}, in which all tangle machines are quandle-coloured and without wyes.
\end{Remark}

Finally, we come to the notion of a tangle machine computation.

\begin{Definition}\label{D:TMComputation}
A \emph{computation} of a (trivalent) tangle machine $M$ is:
\begin{enumerate}
\item A choice of a set $S_\textit{in}$ of \emph{input registers} in $M$.
\item A choice of a set $S_\textit{out}$ of \emph{output registers} in $M$ with $S_\textit{in}\cap S_\textit{out}=\emptyset$.
\item A colouring $\varrho$ of all agents in $M$ (and an assignment of either $\max$ or $\min$ to each wye).
\item A colouring $\rho_\textit{in}$ of all registers in $S_\textit{in}$.
\item A unique (oracle) determination of a colouring $\rho_\textit{out}$ of all registers in $S_\textit{out}$. If $\rho_\textit{in}$ does not uniquely determine $\rho_\textit{out}$, the computation cannot take place.
\end{enumerate}
\end{Definition}

\section{Tangle machines are Turing complete}\label{S:TuringCompleteness}

Our goal in this section is to realize the universal set of gates, $\mathrm{NOT}(\neg)$ and $\mathrm{AND}(\wedge)\}$, as well as a \emph{multiplexer} which duplicates the content of a register, using tangle machines. This realizes the boolean circuit model, thus showing that tangle machines are Turing complete.

\begin{Remark}\label{R:superturing}
A tangle machine computation, which is based on an oracle which tells us the colours of output registers given colours of input registers, can carry out super-Turing computations. Consider the conjugation quandle of a group $\Gamma$ with unsolvable conjugacy problem, and colour agents $\mathrm{In}_1$ and $\mathrm{In}_2$ by elements of $\Gamma$. The colour of $\mathrm{Out}$ might not be Turing computable, but the tangle machine computation computes it.

\begin{equation}
\begin{minipage}{80pt}\psfrag{a}[r]{$_{\mathrm{In}_1}$}\psfrag{b}[c]{$_{\mathrm{Out}}$}\psfrag{c}[l]{$_{\mathrm{In}_2}$}\includegraphics[width=80pt]{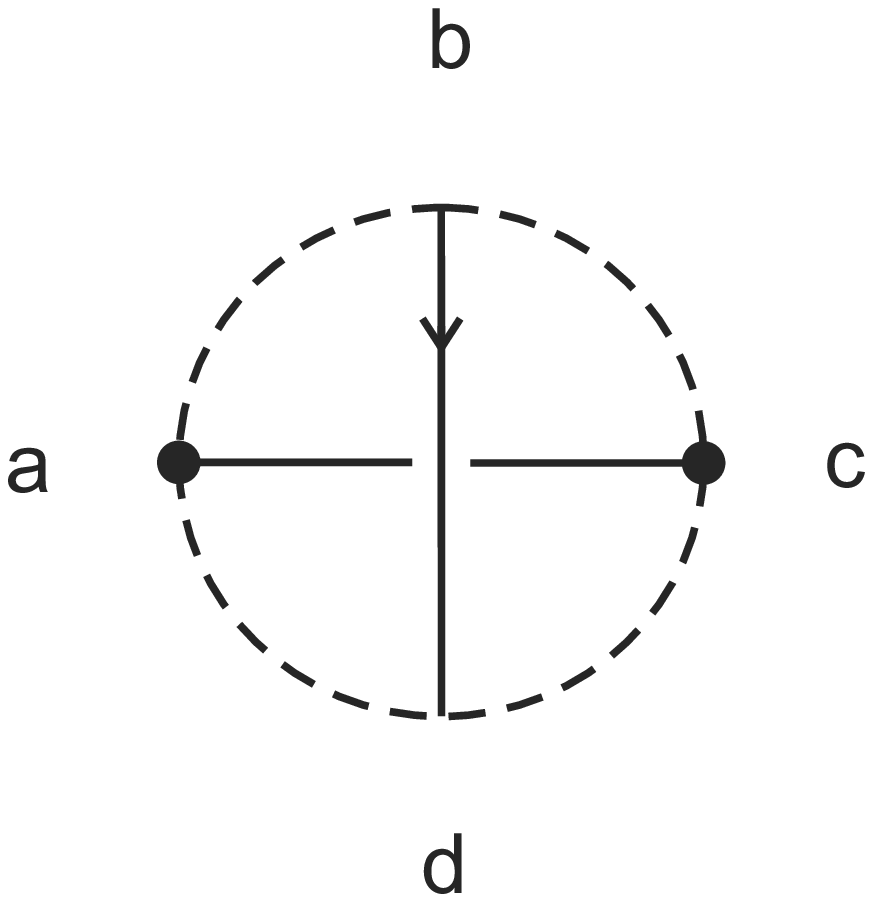}\end{minipage}
\end{equation}
\end{Remark}

\subsection{Quagma approach}\label{SS:Quagma}

Choose the set of colours to be $Q\ass \mathds{Q}^{2 \times 2}$, equipped with a set $B$ of binary operations whose elements are the following:

\begin{subequations}
\begin{equation}
X\blacktriangleright Y = \left\{
               \begin{array}{ll}
                 Y^{-1}XY, & \hbox{if $\det(Y)\neq 0$;} \\
                 X, & \hbox{otherwise.}
               \end{array}
             \right.
\end{equation}
\begin{equation}
X\trr_s Y = (1-s)X+sY, \; \; \text{for $s=\frac{1}{2}$ or $s=2$.}
\end{equation}
\end{subequations}
The structure $(Q,\{\trr_{0.5},\trr_2,\blacktriangleright\})$, henceforth referred to simply as $Q$, is a quagma.

To realize boolean logic, let $A_0\ass \left(\begin{matrix}0 & 1 \\ 1 & 0\end{matrix}\right)$ and $A_1\ass \left(\begin{matrix}1 & \phantom{-}0 \\ 0 & -1\end{matrix}\right)$ stand in for the digits $0$ and $1$ correspondingly. Incidentally, these happen to coincide with Pauli spin matrices of quantum mechanics. 


A $\mathrm{NOT}$ gate is described by a single interaction

\begin{equation}
\begin{minipage}{80pt}
\psfrag{a}[r]{$_{X}$}
\psfrag{b}[c]{$_{\mathrm{A_0 + A_1}}$}
\psfrag{c}[l]{$_{\neg X}$}\includegraphics[width=80pt]{crosm}\end{minipage}
\qquad \qquad \qquad
\small
\begin{tabular}{l*{2}{c}r}
$X$              & $\neg X$ \\
\hline
$A_0$       & $A_1$ \\
$A_1$       & $A_0$
\end{tabular}
\end{equation}
where the respective truth-table is shown to the right of the diagram.
Explicitly,
\begin{equation}
\neg X = X \brr (A_0 + A_1) = \left(\begin{matrix}1 & \phantom{-}1 \\ 1 & -1\end{matrix}\right)^{-1} X \left(\begin{matrix}1 & \phantom{-}1 \\ 1 & -1\end{matrix}\right) = \frac{1}{2} \left(\begin{matrix}1 & \phantom{-}1 \\ 1 & -1\end{matrix}\right) X \left(\begin{matrix}1 & \phantom{-}1 \\ 1 & -1\end{matrix}\right)
\end{equation}

\noindent The input is the register labeled $X$, the output is the register labeled $\neg X$, and the remaining register is always coloured by the constant value $A_0+A_1$.

Realizing an $\mathrm{AND}$ gate can be split into several successive computations. Let $X$ and $Y$ be the inputs to the gate, and write $\mathbf{0}$ for $\left(\begin{matrix}0 & 0 \\ 0 & 0\end{matrix}\right)$. The following instructions end up with the desired operation $X \wedge Y$.
\begin{enumerate}
\item $\beta_1 = (A_0 + A_1) \brr (X \trr_{_{0.5}} Y) = (X + Y)^{-1} (A_0 + A_1) (X + Y)$

\item $\beta_2 = \beta_1 \trr_{_{0.5}} (A_0 + A_1) = \frac{1}{2} \beta_1 + \frac{1}{2} (A_0+A_1)$

\item $\beta_3 = (\beta_2 \trr A_0) \trr_{_{0.5}} \beta_2 = \frac{1}{2} A_0 \beta_2 A_0 + \frac{1}{2} \beta_2$

\item $X \wedge Y = A_1 \brr (\beta_3 \trr_{_{0.5}} A_1) = (\beta_3 + A_1)^{-1} A_1 (\beta_3 + A_1)$
\end{enumerate}
A tangle machine which realizes these steps is given below together with the respective truth-table.
\hspace{5pt}
\begin{equation}
\psfrag{1}[l]{\small $\trr_{_{0.5}}$}
\psfrag{2}[c]{\small $_\brr$}
\psfrag{3}[c]{\small $\phantom{a}\trr_{_2}$}
\psfrag{0}[c]{\small $\mathbf{0}$}
\psfrag{a}[c]{\small $X$}
\psfrag{b}[c]{\small $Y$}
\psfrag{g}[c]{\small $\beta_1$}
\psfrag{h}[c]{\small $\beta_2$}
\psfrag{f}[c]{\small $\beta_3$}
\psfrag{c}[c]{\small $X \wedge Y$}
\psfrag{s}[c]{\small $A_0+A_1$}
\psfrag{t}[c]{\small $A_0$}
\psfrag{u}[c]{\small $A_1$}
\includegraphics[width=0.7\textwidth]{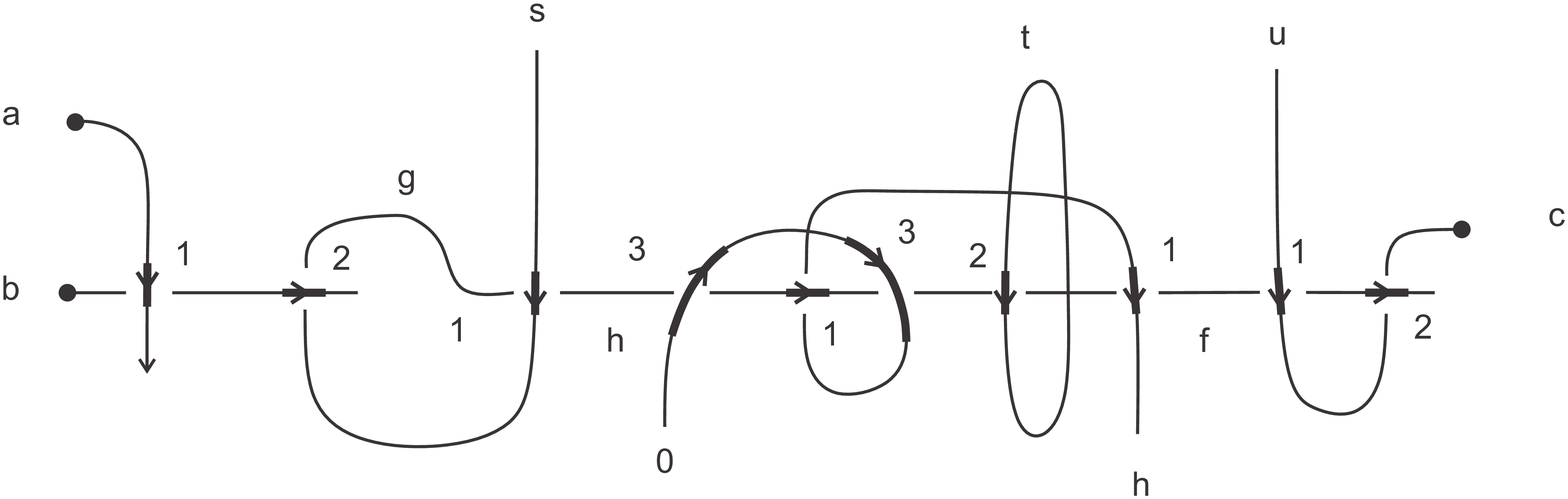}
\end{equation}
\hspace{5pt}
\begin{equation}
\small
\begin{tabular}{l*{6}{c}r}
$X$              & $Y$ & $\beta_1$ & $\beta_2$ & $\beta_3$ & $X \wedge Y$ \\
\hline
$A_0$       & $A_0$ & $A_0 - A_1$ & $A_0$ & $A_0$ & $A_0$ \\
$A_0$       & $A_1$ & $A_0 + A_1$ & $A_0 + A_1$ & $A_0$ & $A_0$ \\
$A_1$       & $A_0$ & $A_0 + A_1$ & $A_0 + A_1$ & $A_0$ & $A_0$ \\
$A_1$       & $A_1$ & $A_1 - A_0$ & $A_1$ & $\mathbf{0}$ & $A_1$ \\
\end{tabular}
\end{equation}

In addition to the universal set of gates we also need to be able to duplicate the content of a register. The conventional
boolean circuit model includes junction points along wires. The tangle machine analogue for such a junction is a \emph{multiplexer},
that is a machine whose output colours are duplicates of the colour in one of its inputs. A multiplexer takes an input of the form
$\{X,\underbrace{0, \ldots,0}_{n-1 \; \text{times}}\}$, and outputs $\{\underbrace{X, \ldots, X}_{n \; \text{times}}\}$. The operation of a multiplexer is
captured by the machine in Figure~\ref{F:Multiplexer}.

\begin{figure}
\centering
\psfrag{1}[l]{\small $\trr_{_{0.5}}$}
\psfrag{2}[c]{\small $\,\trr_{_{2}}$}
\psfrag{a}[c]{\small $\mathbf{0}$}
\psfrag{b}[c]{\small $X$}
\includegraphics[width=0.3\textwidth]{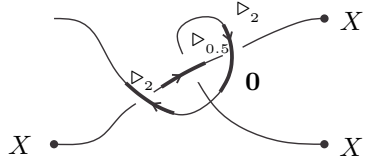}
\caption{\label{F:Multiplexer}A multiplexer.}
\end{figure}

\begin{Remark}\label{R:NoCloning}
Two fundamental properties of quantum information are \emph{no-cloning} and \emph{no-deleting}. The uniform version of no-cloning states that there does not exist a unitary operator $C$  for which $C (A \otimes e) C^\dagger = A \otimes A$ for all states $A$, where $e$ denotes the identity operator. The uniform version of no-deleting states that there does not exist a unitary operator $U$ for which $U (A\otimes A) U^\dagger = A\otimes e$ for all states $A$. Both of these statements are captured by the quagma axioms. Consider the quagma $(Q,\set{\trr_{0.5},\brr})$ where $Q$ is the set of invertible operators and $\brr$ is any binary operation which distributes over $\trr_{0.5}$, \textit{e.g.} conjugation. If $U$ were a universal cloning operator with respect to $\brr$, \textit{i.e.} if $(A\otimes e)\brr U= A\otimes A$ for all $A$, then for any state $A$ both machines in Figure~\ref{F:NoCloning} would carry out the same computation and in particular $\mathrm{Out}_1=\mathrm{Out}_1^\prime$. But then $(A \trr_{0.5} B) \otimes (A \trr_{0.5} B) = \mathrm{Out}_1 = \mathrm{Out}_1^\prime = (A \otimes A) \trr_{0.5} (B \otimes B)$, which is false in general. Thus \emph{universal cloning violates distributivity}. No-deleting follows from no-cloning because if $U$ were a universal deleting operator with respect to $\brr$ then $U$ would also be a universal cloning operator with respect its inverse operation $\blacktriangleleft$ which exists thanks to reversibility. This would violate distributivity, because $\blacktriangleleft$ also distributes over $\trr_{0.5}$ as can be seen by applying $\blacktriangleleft Z$ to both sides of the equation:
\[
X\trr_{0.5} Y = \left(\rule{0pt}{12pt}(X\blacktriangleleft Z)\trr_{0.5}(Y\blacktriangleleft Z)\right)\brr Z\enspace .
\]
We parenthetically note that no-cloning and no-deleting are also captured by a different diagrammatic calculus, that of categorical quantum mechanics \citep{Abramsky:10}.
\end{Remark}

\begin{figure}
\centering
\psfrag{0}[r]{\small $A \otimes e$}
\psfrag{1}[r]{\small $B \otimes e$}
\psfrag{2}[c]{\small $C$}
\psfrag{a}[l]{\small $\mathrm{Out}_1$}
\psfrag{b}[c]{\small $ $}
\psfrag{c}[c]{\small $ $}
\psfrag{t}[c]{\small $ $}
\psfrag{s}[c]{\small $ $}
\psfrag{e}[l]{\small $\mathrm{Out}_1^\prime$}
\psfrag{f}[c]{\small $ $}
\psfrag{g}[c]{\small $ $}
\psfrag{x}[c]{$\trr_{0.5}$}
\psfrag{y}[c]{$\brr$}
\includegraphics[width=0.35\textwidth]{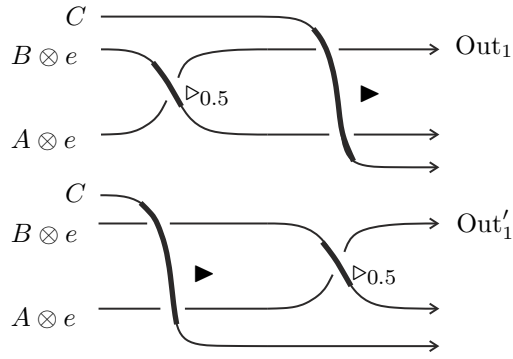}
\caption{\label{F:NoCloning}Universal cloning violates distributivity.}
\end{figure}

\subsection{Wye approach}\label{SS:Wye}

Realizing a universal set of logic gates is easier if we allow wyes, and may be realized with a quandle colouring.

We colour our machines by elements of the quandle $Q\ass \{0,1,2\}$ subject to the quandle operation $x\trr y= 2y-x\bmod 3$ (this is a \emph{Fox $3$--colouring}). We totally order $Q$ by $0<1<2$. Colour-code $0$ as red, $1$ as blue, and $2$ as green. For this particular quandle the direction of the agent does not matter because the operation $\trr$ is its own inverse. Let $0$ stand in for the digit zero and $1$ stand in for the digit one. A universal set of logic gates and a multiplexer can be obtained as in Figure~\ref{F:GatesC}, where an incoming arrow represents input and an outgoing arrow represents output.

 \begin{figure}
\centering
\begin{subfigure}{.28\textwidth}
  \centering
\psfrag{a}[c]{\small $X$}
\psfrag{2}[cb]{\small $2$}
\psfrag{c}[c]{\small $\neg X$}
\includegraphics[width=1in]{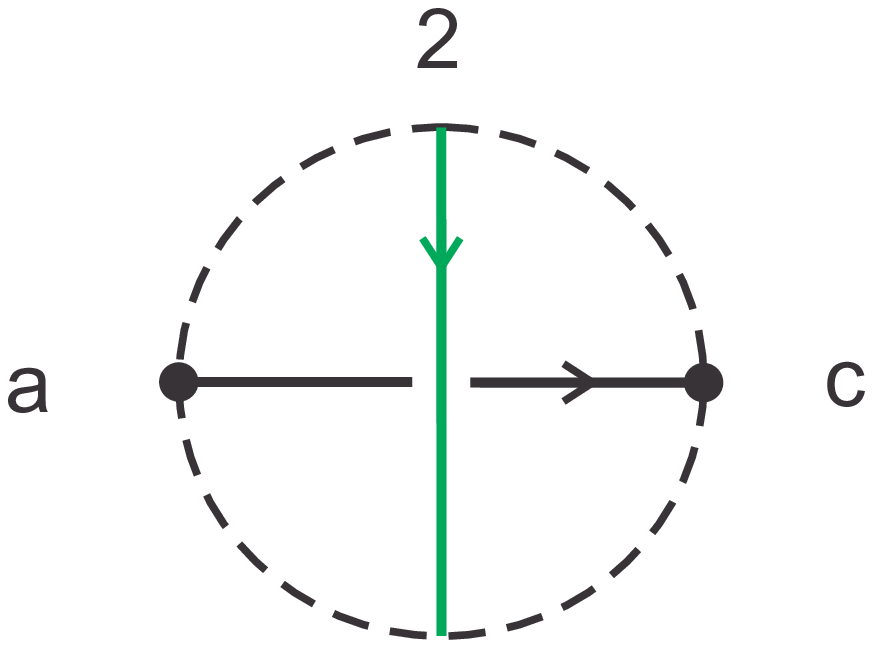}
\caption{$\mathrm{NOT}$ gate.}\label{F:NotGate}
\end{subfigure}%
\begin{subfigure}{.28\textwidth}
  \centering
\psfrag{a}[c]{\small $X$}
\psfrag{b}[c]{\small $Y$}
\psfrag{c}[l]{\small $\min(X,Y)$}
\includegraphics[width=0.8in]{min_gate}
\caption{$\mathrm{AND}$ gate.}\label{F:AndGate}
\end{subfigure}
\begin{subfigure}{.38\textwidth}
  \centering
\psfrag{b}[c]{\small $X$}\psfrag{a}[c]{\,$0$}
\includegraphics[width=0.7\textwidth]{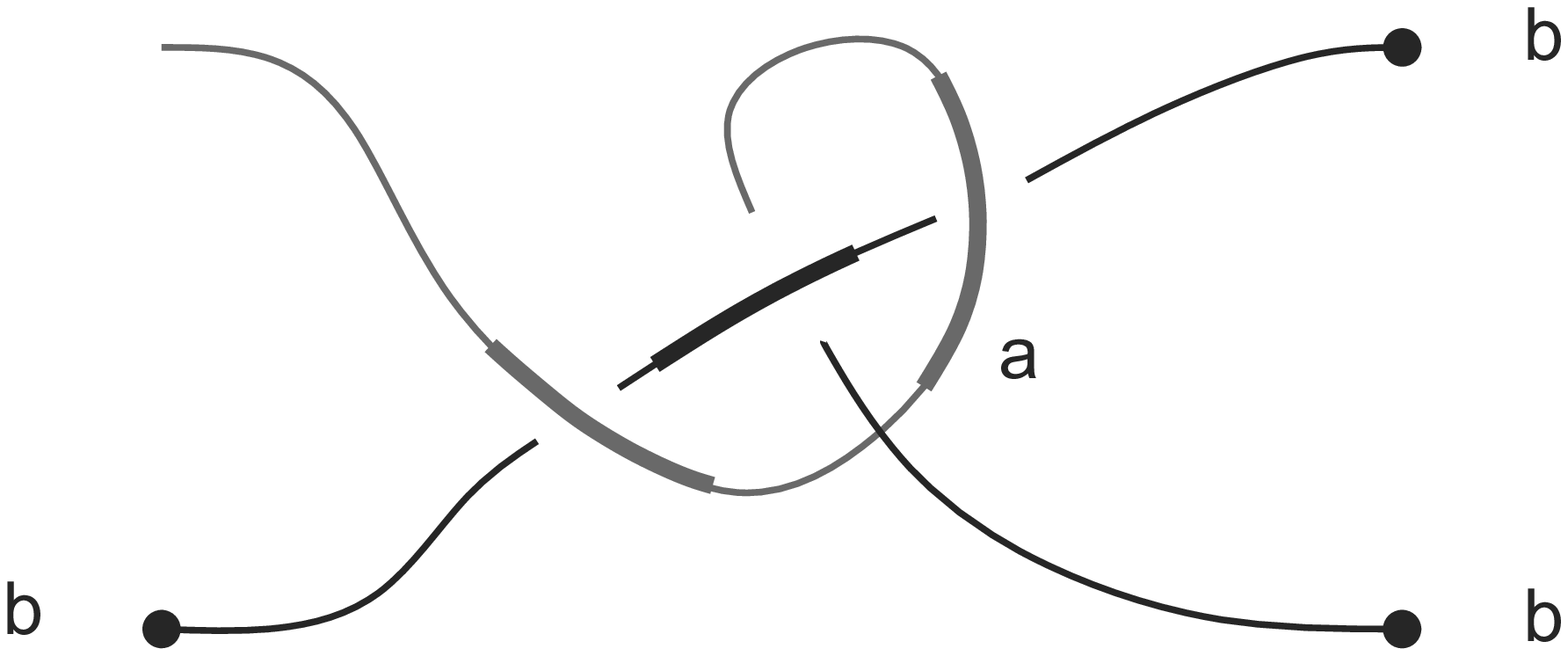}
\caption{Multiplexer.}\label{F:MultiplexerC}
\end{subfigure}
\caption{\label{F:GatesC} Logic gates for the $3$--colour approach. Inputs are the colours on the left, and outputs are colours on the right.}
\end{figure}


\subsection{Nonabelian simple group approach}\label{SS:Barrington}

The $\mathrm{AND}$ gate constructed in Section~\ref{SS:Quagma} was realized as a tangle machine coloured by a quagma which is not a quandle, and the $\mathrm{AND}$ gate of Section~\ref{SS:Wye} used a wye. In fact, it is possible to realize a universal set of logic gates with a conjugation quandle coloured tangle machine without using wyes.

The construction is based on Barrington's Theorem \citep{Krohn:66,Barrington:89}. Following \citep{Monchon:03}, Appendix A of \citep{Alagic:14} constructs a $132$ crossing $14$ strand braid coloured by the conjugation quandle of the finite simple group $A_5$ which realizes the Toffoli gate that is a universal logic gate \citep{Fredkin:82}. This braid is made a tangle machine by interpreting each crossing as an interaction. Three of its registers are inputs, three are outputs, and it contains a number of \emph{ancilla} which are what we called \emph{control registers}.

\section{Turing machine simulation}\label{S:TuringSimulation}

In this section, we show how to simulate a Turing machine using a tangle machine. 

\subsection{Turing tangle machine}

In this section we make use of the linear quandle whose underlying set of elements $Q$ is the set of the rational numbers and whose set of operations $B$ is:

\begin{equation}
x\trr_s y = (1-s)x + sy \qquad s\in \mathds{Q}\setminus\set{1}\enspace.
\end{equation}

There are several sub-machines that recur in the construction of a \emph{Turing tangle machine}, the tangle machine analog of a Turing machine. To simplify our diagrams we represent these sub-machines graphically.

\begin{description}

\item[Multiplexer] We indicate a multiplexer by splitting a strand.
\begin{equation}
\psfrag{b}[c]{\small $x$}
\includegraphics[width=0.15\textwidth]{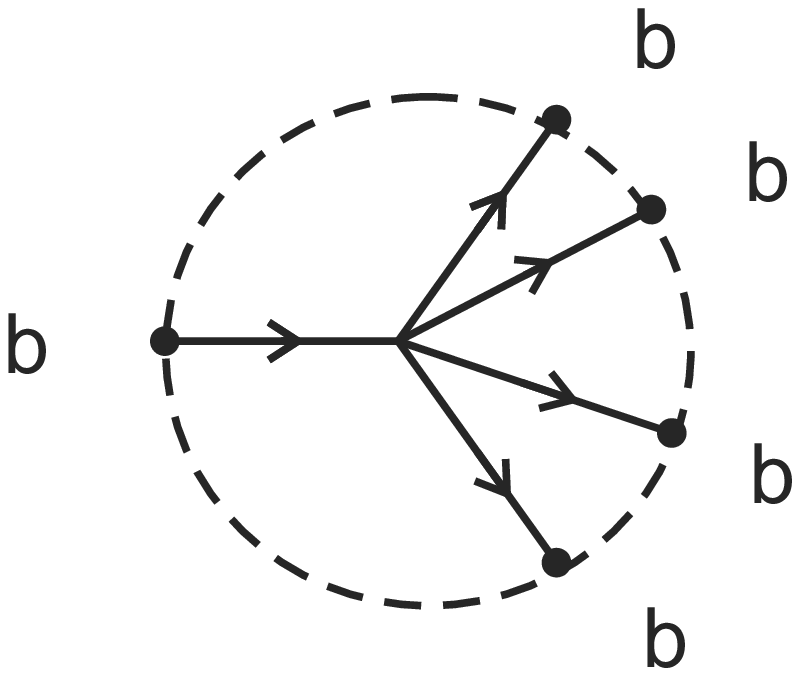}
\end{equation}
\noindent The multiplexer is realized by composing diagrams of the form in Figure~\ref{F:Multiplexer}.

\item[Negation] $\neg x \ass 1-x$.
\begin{equation}
\psfrag{a}[c]{\small $x$}
\psfrag{c}[c]{\small $\neg x$}
\psfrag{b}[c]{\small $0.5$}
\psfrag{d}[c]{\small $\trr_2$}
\psfrag{x}[c]{\large $\neg$}
\includegraphics[width=0.5\textwidth]{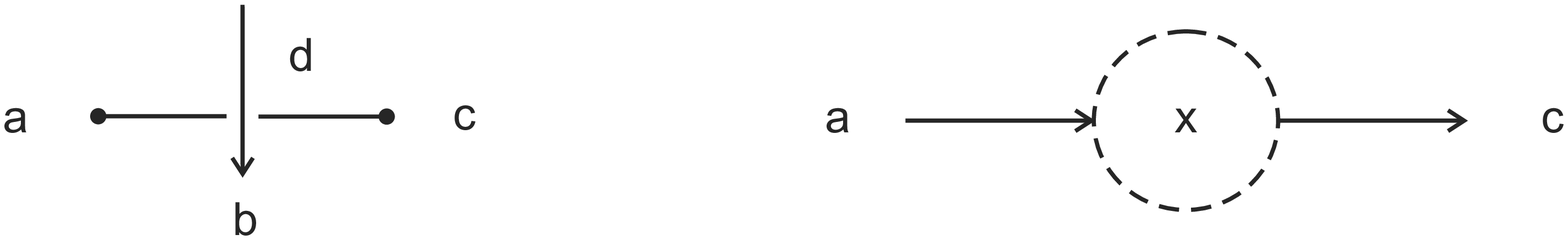}
\end{equation}

\item[Addition]
\begin{equation}
\psfrag{a}[c]{\small $x$}
\psfrag{c}[c]{\small $y$}
\psfrag{b}[c]{\small $x+y$}
\psfrag{d}[c]{\small $\trr_2$}
\psfrag{s}[c]{\small $\trr_{0.5}$}
\psfrag{0}[c]{\small $0$}
\psfrag{+}[c]{+}
\includegraphics[width=0.57\textwidth]{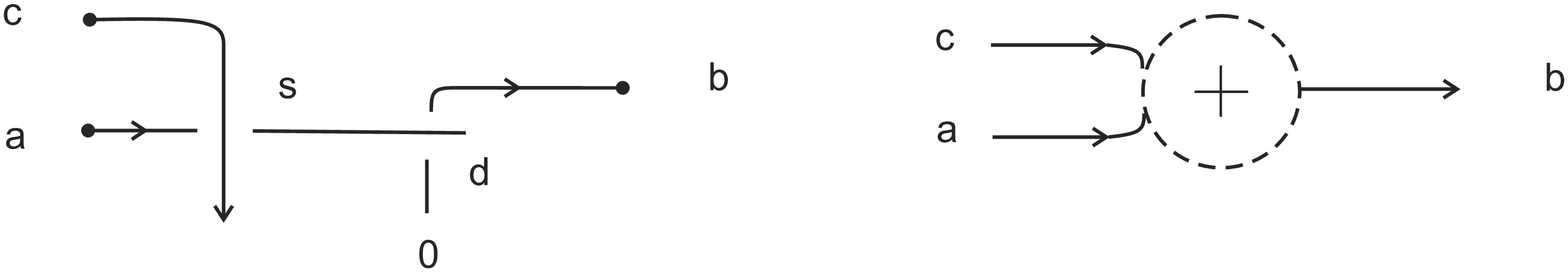}
\end{equation}

\item[Indicator] If $x\geq 1$ then$\imath(x) = 1$, if $0<x<1$ then $\imath(x)=x$, otherwise $\imath(x) = 0$.
\begin{equation}
\psfrag{a}[c]{\small $x$}
\psfrag{b}[c]{\small $\imath(x)$}
\psfrag{x}[c]{$\neg$}
\psfrag{y}[c]{$\; \imath$}
\psfrag{m}[c]{\small $\max$}
\psfrag{0}[c]{\small $0$}
\includegraphics[width=0.70\textwidth]{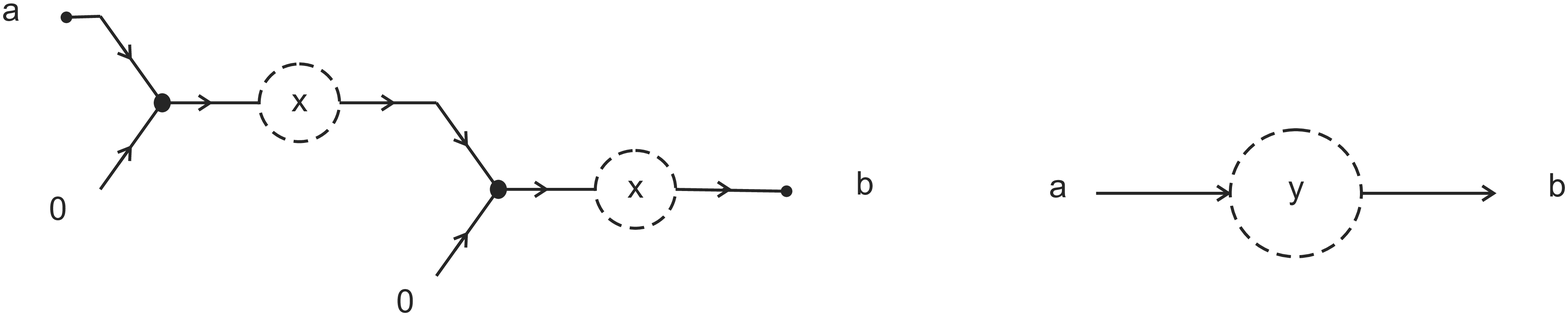}
\end{equation}

\item[Beta] A function $\beta(x)$ which satisfies $\beta(0) = -1$, $\beta(\frac{1}{2})=0$, and $\beta(1) = 1$.
\begin{equation}
\psfrag{a}[c]{\small $x$}
\psfrag{b}[c]{\small $\beta(x)$}
\psfrag{x}[c]{$\neg$}
\psfrag{y}[c]{\small $\; \beta$}
\psfrag{m}[c]{\small $\max$}
\psfrag{0}[c]{\small $0$}
\psfrag{s}[c]{\small $\trr_2$}
\psfrag{d}[c]{\small $\trr_{0.5}$}
\includegraphics[width=0.70\textwidth]{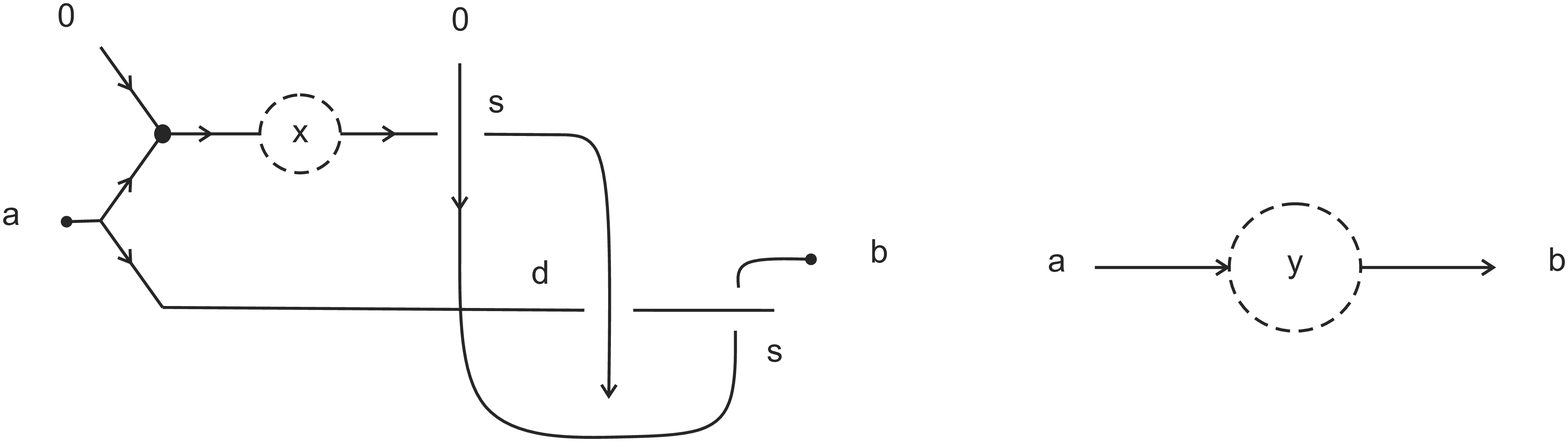}
\end{equation}

\item[Selector] Depending on the colour $c$ of a control strand, either $x$ or $y$ emerges as the output, $z$.
To be precise, $z=x$ if $c=0$ and $z=y$ if $c=1$.
\begin{equation}
\psfrag{a}[c]{\small $y$}
\psfrag{u}[c]{\small $\; \tau$}
\psfrag{b}[c]{\small $x$}
\psfrag{c}[c]{\small $c$}
\psfrag{z}[c]{\small $z$}
\psfrag{g}[c]{\small $s$}
\psfrag{x}[c]{$\neg$}
\psfrag{y}[c]{\small $\; \beta$}
\psfrag{m}[c]{\small $\max$}
\psfrag{0}[c]{\small $0$}
\psfrag{s}[c]{\small $\trr_s$}
\psfrag{d}[c]{\small $\trr_{0.5}$}
\includegraphics[width=0.85\textwidth]{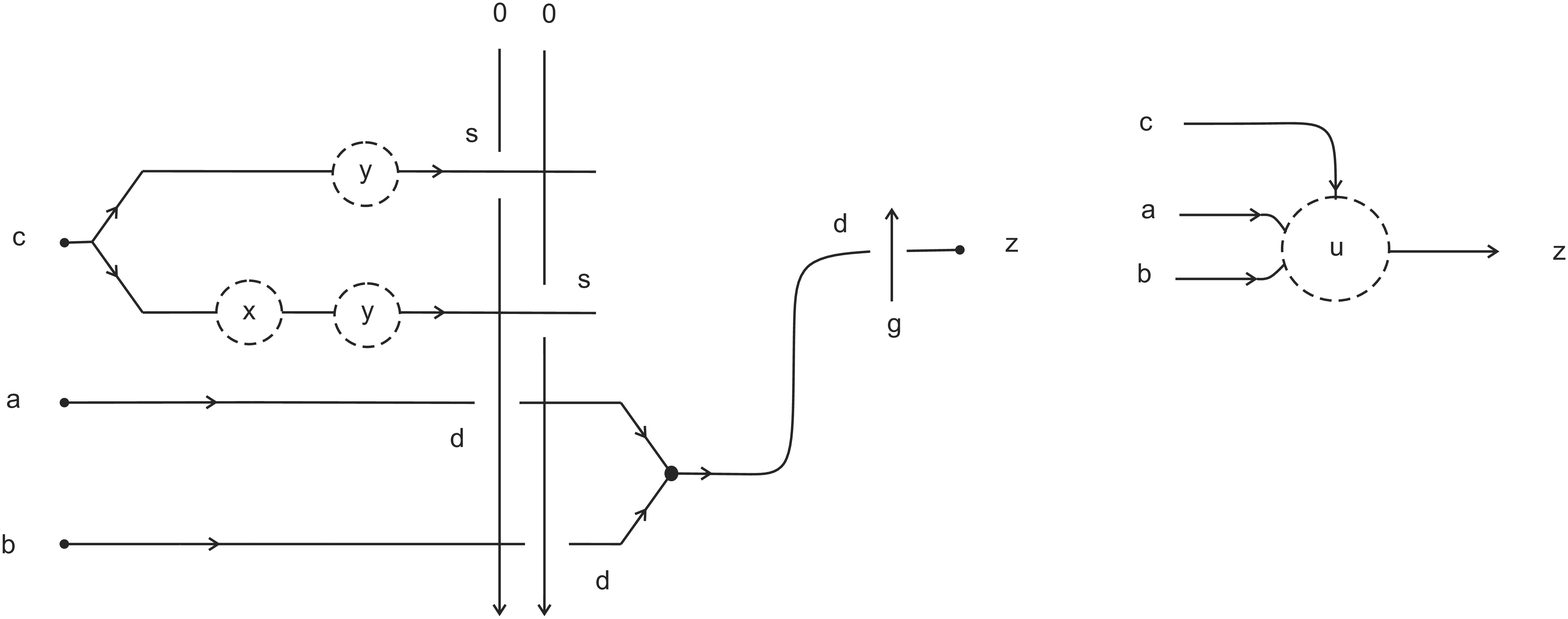}
\end{equation}
where $s$ is an arbitrary constant whose value is greater than $2$.

\item[Mask generating machine] This sub-machine is shown in Figure~\ref{fig:mask}. Its input registers are a register coloured by an integer $p\in Q$ called a \emph{pointer}, and a sequence of coloured registers together called a \emph{mask}. Its output registers are a register coloured $p$ and one register coloured $0$ for all other input coloured registers, except for a single register coloured $1$ in the $p$th position of the output.
  \end{description}

  \begin{figure}[htb]
  \centering
\psfrag{1}[c]{\small $1$}
\psfrag{2}[c]{\small $2$}
\psfrag{3}[c]{\small $3$}
\psfrag{p}[c]{\small $p$}
\psfrag{v}[c]{\small $\vdots$}
\psfrag{c}[c]{\small \emph{mask}}
\psfrag{k}[c]{\small $\; \gamma$}
\psfrag{y}[c]{$\neg$}
\psfrag{x}[c]{\small $\; \imath$}
\psfrag{m}[c]{\small $\min$}
\psfrag{0}[c]{\small $0$}
\psfrag{d}[c]{\small $\trr_2$}
\psfrag{s}[c]{\small $\trr_{0.5}$}
\psfrag{t}[c]{\small $\rrt_{0.5}$}
\includegraphics[width=0.75\textwidth]{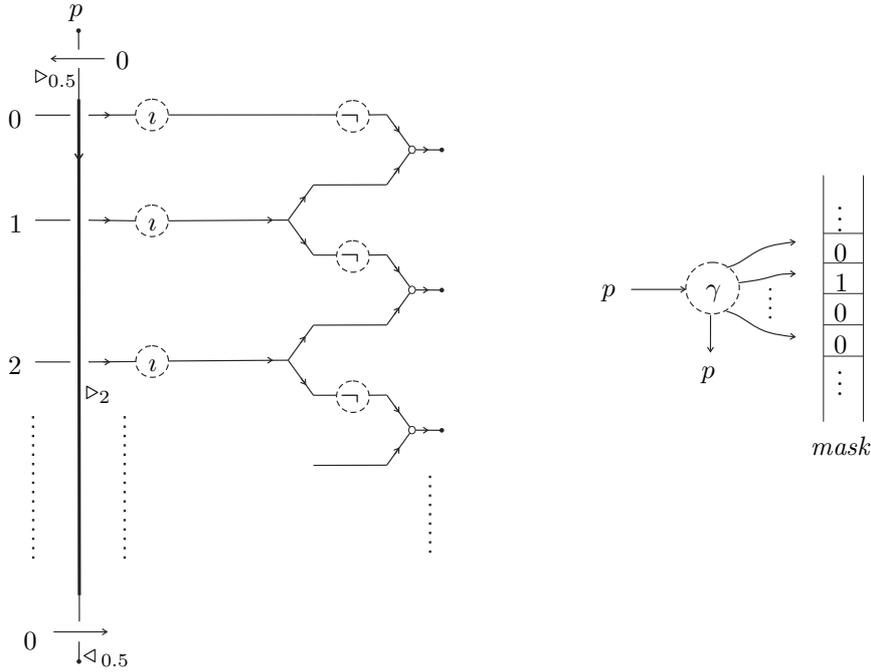}
\caption{\small A mask generating machine.}
\label{fig:mask}
\end{figure}

\subsection{Finite control}

The first step in realizing a Turing machine is to mimic the finite control unit, \textit{i.e.} to simulate the transition
function
\begin{equation}
\delta(q(k),u(k)) = (q(k+1), \; a, \; \epsilon)
\end{equation}
Given the \emph{current machine state} $q(k)$ and the \emph{symbol currently under the R/W head }$u(k)$, the \emph{transition function} determines the \emph{next machine state} $q(k+1)$ together with a pair of tape instructions: the \emph{symbol to be written} $a$, and an $\epsilon$ movement of the head to its next position along the tape. Without loss of generality we shall assume henceforth that the finite control states are all natural numbers, and in particular that $\mathcal{S} = \set{1, \ldots, n}$.

The basic building block of the finite control tangle machine is a \emph{hardwired transition}, that is a function $\delta_{i,j}: \mathcal{S} \times \Sigma \rightarrow \mathds{Z}^3$,
\begin{equation}
(q,u) \mapsto \left \{
\begin{array}{ll}
(\bar{q}+2, \; \bar{a}+2, \; \bar{\epsilon}+2), & \text{if $q=i$ and $u=j+1$;} \\
(-\bar{q}-2, \; -\bar{a}-2, \; -\bar{\epsilon}-2), & \text{otherwise.}
\end{array} \right .
\end{equation}
where $\bar{q}, \bar{a}, \bar{\epsilon}$ denote, respectively, the next assumed state and the tape instructions as specified in the definition of $\delta_{i,j}$. The parameters $i,j$ and the respective output of $\delta_{i,j}$, the triplet $(\bar{q},\bar{a},\bar{\epsilon})$, are hardwired into the tangle machine realization of $\delta_{i,j}$ through the quandle parameters. See Figure~\ref{fig:delta}.

\begin{figure}[htb]
\centering
\psfrag{q}[c]{\small $q$}
\psfrag{u}[c]{\small $u$}
\psfrag{k}[c]{\small $\; \gamma$}
\psfrag{y}[c]{\small $\; \beta$}
\psfrag{x}[l]{\small \emph{$i$th strand}}
\psfrag{z}[l]{\small \emph{$(j+1)$st strand}}
\psfrag{f}[c]{\small $\Bigg \} \delta_{i,j}$}
\psfrag{0}[c]{\small $0$}
\psfrag{s}[l]{\small $\trr_{\bar{q}+2}$}
\psfrag{t}[l]{\small $\trr_{\bar{a}+2}$}
\psfrag{r}[l]{\small $\trr_{\bar{\epsilon}+2}$}
\includegraphics[width=0.7\textwidth]{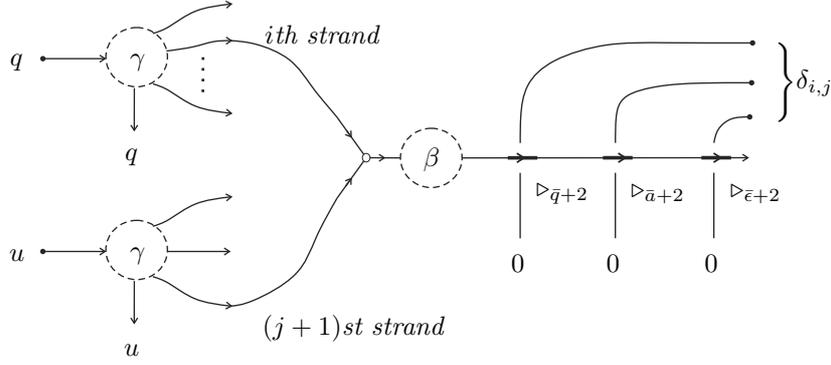}
\caption{\small Hardwired transition.}
\label{fig:delta}
\end{figure}

The hardwired transition is symbolically represented as
\[
\psfrag{q}[c]{\small $q$}
\psfrag{u}[c]{\small $u$}
\psfrag{k}[c]{\small $\; \delta_{i,j}$}
\psfrag{f}[c]{$\Bigg \} \delta_{i,j}$}
\includegraphics[width=0.25\textwidth]{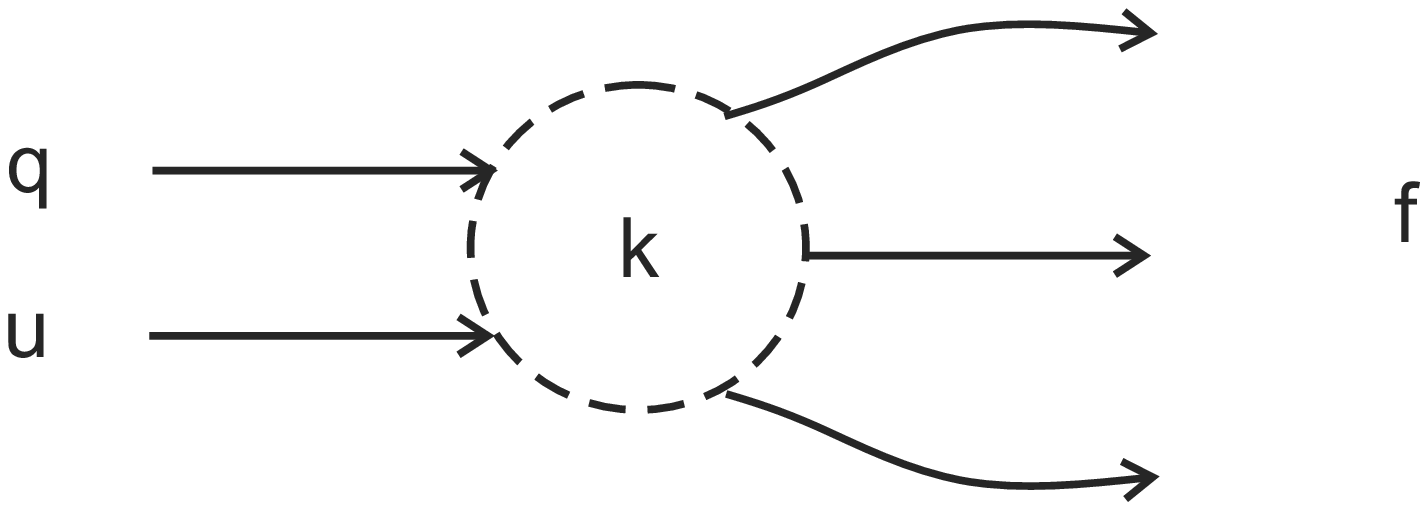}
\]

The transition function $\delta$ is constructed by combining several hardwired transitions. Note that there are not more than $3n$ possible transitions in the finite control ($n$ states multiplied by $3$ input symbols). Thus the number of distinct binary operations in $B$ does not exceed $9n + \mathcal{O}(1)$, which accounts for $3n$ triplets $(\trr_{\bar{q}+2},\trr_{\bar{a}+2},\trr_{\bar{\epsilon}+2})$ and a few other operations required for realizing the memory unit. A detailed construction of the finite control sub-machine is shown in Figure~\ref{fig:fc}.

\begin{figure}[htb]
\centering
\psfrag{q}[c]{\small $q(k)$}
\psfrag{u}[c]{\small $u(k)$}
\psfrag{2}[c]{\small $-2$}
\psfrag{k}[c]{\small $\delta_{1,0}$}
\psfrag{l}[c]{\small $\; \delta_{1,1}$}
\psfrag{m}[c]{\small $\; \delta_{1,2}$}
\psfrag{r}[c]{\small $\; \delta_{n,0}$}
\psfrag{s}[c]{\small $\; \delta_{n,1}$}
\psfrag{t}[c]{\small $\; \delta_{n,2}$}
\psfrag{a}[l]{\small $q(k+1)$}
\psfrag{b}[l]{\small $a$}
\psfrag{c}[l]{\small $\epsilon$}
\includegraphics[width=0.8\textwidth]{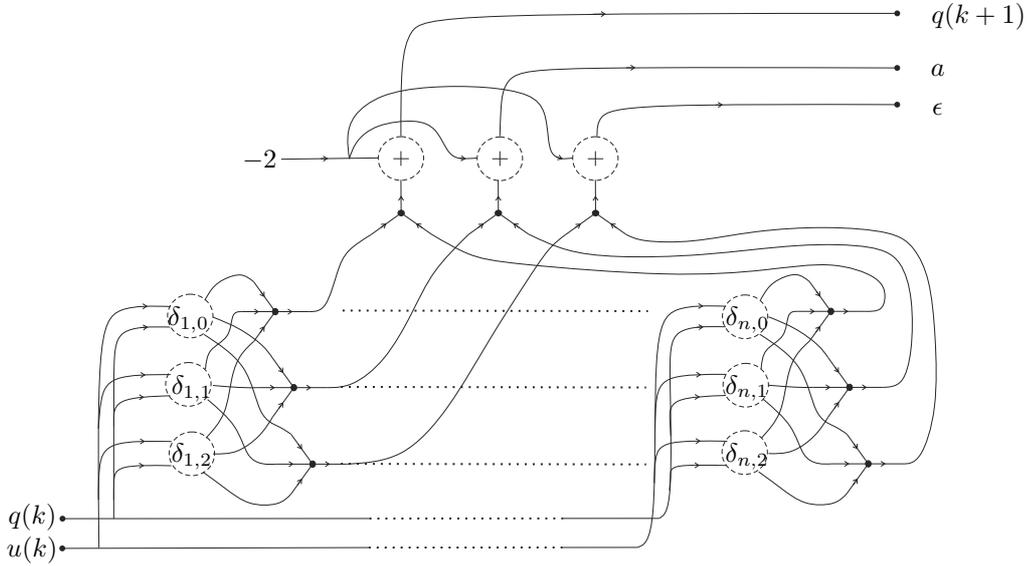}
\caption{\small A tangle machine realization of a finite control unit.} 
\label{fig:fc}
\end{figure}

\subsection{Memory unit and one step computation}

The memory unit consists mainly of the tape logic. It accepts a finite, possibly unbounded set of registers whose colours are manipulated in the basis of commands $a$ and $\epsilon$ received from the finite control. The majority of the registers of the memory unit correspond to tape cells whose content is represented by colours from the set $\Sigma$.

We construct the memory (tape) reading and writing operations using a mask generating machine together with selector machines to output the colour of the $p(k)$th register, either as it originally appeared or modified. See Figure~\ref{fig:write}.

\begin{figure}[htb]
\centering
\psfrag{p}[c]{\small $_{p(k)}$}
\psfrag{k}[c]{\small $\gamma$}
\psfrag{u}[c]{$\; \tau$}
\psfrag{b}[r]{\small $_{c_1(k)}$}
\psfrag{c}[r]{\small $_{c_2(k)}$}
\psfrag{d}[r]{\small $_{c_m(k)}$}
\psfrag{a}[c]{\small $_a$}
\psfrag{g}[l]{\small $_{c_1(k+1)}$}
\psfrag{h}[l]{\small $_{c_2(k+1)}$}
\psfrag{I}[l]{\small $_{c_m(k+1)}$}
\psfrag{v}[l]{\small $\vdots$}
\psfrag{z}[r]{\small $\vdots$}
\psfrag{s}[c]{\small $\trr_s$}
\psfrag{m}[c]{\small $\max$}
\psfrag{1}[c]{\small $_1$}
\psfrag{0}[c]{\small $_0$}
\psfrag{j}[c]{\small $_{u(k)}$}
\begin{subfigure}{.48\textwidth}
  \centering\includegraphics[width=0.74\textwidth]{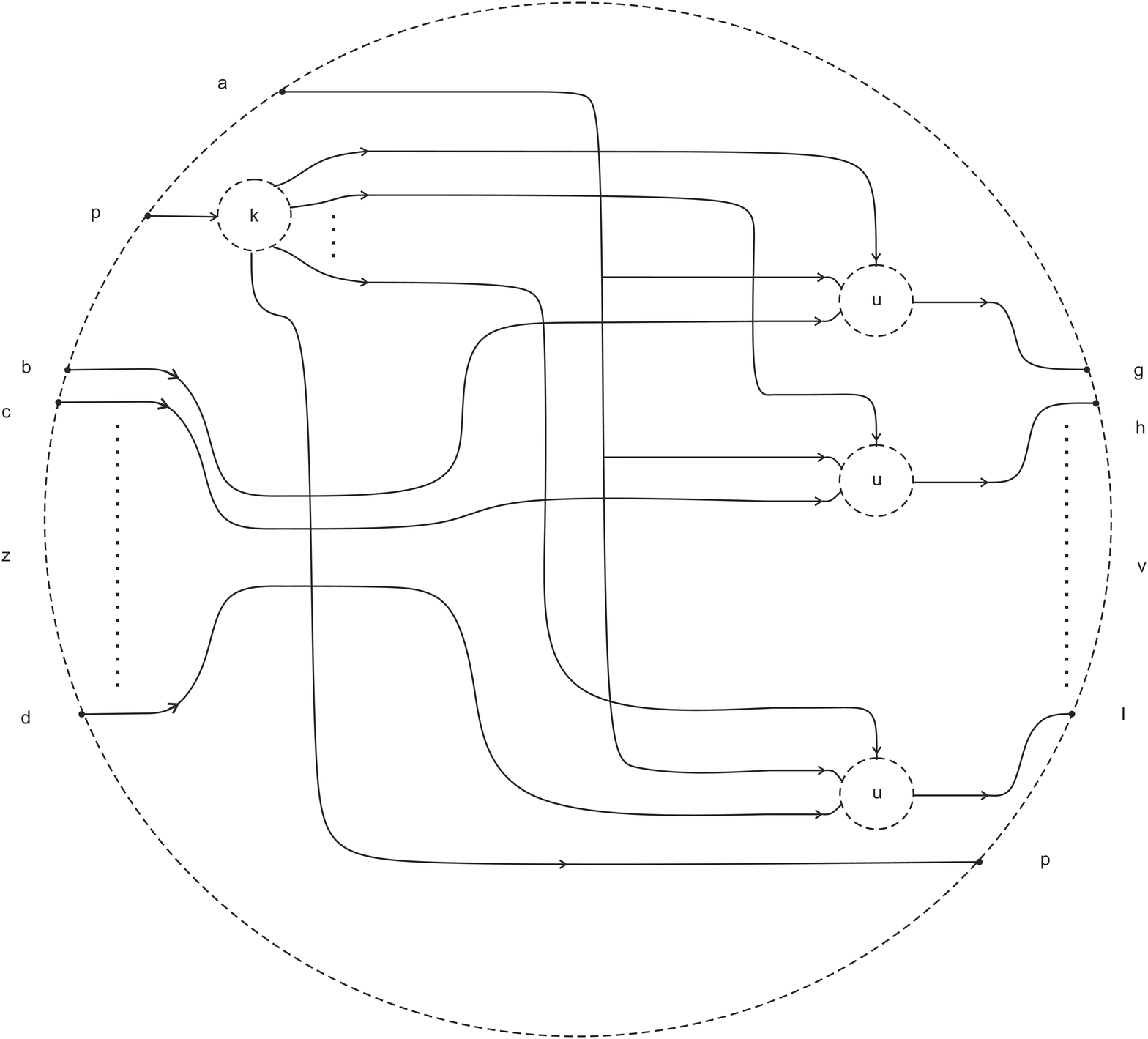}
  \caption{\emph{Write}}
\end{subfigure}
\begin{subfigure}{.48\textwidth}
  \centering
  \includegraphics[width=0.74\textwidth]{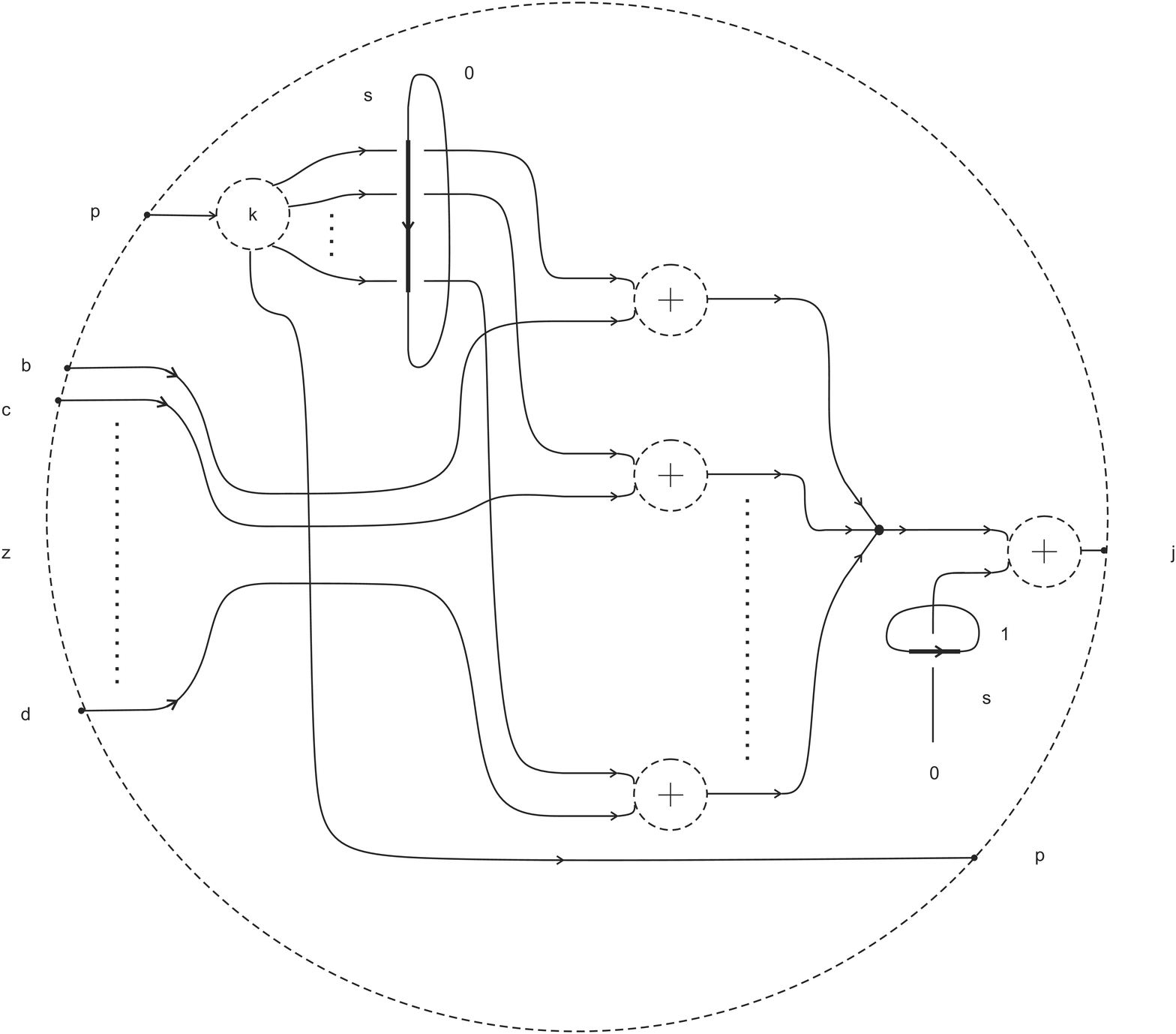}
  \caption{\emph{Read}}
\end{subfigure}
\caption{\small Memory reading and writing sub-machines. Here $\frac{1}{2}<s<1$ is an arbitrary parameter.}
\label{fig:write}
\end{figure}

Figure~\ref{fig:ts2} illustrates the memory unit, connected to the finite control unit. Its inputs are:

\begin{enumerate}
\item An \emph{integer pointer} register coloured $p(k)$, indicating the current head position.
\item A finite, possibly unbounded set of registers each of which represents a single (memory) cell on a tape. These tape registers are coloured $\{c_i(k)\}_{i > 0} \in \Sigma$.
\item A pair of registers coloured correspondingly by a pair instructions $(a, \epsilon)$, where $a$ denotes the symbol to be written in the current cell to which the head points (which is numbered $p(k)$), and $\epsilon$ denotes the (possibly zero) increment to be added to $p(k)$, so that $p(k+1)= p(k)+\epsilon-1$.
\end{enumerate}

The outputs of the memory unit are:

\begin{enumerate}
\item The updated pointer $p(k+1)$.
\item A finite, possibly unbounded set of registers whose colours $\{c_i(k+1)\}_{i > 0} \in \Sigma$ have been all passed unchanged, except for a single cell whose content may have been modified.
\item A strand coloured by $u(k+1)$, the content of the cell to which the head points in its new location $p(k+1)$
\end{enumerate}

\begin{figure}[htb]
\centering
\psfrag{5}[c]{\emph{W}}
\psfrag{4}[c]{\small $-1$}
\psfrag{6}[c]{\emph{R}}
\psfrag{a}[c]{\small $a$}
\psfrag{e}[c]{\small $\epsilon$}
\psfrag{x}[c]{\small \emph{memory}}
\psfrag{b}[r]{\small $c_1(k)$}
\psfrag{c}[r]{\small $c_2(k)$}
\psfrag{d}[r]{\small $c_m(k)$}
\psfrag{f}[r]{\small $p(k)$}
\psfrag{g}[l]{\small $c_1(k+1)$}
\psfrag{h}[l]{\small $c_2(k+1)$}
\psfrag{I}[l]{\small $c_m(k+1)$}
\psfrag{j}[l]{\small $p(k+1)$}
\psfrag{u}[l]{\small $u(k+1)$}
\psfrag{v}[l]{\small $\vdots$}
\psfrag{z}[r]{\small $\vdots$}
\psfrag{q}[c]{\small $q(k)$}
\psfrag{0}[l]{\small $q(k+1)$}
\psfrag{s}[c]{\small $u(k)$}
\psfrag{2}[c]{\small $-2$}
\includegraphics[width=0.5\textwidth]{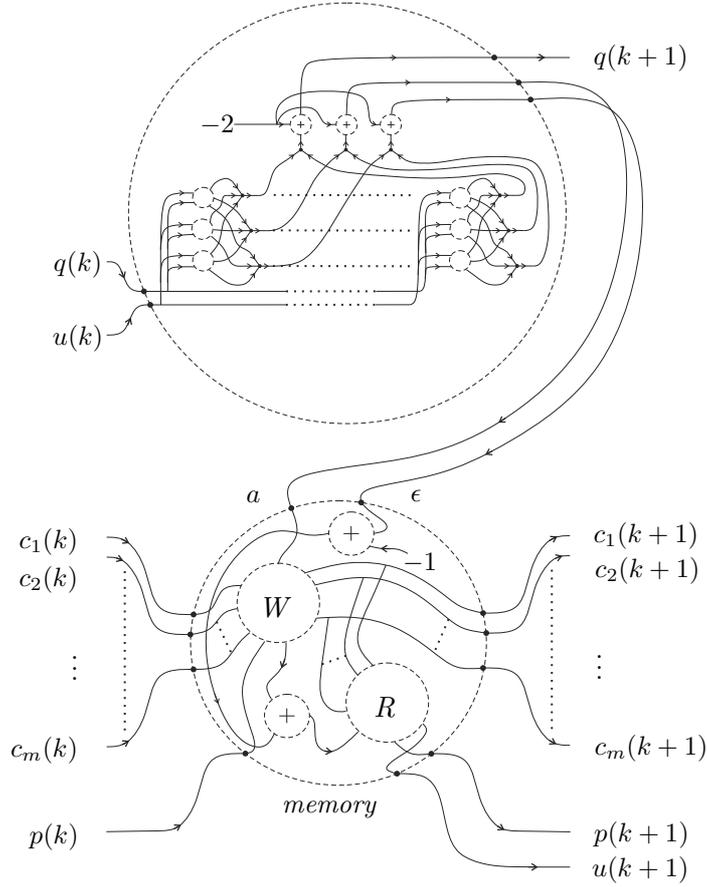}
\caption{\small One-step computation of a Turing tangle.}
\label{fig:ts2}
\end{figure}

To simulate the sequential operation of a Turing machine, copies of the finite control unit and of the memory unit are to concatenated in the obvious
manner. Concatenating $N$ copies of the machine in Figure~\ref{fig:ts2} simulates $N$ successive computations of a Turing machine (see Figure~\ref{fig:ts1}).
This justifies naming such a procedure \emph{iteration}.

\begin{figure}[htb]
\centering
\psfrag{x}[c]{\small \emph{memory}}
\psfrag{y}[c]{\small \emph{finite control}}
\includegraphics[width=0.65\textwidth]{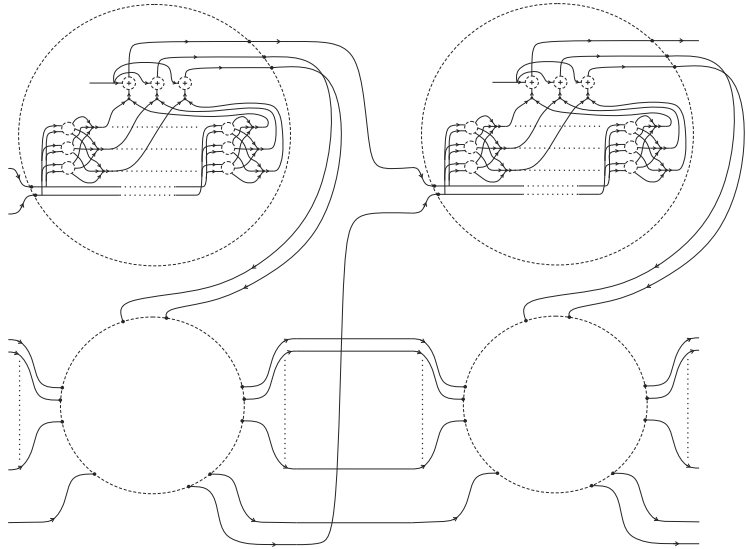}
\caption{\small Iterative computation in a Turing tangle machine.}
\label{fig:ts1}
\end{figure}

\subsection{Halting}

A halting state is a state for which:
\begin{equation}
\delta(q_h, u) = (q_h, u, 1)
\end{equation}
Once arriving at a halting state, further iterations do not alter the memory content of the machine. The state $q_h$ represents an equilibrium which may or may not be reached for a given input sequence $u(0), u(1), \ldots$. A machine whose input registers are coloured by a halting state can be closed by concatenating respective inputs and outputs as shown in Figure~\ref{fig:ts3}.

\begin{figure}[htb]
\centering
\psfrag{x}[c]{\small \emph{memory}}
\psfrag{y}[c]{\small \emph{finite control}}
\includegraphics[width=0.6\textwidth]{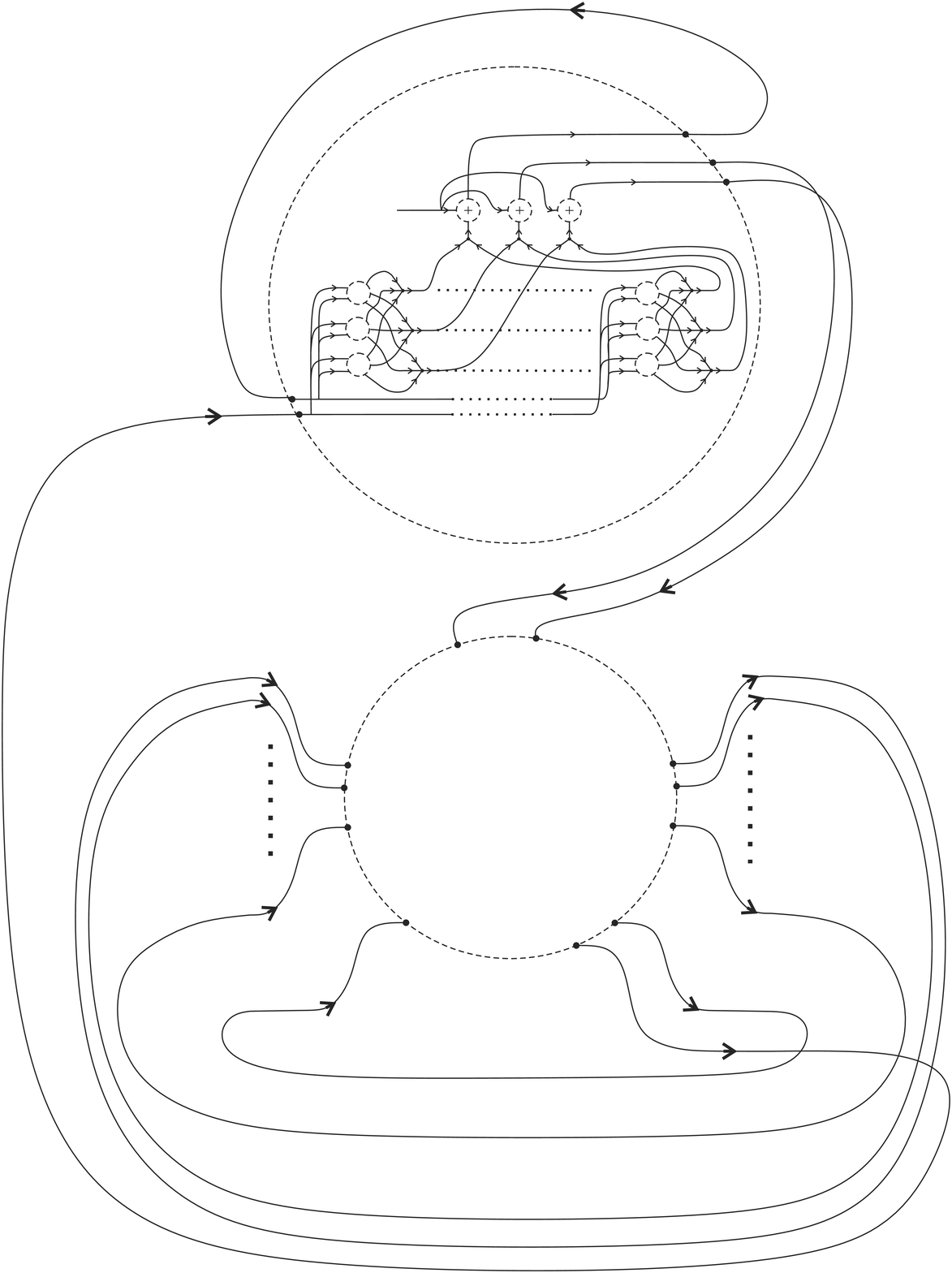}
\caption{\small Closure of a Turing tangle machine in its halting state.}
\label{fig:ts3}
\end{figure}

\section{Interactive proofs: Distribution of knowledge by deformation}\label{S:DeformingIP}

 The decision of a crowd can converge to a correct answer even when each individual has limited knowledge. This phenomenon is known as \emph{wisdom of the crowds} \citep{Surowiecki:05}. In line with `the many being smarter than the few', we extend the notion of interactive proof systems~\citep{Goldwasser:89} to a system in which a collection of verifiers interact to prove a claim together. Might such a crowd of verifiers collaborate to prove more than could be proven by any individual verifier in the crowd?

\subsection{Deformation of a single interaction}

Consider a family of verifiers, each with a belief concerning whether $x\in L$ or $x\notin L$. We model the belief of each verifier $W$ at time $t$ as a Bernoulli random variable $W_t$ whose realizations $w_t$ are either $\brak{\mathrm{True}}$ or $\brak{\mathrm{False}}$. We interpret $w_t=\brak{\mathrm{True}}$ as `$W$ believes at time $t$ that $x\in L$', and we interpret $w_t=\brak{\mathrm{False}}$ as `$W$ believes at time $t$ that $x\notin L$'.

Consider an interaction at time $t$ with agent $V$, one of whose patients is $W$. The realization $w_{t+1}$ of $W_{t+1}$ may equal either the belief of the agent $v_t$ or the belief of the patient $w_t$. In other words, $W$ either retains her belief or is `convinced' by $V$ to change her belief to that of $V$ (we use female pronouns for the verifiers, who are all `Alices'). Whether or not $V$ `succeeds in convincing $W$' depends on a message $\xi_t$ from a \emph{prover} $\Pi$ with access to an \emph{oracle}. 

Only the belief of patients changes at an interaction. The agents and the verifiers who do not participate in the interaction do not change their beliefs, so in particular $v_{t+1}=v_t$ always.

There are three constants associated to the agent $V$ at an interaction at time $t$: A \emph{completeness parameter} $c_V^t$, a \emph{soundness parameter} $s_V^t$ with $0<s_V^t<c_V^t\leq 1$, and a \emph{deformation parameter} $\dep_V^t\in \mathds{Q}\cap(0,1)$. For simplicity, we will assume that these three parameters are the same for all agents in the network, and $s_v^t,c_V^t,\dep_V^t$ will be written $s,c,\dep$ correspondingly.

Our basic requirement for an interaction is that the following pair of inequalities be satisfied:

\begin{equation}
\label{eq:ipdef4}
\begin{array}{ll}
\text{\emph{(deformed completeness)}} & x \in L \; \longrightarrow \; \Pr(W_{t+1} = v_t \mid V_t=v_t) \geq c \dep; \\[1ex]
\text{\emph{(deformed soundness)}} & x \notin L \; \longrightarrow \; \Pr(W_{t+1} = v_t \mid V_t=v_t) \leq s\dep.
\end{array}
\end{equation}

\begin{Remark}
The deformed completeness and soundness, $c \dep$ and $s \dep$, may both be below $\frac{1}{2}$ or may both be above $\frac{1}{2}$ and bounded away from $1$.
\end{Remark}

In the limit $\dep\to 1$, specific values of $c$ and of $s$ turn the pair of inequalities \eqref{eq:ipdef4} into familiar pairs of inequalities in interactive proof theory \citep{AroraBarak:2009}. For example, for $s=2^{-|x|^a}$ and $c=1-2^{-|x|^b}$, where $a,b > 0$, we obtain the completeness and soundness constraints of an $\mathrm{IP}$ verification.

\subsection{Statistics of beliefs and interactions}

When keeping track of the beliefs of many different verifiers at many different times, it is cumbersome to work directly with (\ref{eq:ipdef4}). Instead, we introduce a shorthand to keep track of the belief of a verifier, both if $x\in L$ and also if $x\notin L$, in a single expression.

The \emph{belief statistics} $\brak{W_t}$ of verifier $W$ at time $t$ is written:
\begin{equation}\label{eq:lowerbounds1}
\brak{W_t}\ass a\brak{\mathrm{True}}+ b\brak{\mathrm{False}}\enspace,
\end{equation}
where $a\in [0,1]$ denotes the \emph{greatest lower bound} for the belief of $W$ that $x\in L$ at time $t$ conditioned on this belief indeed being true, and $b\in[0,1]$ denotes the \emph{greatest lower bound} for the belief of $W$ that $x\notin L$ at time $t$ conditioned on this opposite belief indeed being true. Note that $a+b$ need not equal $1$. In particular:
\begin{equation}
\label{eq:lowerbounds2}
\begin{array}{ll}
x \in L \; &\longrightarrow \; a\leq \Pr(W_{t} = \brak{\mathrm{True}}); \\[1ex]
x \notin L \; &\longrightarrow \; b\leq \Pr(W_{t} = \brak{\mathrm{False}}).
\end{array}
\end{equation}

An interaction between $W$ and $V$ at time $t$ concludes either with $W$ accepting the belief of $V$ or with $W$ sticking to her own belief.
Denoting by $h$ the probability (or more precisely a lower bound on it) of $W$ switching to the belief of $V$ in the next time-frame, the distribution of possible beliefs of $W$ in the next time-frame is described by:

\begin{equation}\label{E:BeliefTransition}
\brak{W_{t+1}} = \brak{W_t}^{\brak{V_t}} \ass (1-h)\brak{W_t}+h\brak{V_t}.
\end{equation}

\begin{Remark}
Belief statistics written as in \eqref{eq:lowerbounds1} facilitate calculations of probabilities across a network. This tool is used
throughout the remainder of this note (some examples would be given shortly). Here we explain how \eqref{eq:lowerbounds1} and \eqref{E:BeliefTransition} combine to give a compact way of representing \emph{two entirely different interactions}, one assuming $x\in L$ and the other assuming $x\notin L$. This may be slightly confusing at first, so the reader's attention is called to this point.

Owing to \eqref{E:BeliefTransition}, the probabilities anywhere in the network at time $t>0$ depend on the parameter $h$. We may do all calculations and treat it as a formal parameter. Having in mind that an interaction is ultimately a procedure terminating with the statistics in \eqref{eq:ipdef4}, the parameter $h$ is set either to $c \dep$ or to $s \dep$ depending on whether or not $x$ is in $L$. To verify that a network \emph{decides} $L$ we repeat the computation twice, first for the case where $x \in L$ and then for $x \notin L$. In the former case we will be interested only in the coefficient of $\brak{\mathrm{True}}$ whereas in the latter case we will be interested only in the coefficient of $\brak{\mathrm{False}}$.

Here is an illustrative calculation for a single interaction. Let $W_t=\brak{\mathrm{False}}$ and $V_t=\brak{\mathrm{True}}$. Invoke \eqref{E:BeliefTransition} using the completeness and soundness parameters in \eqref{eq:ipdef4}: first using $h=c\dep$ and then using $h=s\dep$. Hence,
\[
\brak{W_{t+1}} = c\dep\brak{\mathrm{True}} + (1-s\dep)\brak{\mathrm{False}} \longrightarrow \left \{
\begin{array}{ll}
c\dep\brak{\mathrm{True}} + (\cdots) \brak{\mathrm{False}}, & x \in L; \\[0.5ex]
(\cdots) \brak{\mathrm{True}} + (1-s\dep)\brak{\mathrm{False}}, & x \notin L.
\end{array} \right.
\]
This interaction is said to decide $L$ only if $c\dep > \frac{1}{2}$ and $1-s\dep > \frac{1}{2}$.
\end{Remark}

\subsection{Expressive power of a network: The class $\mathrm{BraidIP}$}

Verifiers in our framework are assumed to be implemented as probabilistic polynomial-time Turing machines whose beliefs are either internal states or are stored on tapes. Similarly, an interaction is a polynomial-time procedure. Consider now a crowd of verifiers $V^1,V^2,\ldots, V^\mu$ whose initial beliefs at time $t=0$ are $\brak{V^1_0},\brak{V^2_0},\ldots, \brak{V^\mu_0}$. Allow them to interact at times $t=0,1,2,\ldots,\chi$, subject to parameters $0<c<s\leq 1$ and $\dep\in \mathds{Q}\cap(0,1)$. We write $M$ for this sequence of interactions. A language $L\subseteq \{0,1\}^\ast$ is said to be \emph{decided} by $M$ if $M$ contains a verifier $V$ whose belief at time $\chi$ is $\brak{V_\chi}\ass a\brak{\mathrm{True}}+ b\brak{\mathrm{False}}$, such that for a fixed constant $\kappa>0$, we have:

\begin{equation}
\label{eq:lang}
\begin{array}{l}
x \in L \longrightarrow a \geq \frac{1}{2} + \abs{x}^{-\kappa}; \\[0.5ex]
x \notin L \longrightarrow b \geq \frac{1}{2} + \abs{x}^{-\kappa}.
\end{array}
\end{equation}

This definition depends on the choice of $\kappa>0$. The class \emph{braided interactive polynomial time} ($\mathrm{BraidIP}$) consists of those languages which are decidable for any fixed $\kappa>0$ by some network $M$ in time $\chi$, polynomial in $\abs{x}$. We denote this class $\mathrm{BraidIP}\{\dep,\chi\}$ where $\chi$ is \emph{the number of interactions in $M$}.

\begin{Remark}
The definition of class $\mathrm{BraidIP}$ is similar in spirit to the class $\mathrm{BPP}$. We will see below that it includes class $\mathrm{IP}$. Letting \eqref{eq:lang} reflect the class $\mathrm{PP}$ (\textit{i.e.} taking strict inequalities and right-hand constants equal to $\frac{1}{2}$) will result in networks that decide any $L \in \mathrm{IPP}$.
\end{Remark}

\subsection{Braid of beliefs}

Our networks admit a convenient diagrammatic description. We represent an interaction as a wire cutting through other wires (see Figure~\ref{fig:inter}). The overcrossing wire, which becomes slightly thickened in an interaction, carries the belief statistics of an agent whereas the undercrossing wires carry the belief statistics of her patients. An example of many concatenated interactions is given in Figure~\ref{F:rumour2}.

\begin{figure}[htb]
\centering
\psfrag{b}[c]{\small $\brak{W_t}$}
\psfrag{d}[c]{\small $\brak{W_{t+1}}$}
\psfrag{a}[r]{\small $\brak{V_{t}}$}
\psfrag{c}[l]{\small $\brak{V_{t+1}}$}
\includegraphics[width=0.17\textwidth]{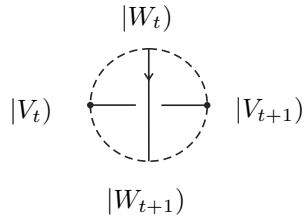}
\caption{\small Diagrammatic representation of an interaction.}
\label{fig:inter}
\end{figure}

\begin{figure}[htb]
\centering
\psfrag{x}[c]{$D_0$}
\psfrag{v}[c]{$B_0$}
\psfrag{u}[c]{$A_0$}
\psfrag{w}[c]{$C_0$}
\psfrag{0}[c]{\small $0$}
\psfrag{1}[c]{\small $1$}
\psfrag{2}[c]{\small $2$}
\psfrag{3}[c]{\small $3$}
\psfrag{4}[c]{\small $4$}
\psfrag{a}[c]{\small $_{\brak{D_0}}$}
\psfrag{b}[c]{\small $_{\brak{C_0}}$}
\psfrag{a1}[c]{\small $D_2$}
\psfrag{b1}[c]{\small $C_2$}
\psfrag{m}[c]{\small $_{\brak{D_2}}$}
\psfrag{n}[c]{\small $_{\brak{C_2}}$}
\psfrag{c}[c]{\small $_{\brak{B_0}}$}
\psfrag{d}[c]{\small $_{\brak{A_0}}$}
\psfrag{e}[c]{\small $_{\brak{A_3}}$}
\psfrag{f}[c]{\small $_{\brak{B_3}}$}
\psfrag{g}[c]{\small $_{\brak{C_3}}$}
\psfrag{h}[c]{\small $_{\brak{D_3}}$}
\psfrag{y}[c]{\small \emph{Tangle machine}}
\psfrag{q}[c]{\small \emph{Rumour-passing network}}
\psfrag{r}[c]{\small \emph{(flow of beliefs)}}
\psfrag{t}[c]{\small \emph{time}}
\includegraphics[width=0.96\textwidth]{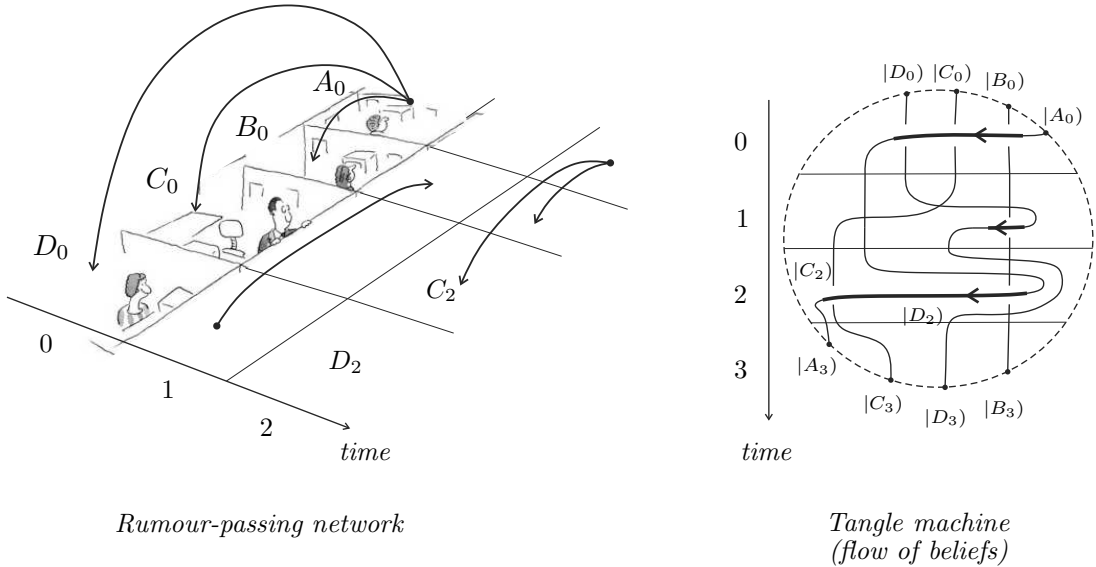}
\caption{\small A rumour-passing network and its respective diagram.}
\label{F:rumour2}
\end{figure}

The diagram representing the rumour passing network in Figure~\ref{F:rumour2} is a \emph{braid}, which is a special sort of \emph{tangle}. Here, such diagrams represent the flow of beliefs within a network of interacting machines (verifiers).

When many patients pass under a single agent, we define this to imply that for each patient, the belief sampled from that patient is \emph{independent} of the belief sampled from every other patient, and the belief sampled from the agent to update that patient's belief is independent of the belief sampled from the agent to update every other agent's belief. To say the same thing in a different way, multiple patients under the same agent are independent and unsynchronized during that time-frame. We will discuss this point further in Section~\ref{SS:reversibility}.

\subsection{An example}

The capacity of a network to prove or disprove a claim is an emergent property. Out of a number of uncertain interactions, none of which prove the claim, the truth may eventually materialize. As an example, consider the two pairs of two consecutive interactions pictured in Figure~\ref{fig:pcpm}. Both sequences involve three verifiers, designated $X$, $Y$ and $Z$. Their initial beliefs are shown at the bottom. 

\begin{figure}[htb]
\centering
\psfrag{a}[c]{\small $\rule{0pt}{12pt}_{\brak{Y_0} = \frac{1}{2}\brak{\mathrm{False}}+\frac{1}{2}\brak{\mathrm{True}}}$}
\psfrag{b}[c]{\small $_{\brak{X_0} = \brak{\mathrm{False}}}$}
\psfrag{c}[c]{\small $_{\brak{Z_0} = \brak{\mathrm{True}}}$}
\psfrag{d}[c]{\small $_{\brak{Y_2}}$}
\psfrag{f}[c]{\small $_{\brak{X_2}}$}
\psfrag{e}[c]{\small $_{\brak{X_1}}$}
\psfrag{e1}[c]{\small $_{\brak{\bar{X}_1}}$}
\psfrag{g}{}
\psfrag{t}[c]{\small \emph{time}}
\includegraphics[width=0.7\textwidth]{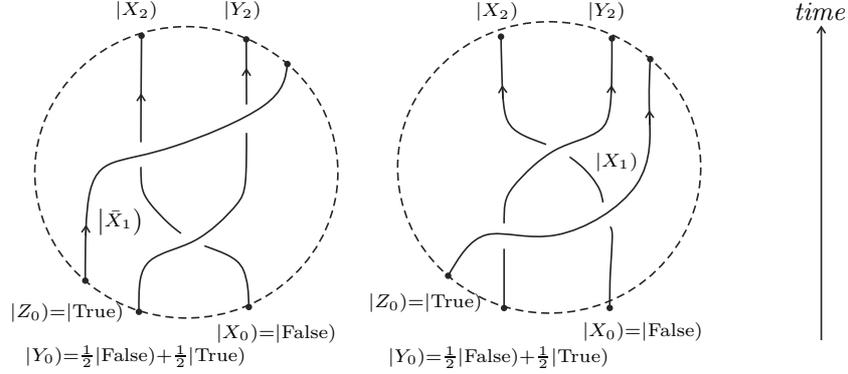}\rule{0pt}{100pt}
\caption{\small Equivalent networks of interactions.}
\label{fig:pcpm}
\end{figure}

Set the parameters to $s\ass \frac{1}{2}$, $c\ass 1$, and $\dep\ass \frac{1}{2}$. Allow the verifiers to interact in the left diagram. At time $t=2$, beliefs $X_0$ and $Y_0$ of $X$ and of $Y$ have been updated to $X_2$ and to $Y_2$ ($Z$ does not change his belief). The distributions are:
\begin{equation}
\brak{X_2}=\left(\brak{X_0}^{\brak{Y_0}}\right)^{\brak{Z_0}} \longrightarrow
\left \{ \begin{array}{ll}
\frac{3}{8} \brak{\mathrm{False}} + \frac{5}{8} \brak{\mathrm{True}}, & x \in L; \\[0.5ex]
\frac{21}{32}\brak{\mathrm{False}} + \frac{11}{32}\brak{\mathrm{True}}, & x \notin L.
\end{array} \right.
\end{equation}
Hence,
\begin{equation}
\begin{array}{lll}
\brak{X_2} = \left(\brak{X_0}^{\brak{Y_0}}\right)^{\brak{Z_0}}  & =
\frac{21}{32}\brak{\mathrm{False}} +  \frac{5}{8} \brak{\mathrm{True}}, & \text{(left network)}; \\[1ex]
\brak{X_2} = \left(\brak{X_0}^{\brak{Z_0}}\right)^{\left(\brak{Y_0}^{\brak{Z_0}}\right)} & =
\frac{21}{32}\brak{\mathrm{False}} +  \frac{5}{8} \brak{\mathrm{True}}, & \text{(right network)}.
\end{array}
\end{equation}

Similarly,

\begin{equation}
\brak{Y_2} = \brak{Y_1} = \brak{Y_0}^{\brak{Z_0}} =
\frac{3}{8}\brak{\mathrm{False}}+ \frac{3}{4} \brak{\mathrm{True}}.
\end{equation}

From this we see that $X$ decides correctly at time $\chi=2$ with probability at least $\frac{5}{8}$ or $\frac{21}{32}$, depending on whether
$x \in L$ or $x\notin L$. Both of these are greater that $\frac{1}{2}$, so the pair of inequalities~\ref{eq:lang}, for a suitable $\kappa>0$, a protocol underlied by the above interactions will succeed in deciding, at $\brak{X_2}$, whether or not $x\in L$.

Note again that
\[\brak{Y_1}=\brak{Y_0}^{\brak{Z_0}}= h\brak{Z_0}+(1-h)\brak{Y_0},\]
and we evaluate with $h= \frac{1}{2}$ for the coefficient of $\brak{\mathrm{True}}$ and with $h= \frac{1}{4}$ for the coefficient of $\brak{\mathrm{False}}$. This is the same for all interactions in both diagrams.

Note that both diagrams in Figure~\ref{fig:pcpm} have the same \emph{initial beliefs} $\brak{X_0}$, $\brak{Y_0}$, and $\brak{Z_0}$ and the same \emph{terminal beliefs} $\brak{X_2}$, $\brak{Y_2}$, and $\brak{Z_2}$, and differ only in the belief of $X$ at time $t=1$. Thus, these two diagrams underlie \emph{equivalent} deformed interactive proofs, each of which can uniquely be reconstructed from the other, which decide the same languages, but which differ at an intermediate step. This equivalence is the topic of Section~\ref{sec:lowdim}.

\section{Deformation of an IP system}
\label{sec:defip}

In this section, we show how we may deform an IP system with any soundness parameters $0<s<\frac{1}{2}<c\leq 1$, for any deformation parameter $\dep\in \mathds{Q}\cap(0,1)$. The completeness and soundness parameters of the deformed system will be $s\dep$ and $c\dep$ correspondingly. The deformation parameter $\dep$ serves to introduce noise between the prover and the verifiers. In the $\dep\to 1$ limit we recover $\mathrm{IP}$, and the information obtained by a verifier at each interaction shrinks as $\dep\to 0$. But Theorem~\ref{thm:existence} proves that we can recover $\mathrm{IP}$ from $\mathrm{BraidIP}$ by concatenating many consecutive deformed interactions.

Before describing how IP may be deformed, we outline the major differences between a single interaction in IP and BraidIP:\\[0.5ex]
\[
\small
\begin{tabular}{p{3cm} | p{4cm} | p{4cm}}
Description & IP system & Deformed IP system \\
\hline
Participants & Verifier, Prover & Many verifiers (patient), Verifier (agent), Prover \\
\hline
Verifier 'state of mind' & Accept/Reject & Belief True/False \\
\hline
Conclusion & Verifier decides Accept/Reject & If the two verifiers do not agree then the patient may change her belief. \\
\hline
Completeness, Soundness & $c,s$ & $c\dep,s\dep$
\end{tabular}
\]

\subsection{Two approaches to deform IP}

We present two approaches to deform an $\mathrm{IP}$ protocol. The end result is the same, but the `story' is different.

\subsubsection{Agent and patient as a single verifier}

We may think of a patient $W$ and an agent $V$ of an interaction at time $t$ as representing different aspects of a single verifier. In this approach we conceive of $W$ and $V$ as being a single unit $(W,V)$. The verifier $(W,V)$ transmits to the prover $\Pi$ the belief of both $W$ and $V$. We may imagine $W$ and $V$ as litigants in a court case, presenting their claims to the judge $\Pi$, where $W$ is the defendant and $V$ is the plaintiff.  If both $W$ and $V$ make the same claim, then $\Pi$ throws the case out (\textit{i.e.} $W_{t+1}=w_t$ and $V_{t+1}=v_t$). On the other hand, if $V$ disagrees with $W$, then $(W,V)$ query the prover $\Pi$ according to the original interactive protocol. If according to the original protocol, $W$'s claim should be accepted, then $\Pi$ rules in $W$'s favour (\textit{i.e.} $W_{t+1}=w_t$ and $V_{t+1}=v_t$). But if according to the original protocol $W$'s claim should be rejected and $V$'s claim should be accepted, then $\Pi$ picks an integer uniformly at random between $1$ and $N$. If the number $\Pi$ picked is less than $\dep N$, then $\Pi$ rules in favour of $V$ (\textit{i.e.} $W_{t+1}=v_t$ and $V_{t+1}=v_t$). Otherwise he rules in favour of $W$.

Perhaps $\dep$ represents a chosen standard of `reasonable doubt'. Constants $s$ and $c$ perhaps represent constants associated with the mechanics of the courthouse procedure. Note that as $\dep\to 1$, a single interaction involving two verifiers with opposite beliefs recovers IP.

\subsubsection{Verifiers communicating through a noisy channel}

The following approach has an information-theoretic interpretation. Consider the prover $\Pi$ as an information source, the agent verifier $V$ as an encoder, and the patient verifier $W$ as a decoder. The query information transmitted from $W$ to $\Pi$ is relayed via a perfect communication channel (\textit{i.e.} there is no loss of information in this direction). The replies from $\Pi$ are passed on to $V$ who encodes them and transmits them back to $W$, this time through a noisy channel. This means that the prover replies emerge corrupted on $W$'s end, which consequently influences her decision.

Introducing a noisy channel into the formalism restricts the information obtained by the patient verifier from the prover. It tempting to state that the combination prover-agent-noisy channel behaves like a mendacious agent, in that the agent $V$ decided whether to `tell the truth' or to `lie' to $W$. But to think of $V$ as a mendacious agent is inaccurate. For one thing, the agent's strategy whether or not to reliably relay $\Pi$'s replies to $W$ must account for the beliefs of both $V$ and $W$, either one of which may not be correct. Her behavior does not stem from her being more knowledgeable; rather we may think of it as a manifestation of her own beliefs.

It is important to note that although $V$ receives the replies from $\Pi$ she is not allowed to use this information to update her own belief. One can think of protocols taking advantage of the fact that $V$ is not aware of $W$'s queries wherein this restriction follows naturally. For now it is enough to assume that $V$ will not use the prover replies for her own benefit.

Here is how such a protocol may run. Upon disagreement between $W$ and $V$, \textit{i.e.} $v_t \neq w_t$, the patient $W$ sends her queries to $\Pi$. The replies to $W$'s queries are then sent by $\Pi$ to the agent $V$. At this point, $V$, who is a verifier much like $W$, runs her own verification test on the prover replies. He obtains $\xi_t=1$ for accept/true and $\xi_t=-1$ for reject/false. In case where $\xi_t=1$ she tampers with the prover replies such that when they are received by $W$ her verification would indicate $v_t$ with probability $\dep$. In case where $\xi_t=-1$ the agent $V$ tampers with the prover replies such that the test of $W$ would indicate $\neg v_t$.

Implicit in the above protocol is the fact that the capacity of $V$ to deceive $W$ is limited by $V$'s own belief. If her belief, $v_t$, coincides with the true nature of claim then she may potentially have more power to deceive $W$. 

\begin{figure}
\centering
\psfrag{z}[c]{\small $W_{t+1}$}
\psfrag{y}[c]{\small $V$}
\psfrag{o}[c]{\small $\Pi$}
\psfrag{t}[c]{\small $_{\brak{\mathrm{True}}}$}
\psfrag{f}[c]{\small $_{\brak{\mathrm{False}}}$}
\psfrag{a}[c]{\small $_{+1}$}
\psfrag{b}[c]{\small $_{-1}$}
\psfrag{1}[c]{\small $1$}
\psfrag{p}[c]{\small $\dep$}
\psfrag{q}[c]{\small $1-\dep$}
\includegraphics[width=0.9\textwidth]{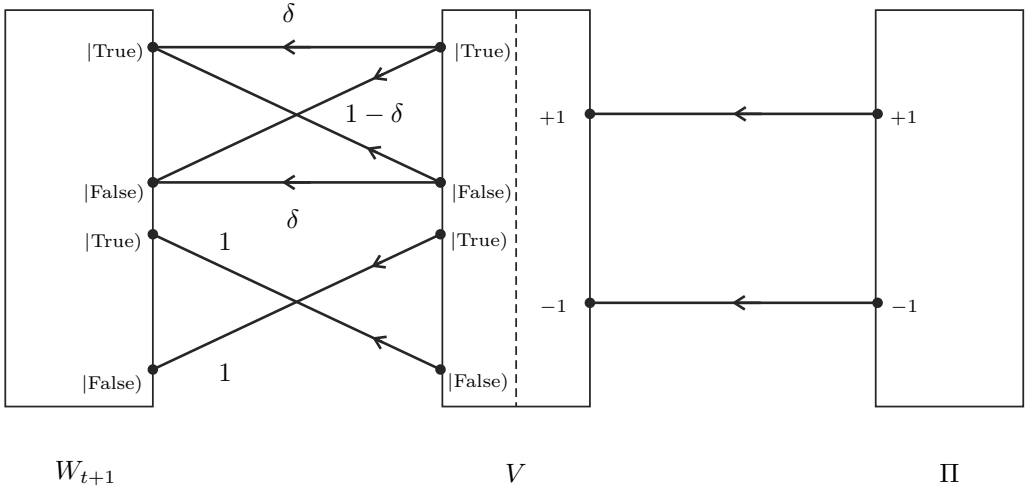}
\caption{\small Communication between $W$ and $V$ when $v_t \neq w_t$. The channels between $W$ and $V$ are symmetric. The labels above edges indicate transition probabilities. The box of $V$ is divided into two sections representing her belief $v_t$ (left section) and the outcome of his verification $\xi_t$ (right section).}
\label{fig:channel}
\end{figure}

The protocol just described underlies a noisy symmetric channel between $W$ and $V$. If $\xi_t=1$, this channel is characterized by $\dep$ and has a capacity of $1-H_2(\dep)$, where $H_2(\dep)$ is the Shannon entropy of a Bernoulli random variable with parameter $\dep$. It induces a maximal loss of $1$ bit of information for $\dep=\frac{1}{2}$. An illustration of the communication between the three parties patient-agent-prover is given in Figure~\ref{fig:channel}.

\subsubsection{Further metaphors for deformed interactions}

The agents in our picture all receive messages from the same oracle. This essentially suggests that a network of interactions is a construct \emph{quantizing} the oracle knowledge. At every location within the network only a quanta of this knowledge is used by way of interaction between a patient and an agent. Later on we will show that although a single interaction may be limited in its capacity to prove the claim, the proof may yet emerge somewhere in the network depending on its topology.  

The triple patients-agent-oracle brings to mind some basic models of reasoning and information transfer. Perhaps a patient is an entity whose beliefs reflect both prior knowledge and observations. The patient is exposed to a genuine phenomenon, which the patient has not seen before. The phenomenon, which is the metaphor for an oracle, is beyond the comprehension of the patient and hence a number of observations are collected in an attempt to reach a definitive conclusion.  These observations, however, may be distorted by limitations of the patient measuring apparatus, or perhaps they contradict prevailing explanations and beliefs. In either cases observations contain, or otherwise introduce, uncertainty. Observations are the metaphor for agents. What the patient tries to accomplish underlies the Bayesian inference paradigm. 

Here is another metaphor. A patient is a decoder, an oracle is an information source, and an agent is an encoder who relays the encoded oracle message through a noisy communication channel. Alternatively, an agent-patient pair is a verifier and the oracle is a prover who relays a message through a noisy communication channel. All metaphors reflect knowledge transfer subject to uncertainties.

\subsection{Probabilistic theorem proving in networks}

The goal of this section is to prove Theorem~\ref{thm:existence}, repeated below for convenience.

\begin{Theorem}
$\mathrm{IP} \subseteq \mathrm{BraidIP} \left\{ \dep, \chi \right\}$ where:
\begin{equation}
I(c\dep) < \chi  < \frac{1}{I(1-s\dep)},
\end{equation}
with $I(p) \ass -p^{-1} \log p$. The growth rate of $\chi$ is $\mathcal{O}(\frac{1}{\dep})$ as $\dep\to 0$.
\end{Theorem}

\begin{proof}[Proof of {Theorem~\ref{thm:existence}}]
We explicitly construct a configuration of interactions which decide a language $L$ in $\mathrm{IP}$. This configuration, which is illustrated in Figure~\ref{fig:pm} for the case $\chi=4$ is a scaled up version of that in Figure~\ref{fig:pcpm}. It involves $\chi+1$ verifiers $W$ and $V^1,V^2,\ldots,V^\chi$, and $\chi$ interactions. The parameters of all interactions are the same, and are $c,s,\dep$.

\begin{figure}[htb]
\centering
\psfrag{a}[c]{\small $_{\brak{W_0}}$}
\psfrag{b}[c]{\small $_{\brak{V_0^1}}$}
\psfrag{c}[c]{\small $_{\brak{V_0^2}}$}
\psfrag{d}[c]{\small $_{\brak{V_0^3}}$}
\psfrag{e}[c]{\small $_{\brak{V_0^4}}$}
\psfrag{g}[c]{\small $_{\brak{W_4}}$}
\includegraphics[width=0.3\textwidth]{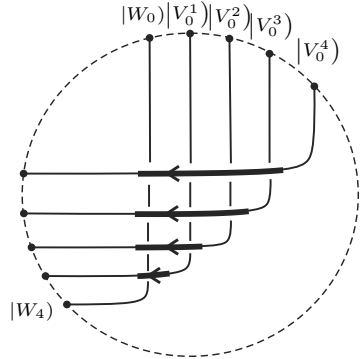}
\caption{\small An interactive $\mathrm{BraidIP}$ theorem proving network with $5$ verifiers $W$, $V^1$, $V^2$, $V^3$, and $V^4$, and $4$ interactions.}
\label{fig:pm}
\end{figure}

Let $L \in \mathrm{IP}$. The initial beliefs at time $t=0$ are set to $\brak{W_0}\ass \brak{\mathrm{False}}$ and
\begin{equation}
\brak{V_0^i} \ass \left \{
\begin{array}{ll}
\frac{1}{2} \brak{\mathrm{False}} + \frac{1}{2} \brak{\mathrm{True}}, & 1 < i < \chi; \\
\brak{\mathrm{True}}, & i=\chi.
\end{array} \right.
\end{equation}

Thus, at time zero there are two verifiers with opposite beliefs and $\chi-1$ verifiers whose initial belief is that the claim $x\in L$ is 50\% true and 50\% false.

Calculating the output statistic $\brak{W_\chi}$ (which occurs at time $\chi$) yields
\begin{equation}
\label{eq:l1}
\brak{W_\chi} \longrightarrow \left \{
\begin{array}{ll}
(1-c\dep)^\chi \brak{\mathrm{False}} + \left[ 1 - c\dep - (1-c\dep)^\chi \right] \left( \frac{1}{2} \brak{\mathrm{False}} + \frac{1}{2} \brak{\mathrm{True}} \right) + c\dep \brak{\mathrm{True}}, & x \in L; \\[0.5ex]
(1-s\dep)^\chi \brak{\mathrm{False}} + \left[ 1 - s\dep - (1-s\dep)^\chi \right] \left( \frac{1}{2} \brak{\mathrm{False}} + \frac{1}{2} \brak{\mathrm{True}} \right) + s\dep \brak{\mathrm{True}}, & x \notin L.
\end{array} \right.
\end{equation}
From \eqref{eq:l1} we see that the configuration in Figure~\ref{fig:pm} decides $L$ if and only if:
\begin{equation}
\label{eq:l2}
\begin{array}{l}
x \in L \longrightarrow (1-c\dep)^\chi < c\dep; \\[0.5ex]
x \notin L \longrightarrow (1-s\dep)^\chi > s\dep.
\end{array}
\end{equation}
From here we obtain the following bounds for $\chi$:
\begin{equation}
\label{eq:gg}
\frac{\log (c\dep)}{\log (1-c\dep)} < \chi < \frac{\log (s\dep)}{\log (1-s\dep)}.
\end{equation}
Equation~\ref{eq:bounds} follows upon noting that $\log(1-p) < -p$ for any $p \in (0,1)$. Therefore:
\begin{equation}
-\frac{1-p}{\log(1-p)} > \frac{\log p}{\log (1-p)} > - \frac{\log p}{p}.
\end{equation}
For $\chi$ within these bounds, the above configuration decides $L$.
\end{proof}

\begin{Remark}
Equation~\eqref{eq:bounds} tells us that $\chi$ has approximately the same growth rate in $\abs{x}$ as $\frac{1}{\delta}$. By definition of $\mathrm{BraidIP}$, $\chi$'s growth rate is polynomial in the word length $\abs{x}$, and so therefore $\dep$ is asymptotically bounded below by approximately one over a polynomial in $\abs{x}$.
\end{Remark}

\section{Efficient IP strategies: Tangled IP}
\label{sec:efficientip}

\subsection{The complexity class $\mathrm{TangIP}$}\label{SS:TangIP}

We may extend class BraidIP by allowing each verifier to have its own `local' time parameter, so that a patient belief $W_t$ may interact with an agent belief $V_s$ for $s\neq t$, and become updated to $W_{t+1}$.
\[
\psfrag{b}[c]{\small $\brak{V_s}$}
\psfrag{a}[c]{\small $\brak{W_t}$}
\psfrag{c}[l]{\small $\brak{W_{t+1}}$}
\includegraphics[width=0.17\textwidth]{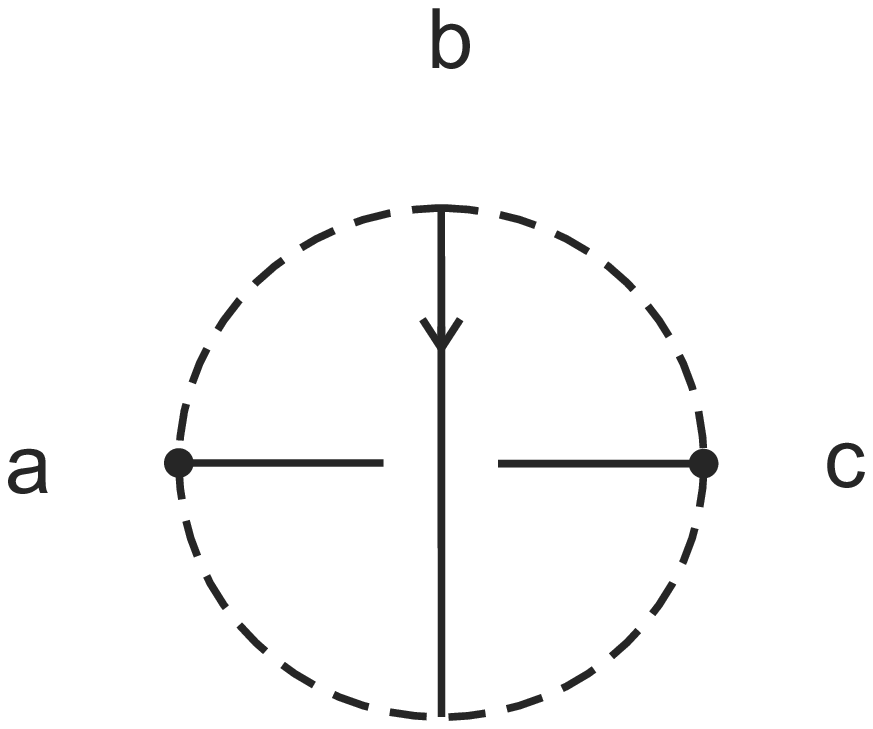}
\]
We also allow verifiers to travel backwards or forwards in time and to update their previous or future beliefs, so that $V_t$ may be updated by agent $V_s$ to become $V_{t+1}$, where $V_t$ and $V_s$ are beliefs of one and the same verifier. Thus, each verifier may update their beliefs with past or future beliefs of itself (via feedback loops) or of other verifiers. We write $M$ for a network of such concatenated interactions, subject to parameters $0<c<s\leq 1$ and $\dep\in \mathds{Q}\cap(0,1)$. An example of such a network is given in Section~\ref{sec:hopf}.

A language $L\subseteq \{0,1\}^\ast$ is said to be \emph{decided} by $M$ if $M$ contains a verifier $V$ whose belief at time $\chi$ is $\brak{V_\chi}\ass a\brak{\mathrm{True}}+ b\brak{\mathrm{False}}$, such that for a fixed constant $\kappa>0$ the inequalities~\eqref{eq:lang} are satisfied.

The class \emph{tangled interactive polynomial time} ($\mathrm{TangIP}$) consists of those languages which are decidable for any fixed $\kappa>0$ by some network $M$ which contains $\chi$ interactions, where $\chi$ is polynomial in $\abs{x}$. We denote this class $\mathrm{TangIP}\{\dep,\chi\}$. 

By Theorem~\ref{thm:existence} we know that
\[\mathrm{IP}\subseteq \mathrm{BraidIP}\subseteq \mathrm{TangIP}.\]

We wonder about the connection between our classes $\mathrm{BraidIP}$ and $\mathrm{TangIP}$ and multi-prover $\mathrm{IP}$ ($\mathrm{MIP}$)~\citep{BenOr:88}. In particular, we wonder whether $\mathrm{MIP}\subseteq \mathrm{BraidIP}$ or $\mathrm{MIP}\subseteq \mathrm{TangIP}$, particularly if we allow different interactions to have different parameters (in this note all interactions are required to have the same parameters because that's all we need, but there is no obstruction to considering the more general case).

\subsection{The Hopf--Chernoff configuration}
\label{sec:hopf}

Consider an $\mathrm{IP}$ system whose soundness $s$ is nearly equal to its completeness $c$ but for a small constant $\epsilon(\abs{x})$ that depends on the word length $\abs{x}$, \textit{i.e.} $c-s = \epsilon(\abs{x})$. As $\epsilon(\abs{x}) \to 0$, the $\mathrm{IP}$ system becomes inefficient in the sense that it accepts every word with probability nearly $c$ regardless of its membership in $L$. Yet we can still construct a deformed $\mathrm{IP}$ system that decides $L$. One may wonder how the number of interactions in such a system is affected by the decreasing gap $\epsilon(\abs{x})$.

 The number of interactions in a network of the form given in Figure~\ref{fig:pm} is implicit in Theorem~\ref{thm:existence}. Fix $\dep\in \mathds{Q}\cap(0,1)$ and note from Equation~\eqref{eq:bounds} that, as $\epsilon(\abs{x})$ decreases, the values bounding $\chi$ become nearly identical. But $\chi$ is an integer, so the two bounds must have an integer between them. In general, that means that the distance between them is at least $1$.  Therefore, the number of interactions grows as $\epsilon(\abs{x})$ decreases. We may need to take $\dep \to 0$ as $\epsilon(\abs{x}) \to 0$. In fact it can be shown that in this case $\dep(\abs{x}) = \mathcal{O}(\epsilon(\abs{x}))$ which means that we require $\chi=\mathcal{O}(1/\epsilon(\abs{x}))$ interactions.

Can fewer interactions decide $L$ ?  The configuration in Figure~\ref{fig:hopfm} decides $L$ for any $\epsilon(\abs{x})$ using substantially less than $\mathcal{O}(1/\epsilon(\abs{x}))$ interactions. This machine is a concatenation of a number of identical smaller configurations of interactions, denoted $M_0,M_1,M_2,\ldots$. When concatenated to form a single configuration, we require approximately $\chi=\mathcal{O}(\log(1/\epsilon(\abs{x})))$ copies of $M_0$ to decide $L$ (the precise argument is given below). If we were to trace its colours (probability generating functions) we would notice that it behaves much like a repetition of a binary random experiment (\textit{e.g.} coin flipping), hence the magnitude of $\chi$. We have named this configuration the \emph{Hopf--Chernoff configuration} suggesting both to its structure and, to some extent, its functionality.

\begin{figure}[htb]
\centering
\psfrag{u}[c]{\small $\mathrm{In}_1^i$}
\psfrag{v}[c]{\small $\mathrm{In}_2^i$}
\psfrag{x}[c]{\small $\mathrm{In}_1^0$}
\psfrag{y}[c]{\small $\mathrm{In}_2^0$}
\psfrag{1}[c]{\small $_{\brak{\alpha}}$}
\psfrag{2}[c]{\small $_{\brak{\beta}}$}
\psfrag{g}[l]{\small $\mathrm{Out}_1^i$}
\psfrag{h}[l]{\small $\mathrm{Out}_2^i$}
\psfrag{f}[c]{\small $M_i$}
\psfrag{m}[c]{\small $M_0$}
\psfrag{n}[c]{\small $M_1$}
\psfrag{k}[c]{\small $M_{\chi}$}
\psfrag{o}[l]{\small $\mathrm{Out}_1^{\chi}$}
\includegraphics[width=0.8\textwidth]{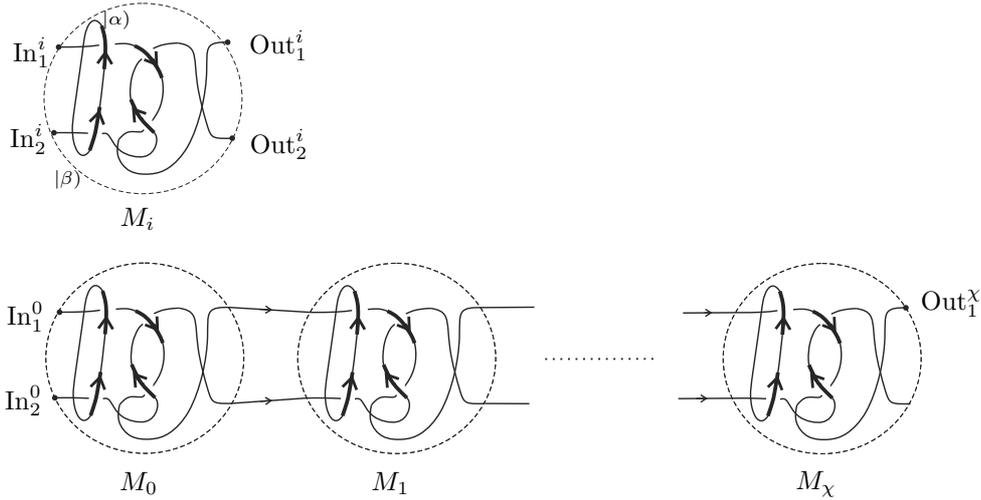}
\caption{\small Hopf--Chernoff configuration(open version).}
\label{fig:hopfm}
\end{figure}

\begin{Theorem}[Hopf--Chernoff configuration]
\label{thm:hopfm}
Consider the configuration of interactions in Figure~\ref{fig:hopfm} which underlies a deformed $\mathrm{IP}$ system with completeness
$c \dep$ and soundness $(c-\epsilon) \dep$, where $\epsilon > 0$. There exists a pair of beliefs, $\alpha$ and $\beta$, independent of any other belief in the machine, such that for any given set of initial beliefs $\mathrm{In}_j^0$ the machine decides any $L \in \mathrm{IP}$ using $\chi=\mathcal{O}(\log(1/\epsilon))$ submachines (and $4\chi$ interactions). In particular, letting
\begin{equation}
\label{eq:beliefsab}
\begin{array}{l}
\brak{\alpha} = \left( \frac{1}{4} + \frac{1}{12} \epsilon \dep \right) \brak{\mathrm{True}} + \left( \frac{3}{4} - \frac{1}{12} \epsilon \dep \right) \brak{\mathrm{False}}; \\[0.5ex]
\brak{\beta} = \left( 1 - \frac{1}{2} c \dep + \frac{1}{12} \epsilon \dep \right) \brak{\mathrm{True}} + \left( \frac{1}{2} c \dep - \frac{1}{12} \epsilon \dep \right) \brak{\mathrm{False}}.
\end{array}
\end{equation}
yields
\begin{equation}
\label{eq:stat}
 \mathrm{Out}_1^\chi \longrightarrow \left \{
\begin{array}{ll}
\left[\frac{1}{2} + \frac{1}{12}\epsilon\dep\right] \brak{\mathrm{True}}+\left[\frac{1}{2} - \frac{1}{12}\epsilon\dep\right]\brak{\mathrm{False}}, & x \in L; \\[0.5ex]
\left[\frac{1}{2} - \frac{1}{12}\epsilon\dep\right] \brak{\mathrm{True}}+\left[\frac{1}{2} + \frac{1}{12}\epsilon\dep\right]\brak{\mathrm{False}}, & x \notin L.
\end{array} \right.
\end{equation}
namely, $\mathrm{Out}_1^\chi = \left[\frac{1}{2} + \frac{1}{12}\epsilon\dep\right] \brak{\mathrm{True}} + \left[\frac{1}{2} + \frac{1}{12}\epsilon\dep\right] \brak{\mathrm{False}}$.
\end{Theorem}

\begin{proof}
Let us begin by writing down the relations between the outputs $\mathrm{Out}_j$ and inputs $\mathrm{In}_j$ of this machine.
Note that
\begin{equation}
\label{eq:ff}
\mathrm{Out}_1 = \left(\mathrm{In}_1^{\brak{\alpha}}\right)^{\left(\mathrm{In}_2^{\brak{\beta}}\right)}, \quad \mathrm{Out}_2 = \left(\mathrm{In}_2^{\brak{\beta}}\right)^{\left(\mathrm{In}_1^{\brak{\alpha}}\right)}.
\end{equation}
Explicitly writing \eqref{eq:ff} using the formal parameter $h$ yields
\begin{equation}
\label{eq:dyn}
\begin{bmatrix}
\mathrm{Out}^i_1 \\
\mathrm{Out}^i_2
\end{bmatrix} =
\underbrace{\begin{bmatrix}
(1-h)^2 & h(1-h) \\
h(1-h) & (1-h)^2
\end{bmatrix}}_{\ass A(h)}
\begin{bmatrix}
\mathrm{In}^i_1 \\
\mathrm{In}^i_2
\end{bmatrix} +
\underbrace{\begin{bmatrix}
h(1-h) & h^2 \\
h^2 & h(1-h)
\end{bmatrix}}_{\ass B(h)}
\begin{bmatrix}
\brak{\alpha} \\
\brak{\beta}
\end{bmatrix}.
\end{equation}

Letting $\mathrm{In}_j^{i+1} = \mathrm{Out}_j^i$, $j=1,2$, equation \eqref{eq:dyn} underlies a linear dynamical system. It is easy to verify
that the eigenvalues of the transition matrix $A(h)$ all are within the unit circle, i.e. $\abs{\lambda(A(h))} < 1$ for any $h > 0$. That means
that the system \eqref{eq:dyn} reaches a steady-state as $i \to \infty$. The steady-state can be obtained as follows. Rewrite
\eqref{eq:dyn} as
\begin{equation}
\begin{bmatrix}
\mathrm{Out}_1 \\
\mathrm{Out}_2
\end{bmatrix} = A(h)
\begin{bmatrix}
\mathrm{Out}_1 \\
\mathrm{Out}_2
\end{bmatrix} + B(h)
\begin{bmatrix}
\brak{\alpha} \\
\brak{\beta}
\end{bmatrix},
\end{equation}
and solve for $\mathrm{Out}_1$ and $\mathrm{Out}_2$. Thus,
\begin{equation}
\label{eq:gg1}
\begin{bmatrix}
\mathrm{Out}_1 \\
\mathrm{Out}_2
\end{bmatrix} =
(I-A(h))^{-1} B(h)
\begin{bmatrix}
\brak{\alpha} \\
\brak{\beta}
\end{bmatrix} =
\frac{1}{3-2h}
\begin{bmatrix}
2(1-h) \brak{\alpha} + \brak{\beta} \\
\brak{\alpha} + 2(1-h) \brak{\beta}
\end{bmatrix}.
\end{equation}

Define
\begin{equation}
\label{eq:ab}
\begin{array}{l}
\brak{\alpha} \ass a \brak{\mathrm{True}} + (1-a) \brak{\mathrm{False}}; \\[0.5ex]
\brak{\beta} \ass b \brak{\mathrm{True}} + (1-b) \brak{\mathrm{False}}.
\end{array}
\end{equation}
For the network to decide $L$ we require the steady-state of $\mathrm{Out}_1$ to satisfy
\begin{equation}
\label{eq:og}
\mathrm{Out}_1 \longrightarrow \left \{
\begin{array}{ll}
(\frac{1}{2} + \sigma)\brak{\mathrm{True}} + (\frac{1}{2}-\sigma) \brak{\mathrm{False}}, & x \in L; \\
\rule{0pt}{12pt}(\frac{1}{2} - \sigma)\brak{\mathrm{True}} + (\frac{1}{2}+\sigma) \brak{\mathrm{False}}, & x \notin L.
\end{array} \right.
\end{equation}
for some $\sigma > 0$. Using both \eqref{eq:gg1} and \eqref{eq:ab} this requirement translates into the following set of equations
\begin{equation}
\label{eq:cc}
\begin{array}{ll}
(3-2c\dep)\left( \frac{1}{2}+\sigma \right) = 2(1-c\dep) a + b, & x \in L; \\[0.5ex]
(3-2(c-\epsilon)\dep)\left( \frac{1}{2}-\sigma \right) = 2(1-(c-\epsilon)\dep) a + b, & x \notin L.
\end{array}
\end{equation}
where the fact that $h=c\dep$ for $x \in L$ and $h=(c-\epsilon)\dep$ for $x \notin L$ has been used. Solving \eqref{eq:cc}
for the coefficients $a$ and $b$ while assuming $\sigma = \frac{1}{12} \epsilon \dep$ yields \eqref{eq:beliefsab}.
The underlying output probabilities in \eqref{eq:stat} are given by \eqref{eq:og}.

To complete the argument we need to show that the network converges within the stated number of iterations. It is sufficient
to consider the case where the output probabilities \eqref{eq:og} are attained to within the order $\mathcal{O}(\sigma) = \mathcal{O}(\epsilon)$.
Growth rate of the system \eqref{eq:dyn} is linear in $\abs{\lambda_1(A(h))}^\chi$ where $\lambda_1(A(h))$ denotes the largest eigenvalue
of $A(h)$. Simple calculation shows that $\lambda_1(A(h)) = 1-h$ which yields $\chi = \mathcal{O}(\log(1/\epsilon))$.
\end{proof}

The Hopf--Chernoff configuration is a recursive structure which is guaranteed to converge irrespective of its initial beliefs $\mathrm{In}_j^0$.
In fact it represents a two-dimensional homogeneous irreducible Markov chain whose rate of convergence is $\mathcal{O}(2^{-\chi})$. Its stationary distribution, which depends on whether $x\in L$ or $x\notin L$, is given by \eqref{eq:stat}. By virtue of its convergence properties we may just let it run forever (\textit{i.e.} $\chi \to \infty$) knowing that it will eventually reach a stationary distribution not far from \eqref{eq:stat}. For that reason we may as well substitute the open network in Figure~\ref{fig:hopfm} with its closed counterpart in Figure~\ref{fig:chopfm}.

\begin{figure}[htb]
\centering
\psfrag{1}[c]{\small $_{\brak{\alpha}}$}
\psfrag{2}[c]{\small $_{\brak{\beta}}$}
\psfrag{g}[c]{\small $\mathrm{Out}_1$}
\psfrag{h}[c]{\small $\mathrm{Out}_2$}
\includegraphics[width=0.3\textwidth]{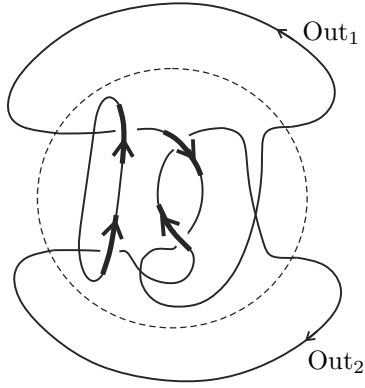}
\caption{\small Hopf--Chernoff configuration (closed version).}
\label{fig:chopfm}
\end{figure}

\section{PCP networks}
\label{sec:pcp}

In this section we specialize to non-adaptive $3$--bit verifiers, to exhibit that tangle machines may exhibit better performance parameters than classical systems.

We deform the H{\aa}stad PCP verifier, which has $c=1$ and $s\approx 0.75$. Suppose the two verifiers disagree, $v_t \ne w_t$, where $v_t,w_t \in \{\brak{\mathrm{True}},\brak{\mathrm{False}}\}$. In this case the interaction proceeds as follows. The patient verifier $W$ provides the addresses of three bits to $\Pi$. These bits are received by the agent verifier $V$ which computes a certificate $\xi_t$, which is equal either to $1$ meaning `accept' or to $-1$ meaning `reject'. The agent then flips one out of the three bits with the following probabilities:
\begin{itemize}
\item $(\xi_t=1)$ $\wedge$ $(v_t=\brak{\mathrm{True}})$ $\longrightarrow$ $V$ flips a bit with probability $1-\dep$.
\item $(\xi_t=1)$ $\wedge$ $(v_t=\brak{\mathrm{False}})$ $\longrightarrow$ $V$ flips a bit with probability $\dep$.
\item $(\xi_t=-1)$ $\wedge$ $(v_t=\brak{\mathrm{True}})$ $\longrightarrow$ $V$ flips no bit.
\item $(\xi_t=-1)$ $\wedge$ $(v_t=\brak{\mathrm{False}})$ $\longrightarrow$ $V$ flips a bit.
\end{itemize}
The corrupted set of bits is then sent back to $W$ who computes her own certificate. This protocol realizes the communication
channel in Figure~\ref{fig:channel}.

\subsection{A better-than-classical PCP verifier}\label{SS:PCPHC}

In this section, we construct a tangle machine whose interactions are deformed H{\aa}stad verifiers, which has perfect accuracy and a completeness of $\frac{2}{3}+\sigma\approx 0.667$. This is worse than the conjectured bound of $0.625$, but better than the best-known classical non-adaptive $3$--bit $\mathrm{PCP}$ protocol, whose soundness is only around $0.741$. Thus, our machine behaves like a single very good verifier.

Our machine makes use of the Hopf--Chernoff configuration in Figure~\ref{fig:hopfm}.

\begin{enumerate}
\item Choose input beliefs $\mathrm{In}_j^0$, $j=1,2$, arbitrarily from $\{\brak{\mathrm{True}},\brak{\mathrm{False}}\}$. Assume $c=1$ and fix $\dep$.
\item Let $\alpha$ be a Bernoulli distribution with parameter $\frac{1}{2}$. Draw a belief from $\alpha$ for the top agent, and set the bottom agent to the negation of that belief. We colour the bottom agent as $\neg \alpha$, which equals $\alpha$ as a distribution, to diagrammatically signify what we are doing.
\item Do the following for $i=1,\ldots, \chi$, where $\chi=\mathcal{O}(1)$:
\begin{itemize}
\item Using the underlying deformed $\mathrm{PCP}$ verifier, perform the four interactions of the Hopf--Chernoff configuration to propagate the beliefs of the two verifiers from $\mathrm{In}_j^i$ to $\mathrm{Out}_j^i$ according to the diagram below.
\[
\psfrag{u}[c]{\small $\mathrm{In}_1^i$}
\psfrag{v}[c]{\small $\mathrm{In}_2^i$}
\psfrag{1}[c]{\small $\alpha$}
\psfrag{2}[c]{\small $\neg\alpha$}
\psfrag{g}[l]{\small $\mathrm{Out}_1^i$}
\psfrag{h}[l]{\small $\mathrm{Out}_2^i$}
\includegraphics[width=0.27\textwidth]{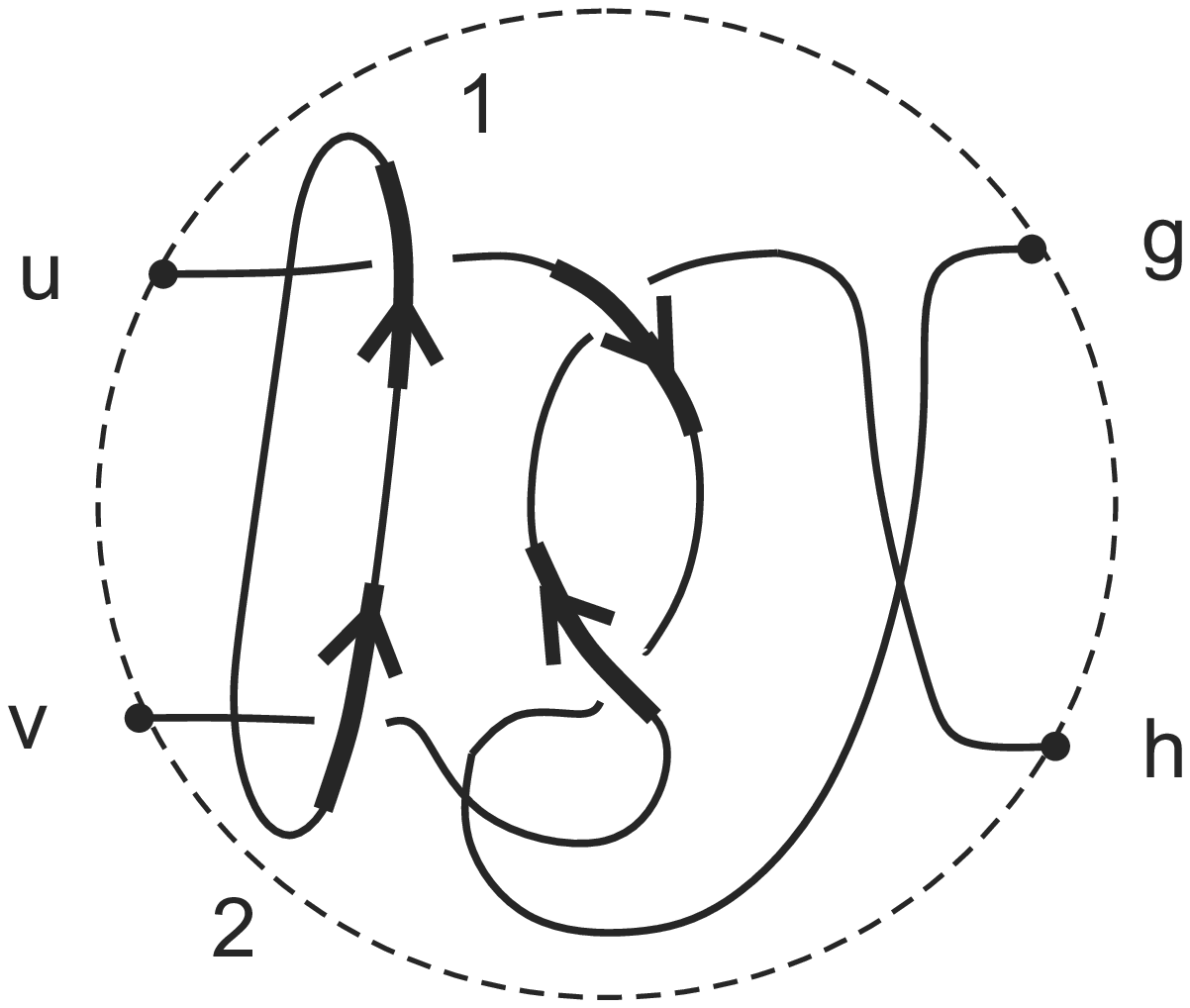}
\]

\item Set $\mathrm{In}_j^{i}=\mathrm{Out}_j^{i-1}$, $j=1,2$.
\end{itemize}
\item If the patient belief in the output $\mathrm{Out}_1^\chi = \neg \alpha$ then return $\brak{\mathrm{True}}$, otherwise return $\brak{\mathrm{False}}$.
\end{enumerate}

\begin{Theorem}
\label{thm:hopfmalg}
The Hopf--Chernoff machine approaches perfect completeness and soundness at most $\frac{1}{3-2s}$ as $\dep \to 1.$ In particular, if we put deformed H{\aa}stad verifiers at interactions, its soundness is at most $\frac{2}{3}$.
\end{Theorem}

\begin{proof}\hfill
\begin{description}
\item[Completeness] Assume that $x \in L$ and note that as $\dep \to 1$, so does $h=c\dep \to 1$. In the limit where $h=1$ note that by \eqref{eq:gg} the Hopf--Chernoff swaps the beliefs $\alpha$ and $\neg \alpha$ such that always $\mathrm{Out}_1=\neg\alpha$ and $\mathrm{Out}_2=\alpha$. Thus step (iv) of the algorithm concludes with $\brak{\mathrm{True}}$.

\item[Soundness] The soundness of the algorithm bounds the probability that $\mathrm{Out}_1 = \neg \alpha$ in case where $x \notin L$. This probability is
given by
\begin{multline}
\label{eq:pr1}
\Pr(\mathrm{Out}_1 = \neg \alpha) = \Pr(\mathrm{Out}_1 = \brak{\mathrm{True}} \mid \alpha = \brak{\mathrm{False}})\Pr( \alpha = \brak{\mathrm{False}})
\\
+ \Pr(\mathrm{Out}_1 = \brak{\mathrm{False}} \mid \alpha = \brak{\mathrm{True}})\Pr( \alpha = \brak{\mathrm{True}}).
\end{multline}
Although not truly essential, the algorithm assumes $\Pr( \alpha = \brak{\mathrm{True}}) = \frac{1}{2}$. The conditional probabilities above can be bounded using \eqref{eq:gg} as follows. Take $h=s\dep \to s$ and assume that $\brak{\alpha}=\brak{\mathrm{False}}$. In this case \eqref{eq:gg} implies:
\begin{equation}
\label{eq:pr2}
\Pr(\mathrm{Out}_1 = \brak{\mathrm{True}} \mid \alpha = \brak{\mathrm{False}}) \le \frac{1}{3-2s} \enspace.
\end{equation}
On the other hand, letting $\brak{\alpha}=\brak{\mathrm{True}}$, the same equation reads:
\begin{equation}
\label{eq:pr3}
\Pr(\mathrm{Out}_1 = \brak{\mathrm{False}} \mid \alpha = \brak{\mathrm{True}}) \le \frac{1}{3-2s}\enspace .
\end{equation}
This follows from the fact that the algorithm decides $x \in L$ if $\mathrm{Out}_1 = \neg \alpha$ irrespective of the beliefs themselves.
For that reason these equations coincide, though for different values of $\alpha$ and $\mathrm{Out}_1$. Both describe a failure of the algorithm to decide $x \notin L$. The theorem now follows from \eqref{eq:pr1}, \eqref{eq:pr2} and \eqref{eq:pr3}.
\end{description}
\end{proof}

\section{Low-dimensional topology and bisimulation}
\label{sec:lowdim}

\begin{Definition}\label{D:Bisimulation}
A tangle machines $M$ and $M^\prime$ which each come equipped with a distinguished set of input and output registers are \emph{bisimilar} if any computation that can be carried out on $M$ can be carried out on $M^\prime$ and vice versa.
\end{Definition}

Because computations are defined only with respect to the pre-chosen sets of input and of output registers, Definition~\ref{D:Bisimulation} encapsulates what may be thought of as a \emph{weak} notion of bisimulation (no requirement is made on `silent' or `internal' interactions).

In Section~\ref{SS:Equivalence} we formulate a set of local moves, such that any two machines related by these local moves are bisimilar. In Section~\ref{SS:reversibility} we discuss a feature of our formalism, that is the `unsplittability' of our agent registers. In Section~\ref{SS:ZeroKnowledge} we suggest an application of machine equivalence to define a notion of \emph{zero knowledge} for tangle machines and to utilize it to construct $\mathrm{TangIP}$ machines which are `more secure' in a specific sense. We give an example in Section~\ref{SS:EquivExample}. Finally, in Section~\ref{SS:EquivWye}, we extend our notion of machine equivalence to machines which may have wyes.

\subsection{Equivalence}\label{SS:Equivalence}
The key property of tangle machines is that they admit a local notion of equivalence \citep{CarmiMoskovich:15}. Two (quandle coloured, without wyes) tangle machines are \emph{equivalent} if they are related by a finite sequence of the moves in Figures~\ref{F:local_moves_machines} and~\ref{F:local_moves_machines1}. It is forbidden for these moves to involve input and output registers of a computation.

\begin{figure}[htb]
\centering
\psfrag{a}[c]{\small $\trr$}
\psfrag{b}[c]{\small $\rrt$}
\psfrag{V}[c]{\small \emph{I1}}
\psfrag{x}[c]{\small $x$}
\psfrag{T}[c]{\small \emph{VR1}}\psfrag{R}[c]{\small \emph{VR2}}\psfrag{S}[c]{\small \emph{VR3}}
\psfrag{Q}[c]{\small \emph{SV}}\psfrag{D}[c]{\small \emph{I2}}\psfrag{E}[c]{\small \emph{FM1}}\psfrag{F}[c]{\small \emph{FM2}}\psfrag{C}[c]{\small \emph{I3}}
\psfrag{Y}[c]{\small \emph{ST}}
\psfrag{X}[c]{\small \emph{ST}}
\includegraphics[width=0.85\textwidth]{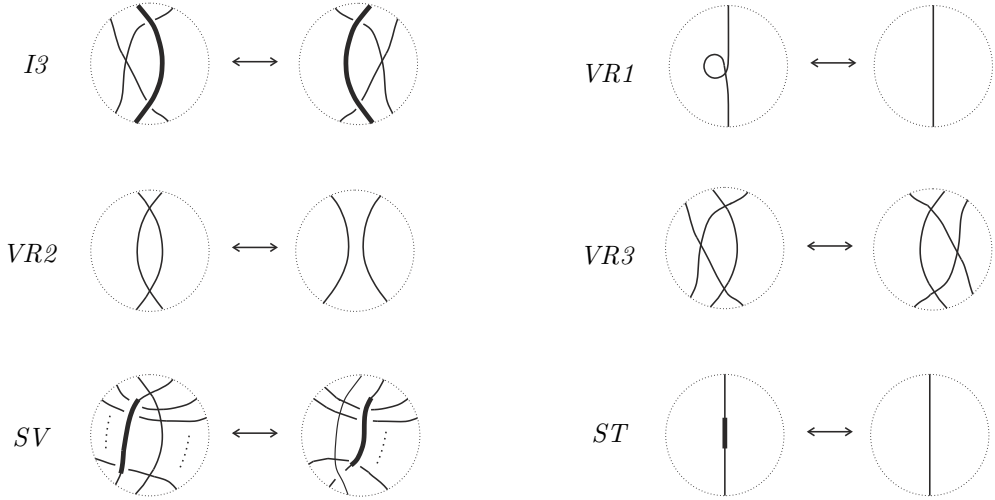}
\caption{\label{F:local_moves_machines} \small Cosmetic moves for machines. Directions are not indicated, meaning that the moves are valid for any directions, and the same for colourings.}
\end{figure}

\begin{figure}[htb]
\centering
\psfrag{T}[c]{\small \emph{VR1}}\psfrag{R}[c]{\small \emph{VR2}}\psfrag{S}[c]{\small \emph{VR3}}
\psfrag{Q}[c]{\small \emph{SV}}\psfrag{D}[c]{\small \emph{R1}}\psfrag{A}[c]{\small \emph{R2}}\psfrag{B}[c]{\small \emph{R3}}\psfrag{C}[c]{\small \emph{UC}}
\psfrag{X}[c]{\small \emph{ST}}
\includegraphics[width=0.85\textwidth]{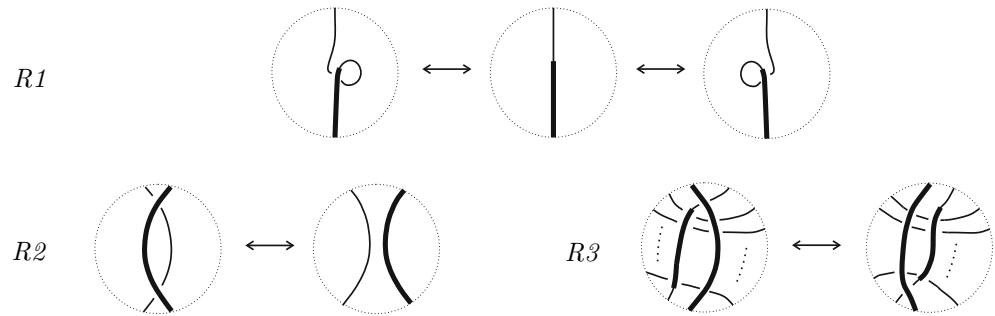}
\caption{\label{F:local_moves_machines1} \small Reidemeister moves for machines, valid for any directions of the agents. It is forbidden for these moves to involve input and output registers of a computation.}
\end{figure}

Two machines related by the local moves in Figure~\ref{F:local_moves_machines} carry out identical computations, and it follows from the construction of an interaction that two machines related by $\mathrm{R1}$ differ only by a trivial computation at which `nothing happens'. Thus, the interesting moves for us are $\mathrm{R2}$ and $\mathrm{R3}$ in Figure~\ref{F:local_moves_machines1}, which we will say more about later on.

\begin{Remark}
First note that, for $\mathrm{R2}$ to make sense, all participating colours must be defined. This requirement is non-trivial for a machine coloured by a quandloid.
\end{Remark}

\begin{Remark}
For a machine coloured by a quagma, the $\mathrm{R3}$ move is replaced by the following

\begin{equation}
\psfrag{b}[c]{\small $_\brr$}
\psfrag{a}[c]{\small $\trr$}
\psfrag{c}[c]{\small R3}
\includegraphics[width=0.4\textwidth]{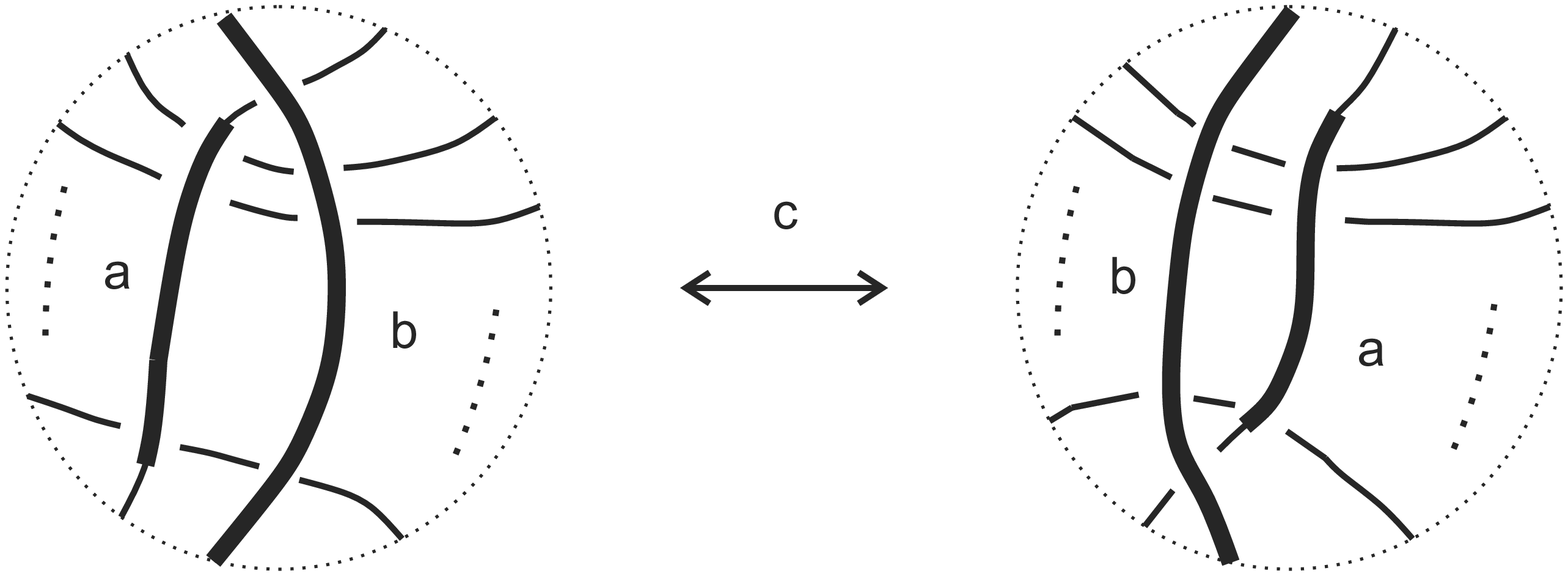}
\end{equation}

for all $\trr,\brr\in B$ satisfying $(x \trr y)\brr z= (x\brr z)\trr (y\brr z)$ for all $x,y,z\in Q$.
\end{Remark}

If we choose input and output registers to be machine endpoints, then equivalent machines have identical initial and terminal beliefs which implies that both machines have the same computational power in terms of deciding a language. Nevertheless, the local behaviour of equivalent may be different, in that the colours of intermediate interactions in between the same initial and terminal statistics may be different in equivalent machines. As in the earlier example in Figure~\ref{fig:pcpm}, equivalent machines may two different prover strategies arriving at the same proof.

To expand that example, Figure~\ref{F:pm_rmoves} features several equivalent prover strategies for the machine in Figure~\ref{fig:pm} all which are obtained by application of $\mathrm{R3}$ moves.

\begin{figure}
\centering
\psfrag{a}[c]{\small $_{_{\brak{W_0}}}$}
\psfrag{b}[c]{\small $_{_{\brak{V_0^1}}}$}
\psfrag{c}[c]{\small $_{_{\brak{V_0^2}}}$}
\psfrag{d}[c]{\small $_{_{\brak{V_0^3}}}$}
\psfrag{e}[c]{\small $_{_{\brak{V_0^4}}}$}
\psfrag{g}[c]{\small $_{_{\brak{W_4}}}$}
\includegraphics[width=0.8\textwidth]{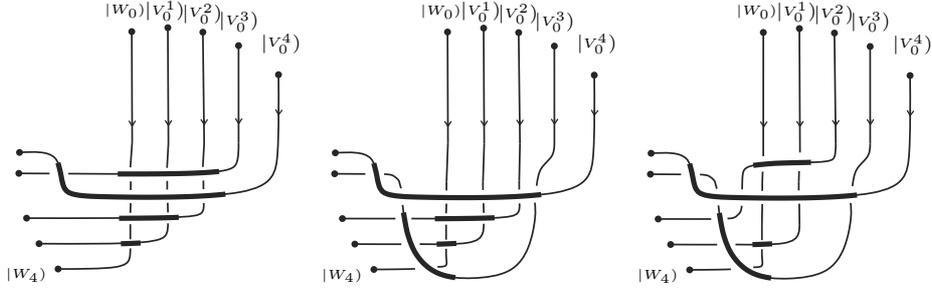}
\caption{\label{F:pm_rmoves} Equivalent prover strategies for the machine in Figure~{\ref{fig:pm}}.}
\end{figure}

\subsection{The single agent in R2 and R3}\label{SS:reversibility}

In this section we discuss the single agent which acts on numerous patients and cannot be split. Such an agent features in Moves $\mathrm{R2}$ and $\mathrm{R3}$, and distinguishes our approach \textit{e.g.} from w--tangles \citep{BarNatanDancso:13}.

The $\mathrm{R2}$ move tells us that computations are reversible, in the sense that any operation $\trr \in B$ has an inverse operation $\rrt \in B$ such that no information is computed from $(x\trr y)\rrt y$ for any $x,y\in Q$. Because we are working not only with colours but with realizations of belief statistics, we are saying more than just $(x\trr y)\rrt y=x$. We require that there be zero knowledge gain about realizations of $(x\trr y)\rrt y$ from a realization of $x$.

\begin{equation}
\label{eq:r2}
\psfrag{a}[c]{\small $y$}
\psfrag{c}[c]{\small $y$}
\psfrag{e}[c]{\small $y$}
\psfrag{h}[c]{\small $y$}
\psfrag{b}[c]{\small $x$}
\psfrag{d}[c]{\small ${x}$}
\psfrag{f}[c]{\small $x$}
\psfrag{I}[l]{\small ${{(x \trr y) \rrt y}}$}
\includegraphics[width=0.38\textwidth]{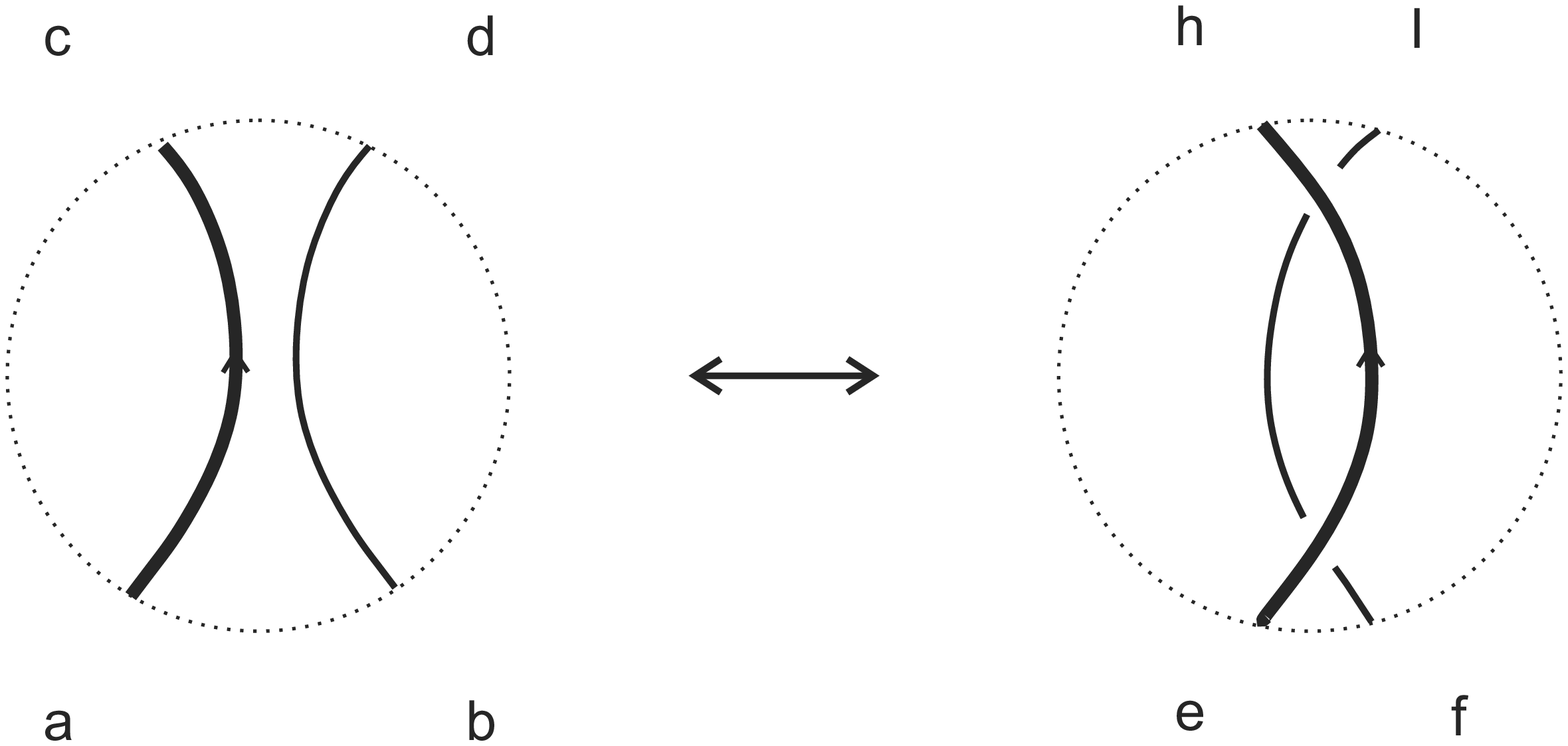}
\end{equation}

Our formalism features agents that act on multiple patients. These actions are independent by definition. Conversely, as we saw in Section~\ref{sec:hopf}, different agents may cooperate, for example by coordinating their realizations to be the same or to be opposed to one another. 


We do not impose the following a-priori reasonable generalization of $\mathrm{R2}$.

\begin{equation}\label{E:void_r2}
\includegraphics[width=0.38\textwidth]{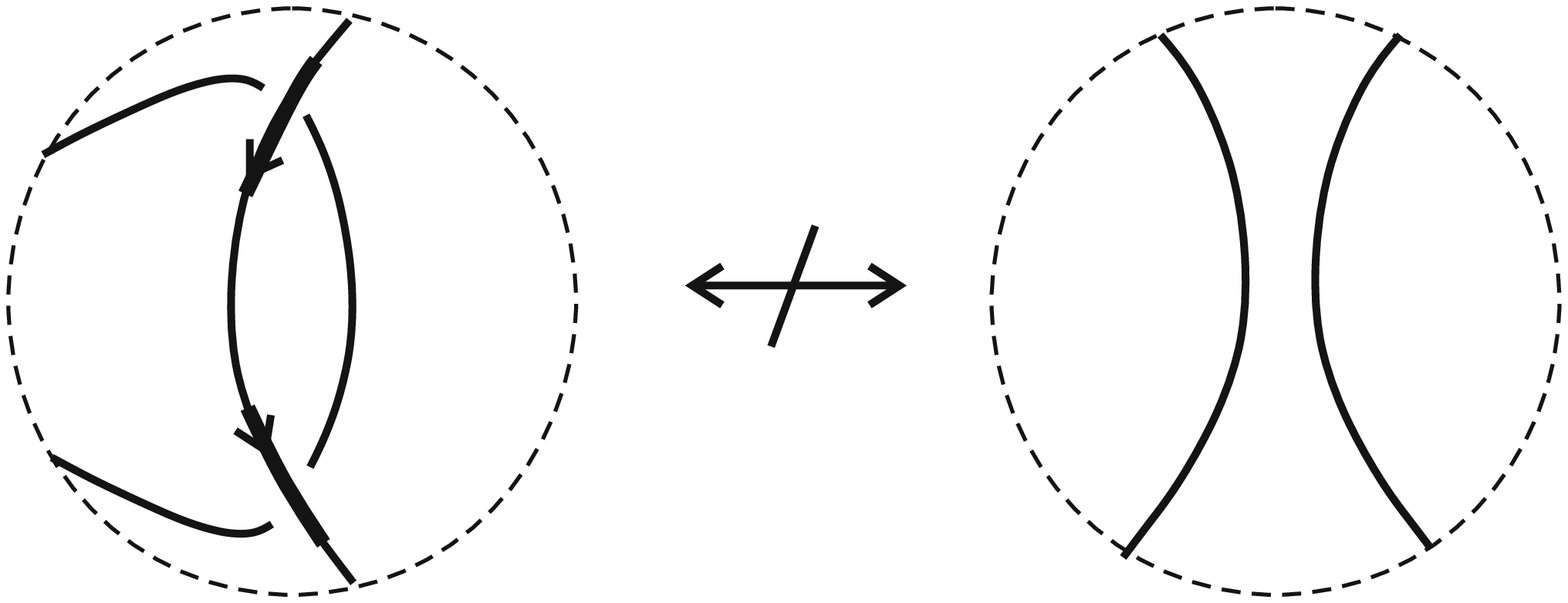}
\end{equation}

One reason that we do not impose \eqref{E:void_r2} can be seen by considering the example in which the two agents in the machine on the left-hand side adopt the strategy of always offering the same realization. If the realization of the patient and of the agent coincide, we can compute the realizations of both other patients. If not then we cannot. Thus we can compute the colours of the remaining patients in \eqref{E:void_r2} for some realizations but not for others. This behaviour is not shared by the machine on the right hand side, in which colours the realizations of patients can never be computed unless they are already given. If the choice of realization is independent for both patients, \textit{i.e.} if there is only a single interaction, as in the case of `honest' $\mathrm{R2}$, there is no such phenomenon, and no choice of realizations for input registers is distinguished from any other. Note also that \eqref{E:void_r2} represents two distinct computations, each of which can be considered separately and each of which is non-trivial, which is not true for the right-hand side of the `honest' $\mathrm{R2}$.

For the same reason, we do not impose the following `fake $\mathrm{R3}$ move':

\begin{equation}\label{E:void_r3}
\includegraphics[width=0.38\textwidth]{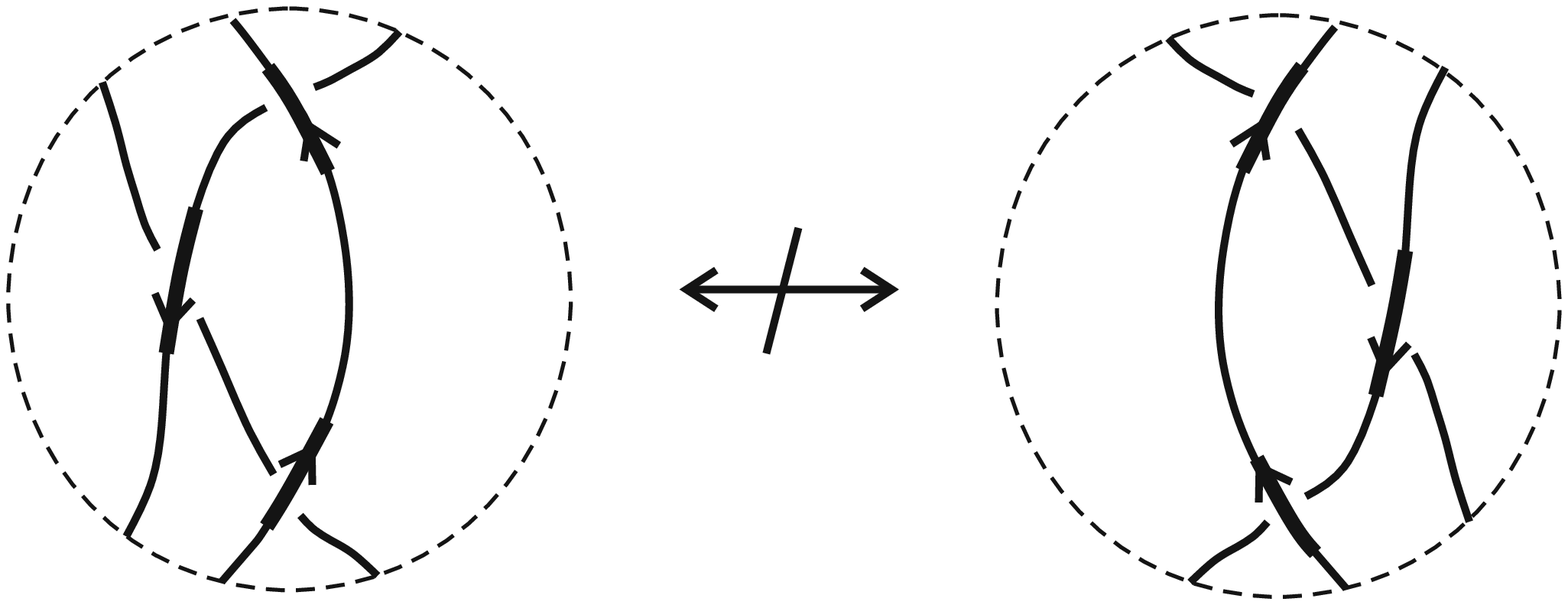}
\end{equation}

\subsection{Zero knowledge}\label{SS:ZeroKnowledge}

The theory of $\mathrm{IP}$ features the notion of a zero-knowledge proof \citep{Goldwasser:89,zkp:2010}. In a zero-knowledge proof, the information that may be gained by the verifier in the course of her interactions with the prover are restricted. This is useful when the verifier may not always be trustworthy. The definition makes use of a simulator which is an arbitrary feasible algorithm that is able to reproduce the transcript of such an interaction without ever interacting with the prover.

We suggest the following definition as a $\mathrm{TangIP}$ analogue to the notion of zero knowledge.

\begin{Definition}[Zero knowledge tangle machine]\label{D:ZeroKnowledge} 
A tangle machine $M$ that decides a language $L\in \mathrm{TangIP}$ is said to be \emph{zero knowledge} if the following is satisfied.
\begin{enumerate}
\item There are no intermediate interactions in $M$ that decides $L$.
\item There exists an equivalent machine $M^\prime$ which decides $L$ at one of its intermediate interactions.
\end{enumerate}
\end{Definition}

\begin{Remark}
The idea of zero knowledge tangled $\mathrm{IP}$ parallels the authors' model of fault-tolerant information fusion networks, except that there we wanted intermediate registers to `know as much as possible' whereas here we want them to `know as little as possible' \citep{CarmiMoskovich:14}.
\end{Remark}

As a generic example, consider machines $M^\prime$ and  $M$ in Figure~\ref{fig:zero}. Both share the same initial and terminal belief statistics. The explicit structure of the machines is mostly irrelevant except that they both contain a submachine $S$ which we graphically represent by a blank disk, with the property that $M$, $M^\prime$, and some of $S$'s terminal statistics $\brak{Y_i}$ decide $L$. In $M$, the verifier $Z$ is an agent to all initial states of $S$ and the resulting beliefs from this interaction are $\brak{X_i}$, $i=1,2,\ldots,$. In $M^\prime$ the same verifier is an agent to the terminal states of $S$ and the resulting beliefs from this interaction are characterized by $\brak{Y_i}$.

That $M$ is zero-knowledge implies the proof should not appear somewhere within it. Assuming none of the $\mathrm{In}'s$ decide $L$, and neither do any intermediate belief states of $S$, this requirement implies that none of the $\brak{X_i}=(\mathrm{In}_i)^{\brak{Z}}$ decide $L$. On the other hand, the machine $M^\prime$ shows us that an interaction between the terminal states of $S$, some of which decide $L$, with the agent $Z$ are able to produce the proof. That is, some of $\mathrm{Out}_i = \brak{Y_i}^{\brak{Z}}$ decide $L$. These requirements completely characterize the belief distribution $\brak{Z}$. Perhaps unsurprisingly, it turns out that $\brak{Z} = \frac{1}{2} \brak{\mathrm{True}} + \frac{1}{2} \brak{\mathrm{False}}$.

Ideally we would require that $S$ reproduces the proof as if it was produced by $M$ itself. This requirement translates into
\begin{equation}
\label{eq:cd}
\brak{Y_i} = \mathrm{Out}_i,
\end{equation}
for any $\mathrm{Out}_i$ that decides $L$. As both sides of \eqref{eq:cd} depend on the deformation parameter $\dep$, this equation may be
used to determine $\dep$ such that $M$ is zero-knowledge. Below we give an example.

\begin{figure}[htb]
\centering
\psfrag{m}[c]{\small $M$}
\psfrag{n}[c]{\small $M^\prime$}
\psfrag{a}[c]{\small $\mathrm{In}_i$}
\psfrag{b}[c]{\small $\mathrm{Out}_i$}
\psfrag{s}[c]{\small $S$}
\psfrag{c}[c]{\small $_{\brak{X_i}}$}
\psfrag{d}[c]{\small $_{\brak{Z}}$}
\psfrag{e}[c]{\small $_{\brak{Y_i}}$}
\includegraphics[width=0.8\textwidth]{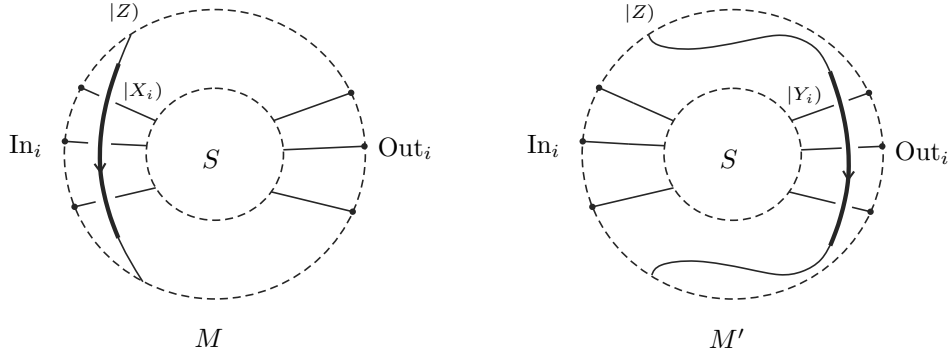}
\caption{\small Two equivalent machines. The machine $M$ is zero-knowledge. The proper submachine $S$ determines $L$.}
\label{fig:zero}
\end{figure}

\subsection{Example}\label{SS:EquivExample}

Consider the two machines in Figure~\ref{fig:zkm}. Assume they both employ a deformed IP system whose completeness and soundness are $\dep$ and $\frac{1}{2} \dep$. Let $\brak{Z}=\frac{1}{2}\brak{\mathrm{False}}+\frac{1}{2}\brak{\mathrm{True}}$ and let us first see what the value of $\dep$ is for the machine to decide $L$. The output $\brak{X_2}$ of either machines is given by
\begin{equation}
\brak{X_2} \longrightarrow \left \{ \begin{array}{ll}
\left[(1-\dep)^2+\frac{1}{2}\dep\right] \brak{\mathrm{False}} + \left[\dep(1-\dep)+\frac{1}{2}\dep\right] \brak{\mathrm{True}}, & x \in L; \\[0.5ex]
\left[(1-\frac{1}{2}\dep)^2+\frac{1}{4}\dep\right] \brak{\mathrm{False}} + \left[\frac{1}{2}\dep(1-\frac{1}{2}\dep)+\frac{1}{4}\dep\right] \brak{\mathrm{True}}, & x \notin L.
\end{array} \right.
\end{equation}
from which we conclude that $\dep > \frac{1}{2}$. The machine on the right in this figure is zero-knowledge because $\brak{X_1}$ does
not decide $L$:
\begin{equation}
\brak{X_1} \longrightarrow \left \{ \begin{array}{ll}
\left[1-\frac{1}{2}\dep\right] \brak{\mathrm{False}} + \frac{1}{2}\dep \brak{\mathrm{True}}, & x \in L ;\\[0.5ex]
\left[1-\frac{1}{4}\dep\right] \brak{\mathrm{False}} + \frac{1}{4}\dep \brak{\mathrm{True}}, & x \notin L.
\end{array} \right.
\end{equation}
and on the other hand the submachine inside the small disk on the left decides $L$:
\begin{equation}
\brak{\bar{X}_1} \longrightarrow \left \{ \begin{array}{ll}
\left[1-\dep\right] \brak{\mathrm{False}} + \dep \brak{\mathrm{True}}, & x \in L; \\[0.5ex]
\left[1-\frac{1}{2}\dep\right] \brak{\mathrm{False}} + \frac{1}{2}\dep \brak{\mathrm{True}}, & x \notin L.
\end{array} \right.
\end{equation}

\begin{figure}[htb]
\centering
\psfrag{c}[c]{\small $_{\brak{Z}}$}
\psfrag{b}[c]{\small $_{\brak{X_0} = \brak{\mathrm{False}}}$}
\psfrag{a}[c]{\small $_{\brak{Y_0} = \brak{\mathrm{True}}}$}
\psfrag{d}[c]{\small $_{\brak{Y_1}}$}
\psfrag{f}[c]{\small $_{\brak{X_2}}$}
\psfrag{e}[c]{\small $_{\brak{X_1}}$}
\psfrag{e1}[c]{\small $_{\brak{\bar{X}_1}}$}
\psfrag{t}[c]{\small \emph{time}}
\includegraphics[width=0.7\textwidth]{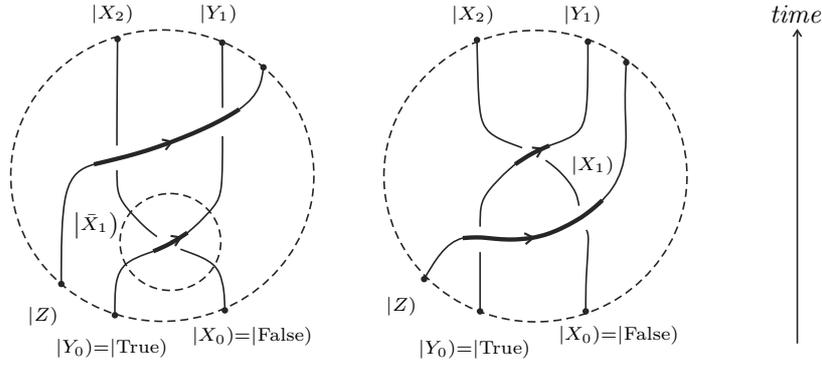}
\caption{\small Example of a zero-knowledge machine.}
\label{fig:zkm}
\end{figure}

If we further restrict the value of $\dep$ so as to satisfy \eqref{eq:cd}, namely,
\begin{equation}
\brak{\bar{X}_1} = \brak{X}_2 \longrightarrow \dep = \dep(1-\dep) + \frac{1}{2}\dep,
\end{equation}
we obtain $\dep = \frac{1}{2}$ which obviously contradicts the basic requirement of deciding $L$. If we slightly relax this
condition to allow a small discrepancy between the underlying distributions then we may take $\dep = \frac{1}{2} + \kappa(x)$
where $\kappa(x)$ is a statistical distance which potentially depends on $x$.

\subsection{Equivalence for machines with wyes}\label{SS:EquivWye}

Machines coloured by quagmas as machines with wyes also have a notion of equivalence, giving them a certain flexibility as a diagrammatic language. Two trivalent machines are \emph{equivalent} if they are related by a finite sequence of moves in Figure~\ref{F:local_moves_machines}, \ref{F:local_moves_machines1}, and \ref{F:WyeR}.

\begin{figure}
\centering
\psfrag{a}[c]{\emph{VYR1}}\psfrag{b}[c]{\emph{YR3}}\psfrag{c}[c]{\emph{VYR3}}
\includegraphics[width=0.5\textwidth]{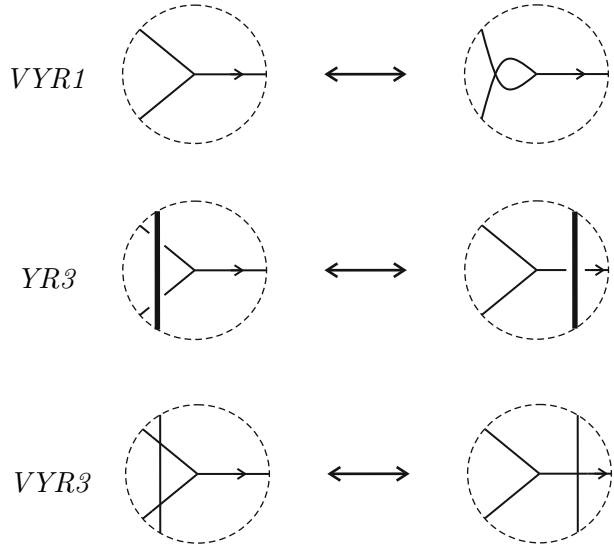}
\caption{\label{F:WyeR} Local moves for wyes. Note that $\mathrm{YR}3$ may reverse the label on the wye ($\max$ to $\min$ or vice versa), depending on what the colours are.}
\end{figure}

\section{Conclusions}

In this paper we have suggested a Turing-complete diagrammatic model of computation, in which computers are drawn as tangles of decorated coloured strings. With bounded resources, our `tangle machines' can decide any language in complexity $\mathrm{IP}$, sometimes more efficiently than known classical single--verifier models. Our machines admit a notion of equivalence that they inherit from low-dimensional topology, with equivalent machines representing bisimilar computations. Topological invariants of our machines would be characteristic quantities for these computations which are invariant over bisimilarity classes.



\begin{thebibliography}{99}
\providecommand{\natexlab}[1]{#1}
\expandafter\ifx\csname urlstyle\endcsname\relax

\bibitem[{Abramsky(2010)}]{Abramsky:10}
Abramsky, S. No-cloning in categorical quantum mechanics.
\newblock In \emph{Semantic Techniques in Quantum Computation}, (I. Mackie \& S. Gay Ed.), Cambridge University Press, 2010; pp. 1--28.
\newblock {\tt arXiv:0910.2401}

\bibitem[{Abramsky \& Coecke(2009)}]{AbramskyCoecke:09}
Abramsky S., \& Coecke, B. 2009 Categorical quantum mechanics.
\newblock In \emph{Handbook of Quantum Logic and Quantum Structures}, Vol. 2, 261--323.
\newblock {\tt arXiv:0808.1023}

\bibitem[{Alagic, Jeffery, \& Jordan(2014)}]{Alagic:14}
Alagic, G., Jeffery, S., \& Jordan, S. Circuit Obfuscation Using Braids.
\newblock In \emph{9th Conference on the Theory of Quantum Computation, Communication and Cryptography (TQC 2014)} (eds. S.T. Flammia \& A.W. Harrow), Vol. 27, pp. 141--160.
\newblock {\tt  arXiv:1212.6458}

\bibitem[{Arora \& Barak(2009)}]{AroraBarak:2009}
Arora, S., \& Barak, B. 2009 Computational complexity:\ \,  A modern approach.
\newblock Cambridge University Press.

\bibitem[{Arora \textit{et~al.}(1988)}]{Arora:98}
Arora, S., Lund, C., Motwani, R., Sudan, M., \& Szegedy, M. 1998 Proof verification and hardness of approximation problems.
 \newblock \emph{Journal of the ACM} \textbf{45}(3), 501--555.

\bibitem[{Arora \& Safra(1998)}]{AroraSafra:98}
Arora, S. \& Safra, S. 1998 Probabilistic checking of proofs: A new characterization of NP.
\newblock \emph{Journal of the ACM} \textbf{45}(1), 70--122.

\bibitem[{Babai, Fortnow, \& Lund(1991)}]{Babai:91}
Babai, L., Fortnow, L., \& Lund, C. 1991 Non-deterministic exponential time has two-prover interactive protocols.
\newblock \emph{Computational Complexity}, \textbf{1}, 3--40.

\bibitem[{Baez \& Stay(2011)}]{BaezStay:11}
Baez, J., \& Stay, M. 2011 Physics, topology, logic and computation: A Rosetta stone.
\newblock In \emph{New Structures for Physics}, Lecture Notes in Phys. \textbf{813}, 95--172.
\newblock {\tt arXiv:0903.0340}

\bibitem[{Bar-Natan \& Dancso(2013)}]{BarNatanDancso:13}
Bar-Natan, D., \& Dancso, S. 2013 Finite type invariants of w-knotted objects \textrm{I}: W-knots and the Alexander polynomial.
\newblock  Manuscript submitted for publication. \newblock {\tt arXiv:1405.1956}

\bibitem[{Barrington(1989)}]{Barrington:89}
Barrington, D.A. Bounded-width polynomial-size branching programs recognize exactly those languages in $NC^1$.
\newblock \emph{J. Comput. System Sci.} \textbf{1989}, \emph{38}(1), 150--164.

\bibitem[{Ben-Or \textit{et~al.}(1988)}]{BenOr:88}
Ben-Or, M., Goldwasser, S., Kilian, J., \& Wigderson, A. 1988 Multi prover interactive proofs: How to remove intractability assumptions.
\newblock In \emph{Proceedings of the 20th ACM Symposium on Theory of Computing}, 113--121.

\bibitem[{Buliga(2011a)}]{Buliga:11a}
Buliga, M. 2011 Braided spaces with dilations and sub-riemannian symmetric spaces.
\newblock In \emph{Geometry. Exploratory Workshop on Differential Geometry and
its Applications}, (D. Andrica \& S. Moroianu Ed.), Cluj-Napoca 21--35.
\newblock {\tt arXiv:1005.5031}

\bibitem[{Buliga(2011b)}]{Buliga:11b}
Buliga, M. 2011 Computing with space:\ \, A tangle formalism for chora and difference.
\newblock Preprint.
\newblock {\tt arXiv:1103.6007}

\bibitem[{Buliga \& Kauffman(2013)}]{BuligaKauffman:13}
Buliga, M., \& Kauffman, L. 2013 GLC actors, artificial chemical connectomes, topological issues and knots.
\newblock In \emph{ALIFE 14: Proceedings of the Fourteenth International Conference on the Synthesis and Simulation of Living Systems}, 490--497.
\newblock {\tt arXiv:1312.4333}

\bibitem[{Carmi \& Moskovich(2014)}]{CarmiMoskovich:14}
Carmi, A.Y. \& Moskovich, D. 2014 Low dimensional topology of information fusion.
\newblock In \emph{BICT14: Proceedings of the 8th International Conference on Bio-inspired Information and Communications Technologies}, ACM/EAI, 251--258.
\newblock {\tt arXiv:1409.5505} %

\bibitem[{Carmi \& Moskovich(2015)}]{CarmiMoskovich:15}
Carmi, A.Y. \& Moskovich, D. 2015 Tangle machines.
\newblock \emph{Proc. R. Soc. A} \textbf{2015}, \emph{471}, 20150111.
\newblock {\tt arXiv:1408.2862}

\bibitem[{Churchill(1974)}]{Churchill:74}
Churchill, F.B. William Johannsen and the genotype concept
\newblock \emph{J. Hist. Biol.} \textbf{1974}, \emph{7}, 5--30.

\bibitem[{Clark, Morrison, \& Walker(2009)}]{ClarkMorrisonWalker:09}
Clark, D., Morrison, S. \& Walker, K. 2009 Fixing the functoriality of Khovanov homology.
\newblock \emph{Geom. Topol.} \textbf{13}(3), 1499--1582.

\bibitem[{Elhamdadi(2014)}]{Elhamdadi:14}
Elhamdadi, M. 2014 Distributivity in Quandles and Quasigroups.
\newblock In \emph{Algebra, Geometry and Mathematical Physics}, Springer Berlin Heidelberg, 325--340.
\newblock {\tt arXiv:1209.6518}

\bibitem[{Fredkin \& Toffoli(1982)}]{Fredkin:82}
Fredkin, E. \& and Toffoli, T. Conservative logic.
\newblock \emph{Int. J. Theor. Phys.} \textbf{1982}, \emph{21}(3/4),  219--253.

\bibitem[{Goldreich(2010)}]{zkp:2010}
Goldreich, O. 2010 A short tutorial of zero-knowledge. \newblock Unpublished. \newblock {\tt http://www.wisdom.weizmann.ac.il/~oded/zk-tut02.html}

\bibitem[{Goldwasser, Micali, \& Rackoff(1989)}]{Goldwasser:89}
Goldwasser, S., Micali, S., \& Rackoff, C. 1989 The Knowledge complexity of interactive proof-systems.
\newblock \emph{SIAM Journal on Computing}, \textbf{18}(1), 186--208.


\bibitem[{H{\aa}stad(1997)}]{Has:97}
H{\aa}stad, J. 1997 Some optimal inapproximability results.
\newblock \emph{Journal of the ACM}, \textbf{48}(4), 798--859.

\bibitem[{Hopcroft, Motwani, \& Ullman(2001)}]{Hopcroft:01}
Hopcroft, J.E., Motwani, R., \& Ullman, J.D. 2001 \emph{Introduction to Automata Theory, Languages, and Computation}.
\newblock (2nd ed.), Reading Mass: Addison–Wesley.

\bibitem[{Ishii \emph{et~al.}(2013) Ishii, Iwakiri, Jang \& Oshiro}]{Ishii:13}
Ishii, A., Iwakiri, M., Jang, Y., \& Oshiro, K. 2013 A $G$--family of quandles and handlebody-knots.
\newblock \emph{Illinois J. Math.}, \textbf{57}, 817--838.
\newblock {\tt arXiv:1205.1855}

\bibitem[{Johannsen(1911)}]{Johannsen:11}
Johannsen, W. The genotype conception of heredity.
\newblock \emph{Am. Nat.} \textbf{1911}, \emph{45}(531), 129--159.

\bibitem[{Joyce(1982)}]{Joyce:82}
Joyce, D. 1982 A classifying invariant of knots:\ \, The knot quandle.
\newblock \emph{J. Pure Appl. Algebra} \textbf{23}, 37--65.

\bibitem[{Kauffman(1994)}]{Kauffman:94}
Kauffman, L.H. 1994 Knot automata.
\newblock In \emph{Twenty-Fourth International Symposium on Multiple-Valued Logic, Conference Proceedings}, 328--333.

\bibitem[{Kauffman(1995)}]{Kauffman:95}
Kauffman, L.H. 1995 Knot logic.
\newblock In \emph{Knots and Applications}, Series of Knots and Everything \textbf{6}, World Scientific Publications, 1--110.

\bibitem[{Kauffman(1999)}]{Kauffman:99}
Kauffman, L.H. 1999 Virtual knot theory.
\newblock \emph{Europ. J. Combinatorics} \textbf{20}(7), 663--690.
\newblock {\tt arXiv:math/9811028}

\bibitem[{Kauffman \& Lomonaco(2004)}]{KauffmanLomonaco:04}
Kauffman, L.H. \& Lomonaco Jr, S.J. 2004 Braiding operators are universal quantum gates.
\newblock \emph{New J. Phys.}, \textbf{6}(1), 134.
\newblock {\tt arXiv:quant-ph/0401090}

\bibitem[{Khot \& Saket(2006)}]{KhotSaket:06}
Khot, S., \& Saket, R. 2006 A 3-query non-adaptive PCP with perfect completeness.
\newblock In \emph{Proceedings of the 21st IEEE Conference on Computational Complexity}, 159--169.

\bibitem[{Kitaev(2003)}]{Kitaev:03}
Kitaev, A.Yu. Fault-tolerant quantum computation by anyons.
\newblock \emph{Ann. Phys.} (2003), \emph{303}, 2--30.
\newblock {\tt arXiv:quant-ph/9707021}

\bibitem[{Krohn, Maurer, \& Rhodes(1966)}]{Krohn:66}
Krohn, K., Maurer, W.D., \& Rhodes, J. Realizing complex boolean functions with simple groups.
\newblock \emph{Inform. Control} \textbf{1966}, \emph{9}(2), 190--195.

\bibitem[{Meredith \& Snyder(2010)}]{MeredithSnyder:10}
Meredith, L.G. \& Snyder, D.F. 2010 Knots as processes:\ \, A new kind of invariant.
\newblock Preprint.
\newblock {\tt arXiv:1009.2107}

\bibitem[{Monchon(2003)}]{Monchon:03}
Mochon, C. Anyons from nonsolvable finite groups are sufficient for universal quantum computation.
\newblock \emph{Phys. Rev. A} \textbf{2003}, \emph{67}(2), 022315.
\newblock {\tt arXiv:quant-ph/0206128}

\bibitem[{Moshkovitz \& Raz(2010)}]{MoshkovitzRaz:10}
Moshkovitz, D., \& Raz, R. 2010 Two-query PCP with subconstant error.
\newblock \emph{Journal of the ACM} \textbf{57}(5), 29.

\bibitem[{Nayak \emph{et~al.}(2008)Nayak, Simon, Stern, Freedman\& Sarma}]{Nayak:08}
Nayak, C., Simon, S.H., Stern, A., Freedman, M., \& Sarma, S.D. 2008 Non-Abelian anyons and topological quantum computation.
\newblock \emph{Rev. Mod. Phys.}, \textbf{80}(3), 1083--1159.
\newblock {\tt arXiv:0707.1889}

\bibitem[{Ogburn \& Preskill(1999)}]{OgburnPreskill:99}
Ogburn, R.W. \& Preskill, J. Topological quantum computation.
\newblock In \emph{Quantum Computing and Quantum Communications}, Lecture Notes in Comput. Sci. \emph{1509}, 1999; pp. 341--356. Springer.

\bibitem[{Peirce(1880)}]{Peirce:80}
Peirce, C.S., 1880 On the algebra of logic.
\newblock \emph{Amer. J. Math.}, \textbf{3} 15--57.

\bibitem[{Przytycki(2011)}]{Przytycki:11}
Przytycki, J.H. 2011 Distributivity versus associativity in the homology theory of algebraic structures.
\newblock\emph{Demonstr. Math.}, \textbf{44}(4), 823--869.
\newblock {\tt arXiv:1109.4850}

\bibitem[{Roscoe(1990)}]{Roscoe:90}
Roscoe, A.W. 1990 Consistency in distributed databases.
\newblock \emph{Oxford University Computing Laboratory Technical Monograph} \textbf{PRG-87}.


\bibitem[{Shamir(1992)}]{Shamir:92}
Shamir, A. 1992 $\mathrm{IP} = \mathrm{PSPACE}$.
\newblock \emph{Journal of the ACM}, \textbf{39}(4), 869--877.

\bibitem[{Surowiecki(2005)}]{Surowiecki:05}
Surowiecki, J. 2005 The wisdom of crowds:\ \, Why the many are smarter than the few and how collective wisdom shapes business, economies, societies, and nations.
\newblock Random House LLC.

\bibitem[{Turing(1937)}]{Turing:37}
Turing, A.M. 1937 \emph{On computable numbers, with an application to the Entscheidungsproblem.}
\newblock P. Lond. Math. Soc. Ser. 2 \textbf{42}, 230--265.
\newblock (and Turing, A.M. 1938 \emph{On computable numbers, with an application to the Entscheidungsproblem:\ \,A correction}. P. Lond. Math. Soc. Ser. 2 \textbf{43}, 544--546.

\bibitem[{Vicary(2012)}]{Vicary:12}
Vicary, J. 2012 Higher Semantics for Quantum Protocols.
\newblock In \emph{Proceedings of the 27th Annual ACM/IEEE Symposium on Logic in Computer Science}, 606--615.
\newblock {\tt arXiv:1207.4563}

\bibitem[{Zwick(1998)}]{Zwick:98}
Zwick, U. 1998 Approximation algorithms for constraint satisfaction problems involving at most
three variables per constraint.
\newblock In \emph{Proceedings of the 9th ACM-SIAM Symposium on Discrete Algorithms}, 201--210.

\end{thebibliography}
\end{document}